\documentclass[11pt,english]{article}
\pdfoutput=1
\usepackage[margin=1in]{geometry}
\usepackage{bbm}
\usepackage{graphicx,hyperref,color}
\usepackage{amsmath,amssymb,amsthm}
\usepackage{enumerate}
\usepackage{enumitem}
\usepackage{fullpage}
\usepackage{algorithm}
\usepackage{mathrsfs}
\usepackage[noend]{algcompatible}
\usepackage[noblocks]{authblk}
\usepackage{varwidth}
\usepackage{mathrsfs}
\usepackage{tcolorbox}
\usepackage{accents}
\usepackage{mathtools}
\usepackage{wrapfig}
\usepackage{makeidx}
\usepackage{thm-restate}

\newtheorem{question}{Question}

\def\denseformat{
\setlength{\textheight}{9in}
\setlength{\textwidth}{6.9in}
\setlength{\evensidemargin}{-0.2in}
\setlength{\oddsidemargin}{-0.2in}
\setlength{\headsep}{10pt}
\setlength{\topmargin}{-0.3in}
\setlength{\columnsep}{0.375in}
\setlength{\itemsep}{0pt}
}
\newtheorem{theorem}{Theorem}[section]
\newtheorem{proposition}[theorem]{Proposition}
\newtheorem{definition}[theorem]{Definition}
\newtheorem{claim}[theorem]{Claim}
\newtheorem{lemma}[theorem]{Lemma}

\newtheorem{corollary}[theorem]{Corollary}
\newtheorem{fact}[theorem]{Fact}

\newtheorem{observation}[theorem]{Observation}

\def\boldhead#1:{\par\vskip 7pt\noindent{\bf #1:}\hskip 10pt}
\def\ithead#1:{\par\vskip 7pt\noindent{\it #1:}\hskip 10pt}
\def\ceil#1{\lceil #1\rceil}

\def\inline#1:{\par\vskip 7pt\noindent{\bf #1:}\hskip 10pt}
\def\midinline#1:{\par\noindent{\bf #1:}\hskip 10pt}
\def\dnsinline#1:{\par\vskip -7pt\noindent{\bf #1:}\hskip 10pt}
\def\ddnsinline#1:{\newline{\bf #1:}\hskip 10pt}
\def\largeinline#1:{\par\vskip 7pt\noindent{\large\bf #1:}\hskip 10pt}
\long\def\commhide #1\commhideend{}
\long\def\commfull #1\commend{#1}
\long\def\commabs #1\commenda{}
\long\def\commtim #1\commendt{#1}
\long\def\commb #1\commbend{}
\long\def\commedit #1\commeditend{} % Editing comments, marked also by $>>>$

\long\def\commB #1\commBend{}       % Omit in 1996 (both TR and Siena)
\long\def\commex #1\commexend{}     % LN home exercise (hide solutions)

\long\def\commsiena #1\commsienaend{}  % omit in Siena, show in TR

\long\def\commBI #1\commBIend{}  % omit in Bar-Ilan

\long\def\CProof #1\CQED{}

\def\blackslug{\hbox{\hskip 1pt \vrule width 4pt height 8pt
    depth 1.5pt \hskip 1pt}}
\def\QED{\quad\blackslug\lower 8.5pt\null\par}
\def\inQED{\quad\quad\blackslug}

\long\def\PPP#1{\noindent{\bf Proof:}{ #1}{\quad\blackslug\lower 8.5pt\null}}

\long\def\denspar #1\densend
{#1}
\newif\ifnotesw\noteswtrue% T to show box & marginal notes; F supresses.
   {\ifnotesw\marginpar[\hfill\(\top\)]{\(\top\)}\fi}%
   {\ifnotesw\marginpar[\hfill\(\bot\)]{\(\bot\)}\fi}

\newcommand{\mnote}[1]%
    {\ifnotesw\marginpar%
        [{\scriptsize\it\begin{minipage}[t]{\marginparwidth}
        \raggedleft#1%
                        \end{minipage}}]%
        {\scriptsize\it\begin{minipage}[t]{\marginparwidth}
        \raggedright#1%
                        \end{minipage}}%
    \fi}

\def\MathF{\hbox{\rm I\kern-2pt F}}
\def\MathP{\hbox{\rm I\kern-2pt P}}
\def\MathR{\hbox{\rm I\kern-2pt R}}
\def\MathZ{\hbox{\sf Z\kern-4pt Z}}
\def\MathN{\hbox{\rm I\kern-2pt I\kern-3.1pt N}}
\def\MathC{\hbox{\rm \kern0.7pt\raise0.8pt\hbox{\footnotesize I}
\kern-4.2pt C}}
\def\MathQ{\hbox{\rm I\kern-6pt Q}}

\newsavebox{\ttop}\newsavebox{\bbot}

\newcommand{\emp}{\varnothing}
\def\eps{\epsilon}
\def\epsi{\varepsilon}

\newcommand{\mst}{\mathrm{MST}}

\newcommand{\mc}{\mathcal{C}}
\newcommand{\dm}{\mathrm{\textsc{Dm}}}
\newcommand{\adm}{\mathrm{\textsc{Adm}}}

\newcommand{\adj}{\mathtt{Adj}}
\newcommand{\ce}{c(\epsilon)}
\newcommand{\dk}{\Delta_{\mathcal{K}}}

\newcommand{\wmc[1]}{\widehat{\mathcal{#1}}}
\newcommand{\pzc[1]}{\mathpzc{#1}}
\newcommand{\dl}{\mathtt{L_{left}}}
\newcommand{\dr}{\mathtt{L_{right}}}
\newcommand{\pr}{\mathtt{proj}}
\newcommand{\ngc}{\mathtt{negCost}}
\newcommand{\be}{\mathtt{Bell}}

\newcommand{\mv}{\mathcal{V}}
\newcommand{\me}{\mathcal{E}}
\newcommand{\mg}{\mathcal{G}}
\newcommand{\un}{\mathsf{ucon}}
\newcommand{\cn}{\mathsf{con}}
\newcommand{\ddim}{\mathsf{ddim}}

\newcommand{\cluster}{\mathrm{\textsc{TreeClustering}}}
\def\eps{\epsilon}

\DeclareMathAlphabet{\mathpzc}{OT1}{pzc}{m}{it}

\newcommand{\gr}{\mathrm{grd}}
\newcommand{\child}{\mathtt{child}}
\newcommand{\cred}{\mathtt{cred}}

 \newcommand{\dwh}[1]{\widetilde{#1}}

\newcommand {\ignore} [1] {}

\denseformat

\title{Truly Optimal Euclidean Spanners}
\author{Hung Le}
\affil{University of Victoria and University of Massachusetts Amherst}
\author{Shay Solomon}
\affil{Tel Aviv University}
\date{}
\begin{document}
\maketitle

\maketitle
\begin{abstract}
Euclidean spanners are important geometric structures, having found numerous applications over the years. Cornerstone results in this area from the late 80s and early 90s state that for any $d$-dimensional $n$-point Euclidean space, there exists a $(1+\eps)$-spanner with $n \cdot O(\epsilon^{-d+1})$ edges and lightness (normalized weight) $O(\eps^{-2d})$.\footnote{The lightness of a spanner is the ratio of its weight and the MST weight.}
Surprisingly, the fundamental question of whether or not these dependencies on $\eps$ and $d$ for small $d$ can be improved  has remained elusive, even for $d = 2$.
This question naturally arises in any application of Euclidean spanners where precision is a necessity (thus $\eps$ is tiny).
In the most extreme case $\eps$ is inverse polynomial in $n$, and then one could potentially improve the size and lightness bounds by factors that are polynomial in $n$.

The state-of-the-art bounds $n \cdot O(\eps^{-d+1})$  and $O(\eps^{-2d})$ on the size and lightness of spanners are realized by the {\em greedy} spanner.
In 2016, Filtser and Solomon \cite{FS16} proved that, in low dimensional spaces, the greedy spanner is ``near-optimal''; informally, their result states that the greedy spanner for dimension $d$ is just as sparse and light as any other spanner {\em but for dimension larger by a constant factor}. Hence the question of whether the greedy spanner is truly optimal remained open to date.

The contribution of this paper is two-fold.
\begin{enumerate}
\item We resolve these longstanding questions by nailing down the dependencies on $\eps$ and $d$ and showing that the greedy spanner is truly optimal.
Specifically, for any $d= O(1), \eps = \Omega({n}^{-\frac{1}{d-1}})$:
\begin{itemize}
\item We show that there are $n$-point sets in $\mathbb{R}^d$ for which any $(1+\eps)$-spanner must have $n \cdot \Omega(\eps^{-d+1})$ edges,  implying that the greedy (and other)
spanners
achieve the optimal size.
\item We show that there are $n$-point sets in $\mathbb{R}^d$ for which any $(1+\eps)$-spanner must have lightness $\Omega(\eps^{-d})$,
and then improve the upper bound on the lightness of the greedy spanner from $O(\eps^{-2d})$   to $O(\eps^{-d} \log(\eps^{-1}))$. (The lightness upper and lower bounds match up to a lower-order term.)
\end{itemize}
\item We then complement our negative result for the size of spanners with a  rather counterintuitive positive result: Steiner points lead to a quadratic improvement in the size of spanners! Our bound for the size of Steiner spanners  in $\mathbb{R}^2$ is tight as well (up to a lower-order term).
\end{enumerate}

\end{abstract}

\section{Introduction}

\paragraph{1.1~ Background and motivation}

\paragraph{Sparse spanners.}
\vspace{-9pt}
Let $P$ be a set of $n$ points in $\mathbb R^d, d \ge 2$, and consider the complete weighted graph $G_P = (P,{P \choose 2})$ induced by $P$,
where the weight of any edge $(x,y) \in {P \choose 2}$ is  the Euclidean distance $|xy|$ between its endpoints.
Let $H = (P,E)$ be a spanning subgraph of $G_P$, with $E \subseteq {P \choose 2}$, where,  as in $G_P$, the weight function is given by the Euclidean distances. For any $t \ge 1$, $H$ is called a \emph{$t$-spanner} for  $P$ if for
every $x,y \in P$, the distance $d_G(x,y)$ between $x$ and $y$ in $G$ is at most $t |xy|$;
the parameter $t$ is called the {\em stretch} of the spanner and the most basic goal is to get it down to $1+\eps$, for arbitrarily small $\eps > 0$,
without using too many edges.
Euclidean spanners were introduced
in the pioneering SoCG'86 paper of Chew~\cite{Chew86}, who showed that $O(n)$ edges can be achieved with stretch $\sqrt{10}$, and later improved the stretch bound to 2~\cite{Chew89}.
The first Euclidean spanners with stretch $1+\eps$, for an arbitrarily small $\eps > 0$, were presented independently in the seminal works of
Clarkson~\cite{Clarkson87} (FOCS'87) and Keil~\cite{Keil88} (see also \cite{KG92}), which introduced the {\em $\Theta$-graph} in $\mathbb R^2$ and $\mathbb R^3$, and soon afterwards was generalized for any $\mathbb R^d$ in \cite{RS91,ADDJS93}.
The $\Theta$-graph is a natural variant of the {\em Yao graph}, introduced by Yao~\cite{Yao82} in 1982, where, roughly speaking, the space $\mathbb R^d$ around each point $p \in P$ is partitioned into cones of angle $\Theta$ each,
and then edges are added between each point $p \in P$ and its closest points in each of the cones centered around it.
The $\Theta$-graph is defined similarly, where, instead of connecting $p$ to its closest point in each cone, we connect it to a point
whose orthogonal projection to some fixed ray contained in the cone is closest to $p$. Taking $\Theta$ to be $c  \eps$, for   small enough constant $c$, one obtains a $(1+\eps)$-spanner with $O(n\epsilon^{-d+1})$ edges.

Euclidean spanners turned out to be a fundamental geometric construct, which evolved into an important research area \cite{Keil88,DN94,ADMSS95,RS98,GLN02,
AWY05,CG06,DES08,ES15},
with a plethora of applications, such as in geometric approximation algorithms \cite{RS98,GLNS02,GLNS02b,GLNS08}, geometric distance oracles
\cite{GLNS02,GLNS02b,GNS05,GLNS08}, network design \cite{HP00,MP00} and machine learning \cite{GKK17}. (See the book by Narasimhan and Smid \cite{NS07} for an excellent account on  Euclidean spanners and some of their applications.)

The tradeoff  between stretch $1+\eps$ and $O(n\epsilon^{-d+1})$ edges is the current state-of-the-art, and is also achieved by other spanner constructions,
including the {\em path-greedy} (abbreviated   as ``greedy'') spanner~\cite{ADDJS93,CDNS92,NS07}  and the
gap-greedy spanner~\cite{Salowe92,AS97}.
Surprisingly, despite the extensive body of work on Euclidean spanners since the  80s, the following fundamental question remained open, even for $d = 2$.
\begin{question} \label{qsize}
Is the tradeoff between stretch $1+\eps$ and $n \cdot O(\epsilon^{-d+1})$ edges tight?
\end{question}

We remark that the $\Theta$-graph and its variants provide stretch $1+\eps$ only for sufficiently small angle $\Theta$. These graphs have also been studied for fixed values of $\Theta$; see \cite{BDDOSSW12,BD13,BR14,BBDFKORTVX14,LZ16,JLZ18,BCHS19}, and the references therein.
The general goal here is to determine the best possible stretch for small values of $\Theta$.
E.g., it was shown in SODA'19 \cite{BCHS19} that the $\Theta$ graph with 4 cones, $\Theta_4$, has stretch $\le 17$. This line of work is somewhat orthogonal to Question~\ref{qsize}, which concerns the {\em asymptotic} behavior of the tradeoff.

\paragraph{Light spanners.} Another basic property of spanners,  important for various applications, is {\em lightness}, defined as the ratio of the spanner \emph{weight} (i.e., the sum of all edge weights in it) to the weight $w(\mst(P))$ of the minimum spanning tree $\mst(P)$ for
$P$. In SoCG'93, Das et al.~\cite{DHN93} showed that the aforementioned greedy spanner of \cite{ADDJS93} has constant lightness in $\mathbb R^3$,
which was generalized in SODA'95~\cite{DNS95} to $\mathbb R^d$ for any constant $d$;
the dependencies on $\eps$ and $d$ in the lightness bound were not explicated in \cite{ADDJS93,DHN93,DNS95}.
Later, in their seminal STOC'98 paper on approximating TSP in $\mathbb{R}^d$ using light spanners, Rao and Smith~\cite{RS98} showed that the greedy spanner has lightness $\epsilon^{-O(d)}$ in $\mathbb{R}^d$ for any constant $d$.
 In the open problems section of their paper~\cite{RS98}, they raised the question of determining the exact constant hiding in the $O$-notation $O(d)$ in the exponent of their upper bound.
 \footnote{In the full (unpublished) version of their paper, Rao and Smith remarked that in the Euclidean plane, a lightness bound of $O(\epsilon^{-2})$ is optimal by  pointing out that any $(1+\eps)$-spanner of a set of $\Theta(\frac{1}{\epsilon})$ points evenly placed on the boundary of a circle has lightness $\Omega(\frac{1}{\epsilon^2})$;
 this statement was not accompanied with a proof. In general, the full unpublished version of \cite{RS98} contains several claims on light spanners whose proofs are incomplete.}
All the proofs in \cite{ADDJS93,DHN93,DNS95,RS98} had many missing details.
The first complete proof was given in the book of \cite{NS07},  where a 60-page chapter was devoted to it, showing that the greedy $(1+\eps)$-spanner has lightness $O(\eps^{-2d})$.
In SODA'19, Borradaile, Le and Wulff-Nilsen~\cite{BLW19} presented a much shorter and arguably simpler alternative proof that, in fact, applies to the wider family of {\em doubling metrics} (see also \cite{Gottlieb15}),
but the lightness bound of $O(\eps^{-2d})$ remains the state-of-the-art.\footnote{The \emph{doubling dimension} of a metric space $(X,\delta)$ is the smallest value $\ddim$
such that every ball $B$ in the metric space can be covered by at most
$2^{\ddim}$ balls of half the radius of $B$.
This notion generalizes the Euclidean dimension, since the doubling dimension
of the Euclidean space $\mathbb R^d$ is $\Theta(d)$.
A metric space is called \emph{doubling} if its doubling dimension is constant.}
Therefore,  the following question remained open all these years, even for $d=2$.
\begin{question} \label{qlight}
Is the tradeoff between stretch $1+\eps$ and lightness $O(\epsilon^{-2d})$ tight?
\end{question}

\paragraph{Existential near-optimality}
In PODC'16, Filtser and Solomon \cite{FS16} studied the optimality of the greedy spanner in {\em doubling metrics},
which is wider than the family of low-dimensional Euclidean spaces.
They showed that the greedy spanner is \emph{existentially near-optimal} with respect to both the size and the lightness.
Roughly speaking, the greedy spanner is said to be existentially optimal
for a graph family $\mathcal G$ if its worst performance (in terms of size and/or lightness) over all graphs in $\mathcal G$
is just as good as the worst performance of an optimal spanner over all graphs in $\mathcal G$.
For doubling metrics, the loss encapsulated by the ``near-optimality'' guarantee comes into play with the dimension $\ddim$:
one compares the greedy $(1+\eps)$-spanner over metrics with doubling dimension $d$ with any other $(1+\eps)$-spanner, but over metrics with doubling dimension $2d$.
This loss in the dimension becomes more significant if we restrict the attention to Euclidean spaces,
as then the comparison is between Euclidean dimension $d$ and doubling dimension $2d$,
but spanners for metrics with doubling dimension $2d$ (or even $d$) tend to admit significantly weaker guarantees (as a function of $\eps$ and $d$)
than the corresponding ones for $d$-dimensional Euclidean spaces.

Consequently, this result by \cite{FS16} does not resolve Questions~\ref{qsize} and~\ref{qlight} for two reasons.
First, it only implies near-optimality of the greedy spanner, which, as mentioned, comes with a constant factor loss in the dimension,
and this constant factor slack appears in the exponents of the size and lightness bounds.
Second, and more importantly, even if we knew that the greedy spanner is truly optimal,
this still does not unveil the tight dependencies on $\eps$ and $d$.
In the current work we unveil the tight dependencies on $\eps$ and $d$
and as a corollary conclude that the greedy spanner is truly optimal. 

\paragraph{1.2~ Our contribution.~} Throughout we assume that $\epsilon \ll 1$.  We use $\tilde{O}_{\epsilon}$ and $\tilde{\Omega}_{\epsilon}$ to suppress poly-logarithmic factors of $ \frac{1}{\epsilon}$.
Our starting point is a surprisingly simple observation regarding evenly spaced point sets on the $d$-dimensional sphere,
using which we  prove:
\begin{theorem}\label{thm:lb-d} For any constant $d$  and any $n$ and $\eps$ such that $\eps = \Omega({n}^{-\frac{1}{d-1}})$,
there is a set $P$ of $n$ points in $\mathbb{R}^d$ such that any $(1+\epsilon)$-spanner for $P$ must have lightness $\Omega(\epsilon^{-d})$
and $n \cdot \Omega(\epsilon^{-d+1})$ edges.
\end{theorem}
Theorem \ref{thm:lb-d} immediately resolves Question \ref{qsize} in the affirmative, and it also shows that the greedy spanner is truly optimal with respect to the size parameter.

We then improve the lightness bound of the greedy spanner to match our lower bound.
\begin{theorem}\label{thm:ub-d}
The greedy $(1+\epsilon)$-spanner in $\mathbb{R}^d$ has lightness $\tilde{O}_\eps\left(\epsilon^{-d}\right)$.
\end{theorem}
Theorem \ref{thm:ub-d} answers Question \ref{qlight} in the negative,
and it also shows that the greedy spanner is truly optimal with respect to the lightness parameter. The exact upper bound on the lightness is $O\left(\epsilon^{-d}\log(\frac{1}{\epsilon})\right)$.
The proof of Theorem \ref{thm:ub-d} is intricate.

Our lightness analysis of the greedy algorithm builds on exciting developments on light spanners from recent years, which started from the works of Gottlieb~\cite{Gottlieb15}
and Chechik and Wulff-Nilsen~\cite{CW16} on non-greedy spanners. Using the result of \cite{FS16}, the framework of \cite{CW16} was refined in the works of Borradaile, Le and Wulff-Nilsen~\cite{BLW17,BLW19}.
As mentioned, it was shown in~\cite{BLW19} that the greedy spanner in metrics of doubling dimension $\ddim$ has lightness   $\epsilon^{-O(\ddim)}$.
We demonstrate that, by adapting the analysis in~\cite{BLW19} to Euclidean spaces and applying a few tweaks, one can obtain a lightness bound of $\epsilon^{-(d+2)}$.
To shave the remaining slack of $\epsilon^{-2}$ factor, we introduce several  highly nontrivial geometric insights to the analysis.

\vspace{-3pt}
\paragraph{Sparse Steiner spanners?}
{\em Steiner points} are additional points that are not part of the input point set.
A standard usage of Steiner points is for reducing the weight of the tree,
with the Steiner Minimum Tree (SMT) problem serving as a prime example:
In any metric, the Steiner ratio, which is the ratio of the SMT weight to the MST weight,
is at least $\frac{1}{2}$ (by the triangle inequality) and at most 1 (by definition).
In $\mathbb R^2$ the Steiner ratio is known to be between $\approx 0.824$ and $\frac{\sqrt{3}}{2} \approx 0.866$, and the famous (still open) ``Gilbert-Pollak Conjecture'' is that the upper bound $\frac{\sqrt{3}}{2}$ is tight \cite{GP68}.
As another example, a spanning tree that simultaneously approximates a shortest-path tree
and a minimum spanning tree is called a \emph{shallow-light tree} (shortly, SLT).
In FOCS'11, Elkin and Solomon~\cite{ES11} showed that in general metric spaces, Steiner points can be used to get an exponential improvement to the lightness of SLTs.
The construction of \cite{ES11} does not apply to Euclidean spaces,
but Solomon~\cite{Solomon14} showed that Steiner points can be used to get a quadratic improvement to the lightness of SLTs in $\mathbb R^d$ for $d= O(1)$.

Although these examples demonstrate that Steiner points could be very useful for reducing the weight of tree structures,
note that the resulting Steiner trees must contain more edges than the original trees by definition.
In other words, in trees, Steiner points cannot be used for reducing the number of edges by their very definition. 
Broadly speaking, it seems counterintuitive that Steiner points could be used as means for reducing the number of edges of graph structures such as spanners.
And indeed, essentially all the prior work in this context only support this intuition; in particular, Alth\"{o}fer et al.~\cite{ADDJS93} assert that, in general metrics,
Steiner points provably do not help (much) in reducing the spanner size (see Theorems 6-8 therein), and this result was strengthened in \cite{RR98} (see Theorem 1.2 therein).
We remark that these hardness results of \cite{ADDJS93,RR98} are based on girth arguments, and are not applicable in low-dimensional Euclidean spaces.

The size lower bound provided by Theorem \ref{thm:lb-d} implies that cornerstone spanner constructions from the 80s, such as the $\Theta$-graph and the greedy spanner, cannot be improved in size.
We contrast this negative message with a positive and counterintuitive one: Steiner points can be used to obtain a quadratic improvement on the size of spanners!
We'll focus on the Euclidean plane $\mathbb R^2$, but we get this quadratic improvement in any constant dimension $d \ge 2$:
$n \cdot \tilde{O}_{\epsilon}(\epsilon^{(-d+1)/2})$ edges using Steiner points versus $n \cdot \Omega(\eps^{-d+1})$ edges without using them.

 \begin{theorem}\label{thm:sparse-Steiner}
 For any set of $n$ points $P$ in $\mathbb{R}^2$, there is a Steiner $(1+\eps)$-spanner for $P$ with $\tilde{O}_\epsilon(\frac{n}{\sqrt{\epsilon}})$ edges. For a general constant $d$, there is a Steiner spanner with $\tilde{O}_\epsilon(\frac{n}{\epsilon^{(d-1)/2}})$ edges,
for any set of $n$ points $P$ in $\mathbb R^d$.
  \end{theorem}
  \noindent{\bf Remarks.} 
  The exact upper bound is  $O(\frac{n}{\sqrt{\epsilon}} \log^2\frac{1}{\epsilon})$
  in $\mathbb R^2$ and $O(\frac{n}{\epsilon^{(d-1)/2}}\log^2\frac{1}{\epsilon})$ in $\mathbb R^d$; we did not try to optimize $\log^2\frac{1}{\epsilon}$ factor.

The following lower bound shows that our construction of sparse Steiner spanners (Theorem~\ref{thm:sparse-Steiner}) is optimal to within polylogarithmic factor of $\frac{1}{\epsilon}$ for 2-dimensional Euclidean spaces.
 \begin{theorem}\label{thm:lb-Steiner-R^2}
 For any $n$ and $\eps$ such that $\eps = \tilde{\Omega}(\frac{1}{n^2})$, there exists a set of $n$ points in $\mathbb{R}^2$ such that any Steiner $(1+\epsilon)$-spanner must have at least $\tilde{\Omega}_\epsilon(\frac{n}{\sqrt{\epsilon}})$ edges and lightness at least $\tilde{\Omega}_\epsilon(\frac{1}{\epsilon})$.
 \end{theorem}
 \noindent{\bf Remark.} Since the SMT and MST weights are the same up to a small constant,  we can define the lightness of Steiner spanners with respect to the MST weight, just as with non-Steiner spanners.  The exact lower bound on the number of edges is $\Omega(\frac{n}{\sqrt{\epsilon  \log \frac{1}{\epsilon}}})$ and the exact lightness lower bound is $\Omega(\frac{1}{\epsilon \log \frac{1}{\epsilon}})$.

To prove our upper and lower bounds for Steiner spanners (Theorems \ref{thm:sparse-Steiner} and \ref{thm:lb-Steiner-R^2}),
we come up with novel geometric insights, which may be of independent interest, as discussed  in the next section.

\paragraph{1.3~ Proof overview, comparison with prior work, and technical highlights.~}
The starting point of this work is
in making a remarkably simple observation regarding a set of evenly spaced points along the boundary of a circle, which suffices for
getting the lower bound for the size and lightness of spanners in $\mathbb{R}^2$ (Theorem~\ref{thm:lb-d}).
Numerous papers have identified this point set as a natural candidate for lower bounds (see, e.g., \cite{RS98,ES11,Sol14}),
yet we are not aware of any paper that managed to rigourously prove such a result.
The $d$-dimensional analogue is a set of evenly spaced points along the sphere,
providing a set of $\Theta(\epsilon^{-d+1})$ points corresponding to the codewords of a spherical code in $\mathbb{R}^d$.
The distance between any two codewords is $\Omega(\epsilon)$, using which we show that for any two points $x,y$ with $|xy| = \Theta(1)$, any $(1+\eps)$-spanner must take $xy$ as an edge.  Since $P$ has $\Theta(\epsilon^{-2d+2})$ pairs of points of distance $\Theta(1)$, any spanner for $P$ must have weight $\Omega(\epsilon^{-2d+2})$
and
$\Omega(\epsilon^{-2d+2})$ edges. Noting that $w(\mst(P))= O(\epsilon|P|)$, the lightness bound of $\Omega(\epsilon^{-d})$ immediately follows. The size lower bound holds for a point set of size $\Theta(\epsilon^{-d+1})$; to extend it to an $n$-point set, we consider multiple copies of the same point set that lie sufficiently far from each other,
so that each point set must be handled with a separate vertex-disjoint spanner; see Section~\ref{sec:lb-spanners}.

We bypass the size lower bound by using Steiner points. For simplicity of presentation, we  mostly focus on the Euclidean plane $\mathbb{R}^2$, but the argument can be naturally generalized for $\mathbb{R}^d, d >2$.
We start with constructing Steiner spanners for point sets of bounded spread $\Delta$, with
$O(\frac{n}{\sqrt{\epsilon}}\log \Delta)$ edges.\footnote{The \emph{spread} of a point set $P$, denoted by $\Delta(P)$, is the ratio of the largest to the smallest pairwise distance.} We start our construction by partitioning pairs of points into $O(\log \Delta)$ sets where $i$-th set contains pairs of distance in $[2^{i-1},2^i)$, and focusing on preserving distances between pairs in each set separately. To this end, we divide the bounding box of the point set into \emph{overlapping subsquares} of side length $5\cdot 2^{i}$. The overlap allows us to treat each subsquare separately.  Each subsquare is then divided into horizontal bands and vertical bands. The  observation is that it suffices to construct a Steiner spanner for each pair of (non-adjacent) horizontal/vertical bands. We then show that the total number of edges of all Steiner spanners (to preserve distances between pairs in $i$-th set) is $O(\frac{n}{\sqrt{\epsilon}})$, which implies the desired bound  on  number of edges for all sets. We then show a reduction from a general point set (of possibly huge spread) to a point set of spread $O(\frac{1}{\epsilon})$.
This reduction builds on the standard net-tree spanner (see, e.g., \cite{GGN04,CGMZ16,GR082}) in a novel way, using a notion that we shall refer to as a \emph{ring spanner}.
A $t$-ring spanner of a point set, for $t \ge 1$, is a spanner that preserves (to within a factor of $t$) distances between every pair of points $p,q$ such that $q$ belongs to a ring (or annulus) around $p$.
The net-tree spanner is obtained, in fact, as a union of $\Theta(\log \Delta)$ ring 1-spanners, where the inner and outer radii of the annulus are within a factor of $1/\eps$.
Using this fact, we are able to reduce the problem of constructing a Steiner spanner for a general point set to the problem of constructing $\Theta(\log \Delta)$ ring $(1+\eps)$-spanners for
point sets of spread $O(\frac{1}{\epsilon})$ each, and show how to reduce it further to the construction of just one such ring $(1+\eps)$-spanner.
Our strategy of constructing Steiner spanners by building on the net-tree spanner is somewhat surprising, since all known (non-Steiner) net-tree spanners have $\Omega((\frac{1}{\epsilon})^2)$ edges, which exceeds the optimal bound  $O(\frac{1}{\epsilon})$ obtained by other spanners (such as the $\Theta$-graph) by a factor of $1/\eps$.
However,  by looking at the net-tree spanner through the lens of ring spanners and, of course, through the use of Steiner points, we are able to achieve the improved size bound;
 the details appear in Section~\ref{sec:ub-steiner-spanner}.

To prove the lower  bound on the size of Steiner spanners in $\mathbb{R}^2$, we can use the same point set used for our lower bounds for non-Steiner spanners,
of evenly spaced points along the boundary of a circle.
The argument here, however, is significantly more intricate.
It is technically more convenient to work with a similar point set $P$, where the points are evenly spaced along two opposite sides of a unit square $U$,
denoted by $N$ (``north'') and $S$ (``south'').
The distance between any two consecutive points along $N$ and along $S$ will be $\Theta(\sqrt{\epsilon \log \frac{1}{\epsilon}})$,
so that  $|P| = \Theta_{\epsilon}(\frac{1}{\sqrt{\epsilon}})$.
Our goal is to show that any Steiner spanner must use roughly $\Theta(|P|^2)$ edges to preserve the distances for all pairs of points from $N$ and $S$ to within a factor of $1+\eps$, and then taking multiple copies of the same point set that are sufficiently far from each other would complete the proof.
Instead of proving the size lower bound  directly, we show that any Steiner spanner for $P$ must incur a weight of $\Omega(|P|^2)$; the size lower bound would follow easily, as the distance between any pair of points in $P$ is $O(1)$.
We then demonstrate that the problem of lower bounding the spanner weight for $P$ boils down to the problem of determining the lengths of intersecting shortest paths in the spanner.
Next,  we say that the {\em intersecting pattern} of two shortest paths of two pairs of points is ``good'' if the total length of all intersecting subpaths between them is small; the smaller the intersection is, the ``better'' the pattern is. Determining the ``quality'' of intersecting patterns of arbitrary pairs of shortest paths is challenging.
To this end, define the {\em distance} between two pairs of points $\{x_1,x_2\}, \{y_1,y_2\}$, where $x_1,y_1 \in N$ and $x_2,y_2 \in S$, to be $\max\{|x_1y_1|, |x_2y_2|\}$, and denote it by $d(\{x_1,x_2\}, \{y_1,y_2\})$.
Let $Q_x$ and $Q_y$ denote  fixed  shortest paths between the pairs $x_1,x_2$ and $y_1,y_2$ in the Steiner spanner, respectively;
the key ingredient in our proof is establishing an inverse-quadratic relationship between $w(Q_x\cap Q_y)$ and $d(\{x_1,x_2\}, \{y_1,y_2\})$:  $w(Q_x\cap Q_y) =  O(\frac{\epsilon}{d(\{x_1,x_2\}, \{y_1,y_2\})^2})$. A charging argument that employs this relationship is then applied to derive the weight lower bound;
see Section \ref{sec:lb-steiner-spanners}.

As mentioned, our lightness analysis of the greedy algorithm builds on several earlier works.
In particular, the framework of \cite{CW16} was refined in the works of Borradaile, Le and Wulff-Nilsen~\cite{BLW17,BLW19};
in what follows, BLW shall be used as a shortcut for the approach of Borradaile, Le and Wulff-Nilsen~\cite{BLW19}, though we emphasize that some of
the credit that we attribute to BLW (for brevity reasons) should be attributed to the aforementioned previous works.
In BLW, the first step is to construct a hierarchical clustering $\mathcal{C}_0, \mathcal{C}_1, \ldots, \mathcal{C}_L$. Clusters in $\mathcal{C}_i$ have diameter roughly $\Theta(L_i)$ where $L_i = \frac{L_{i-1}}{\epsilon}$ and $L_0 = \frac{w(\mst)}{n-1}$.  The edge set of the greedy spanner, denoted by $E$, is also partitioned according to the clustering hierarchy, $E = E_0 \cup \ldots \cup E_L$, where the edges in $E_i$ have length $\Theta(L_i)$. Credit is then allocated to the clusters in $\mathcal{C}_0$ for a total amount of $\ce w(\mst)$ for some constant $\ce$ depending on $\epsilon$ and $d$, which will ultimately be the lightness bound. Clusters in $\mathcal{C}_0$ spend their credits in two different ways: (1) they give the clusters in $\mathcal{C}_1$ a (major) part of their credit and (2)
 use the remaining credit to pay for the spanner edges in $L_1$. Clusters in $\mathcal{C}_1$, after getting the credit from $\mathcal{C}_0$, also spend their credit in the same way:  they give clusters in $\mathcal{C}_2$ a part of their credit and use the remaining credit to pay for the spanner edges in $L_2$. Inductively, clusters in $\mathcal{C}_{i-1}$, after being given credit by the clusters in $\mathcal{C}_{i-2}$, give the clusters in $\mathcal{C}_i$ a part of their credit and use the remaining credit to pay for the edges in $L_i$. BLW showed roughly that for all $2\leq  i \leq L$:
\\(a) Each cluster $C \in \mathcal{C}_{i-1}$ would get roughly $\Theta(\ce L_{i-1})$ credits from clusters in $\mathcal{C}_{i-2}$.
\\(b) Each cluster $C \in \mathcal{C}_{i-1}$, after giving their credit to clusters in $\mathcal{C}_i$, has $\Omega(\epsilon^{O(1)}\ce L_{i-1})$ leftover credits.

Using a standard packing argument, BLW showed that each cluster in $C \in \mathcal{C}_{i-1}$ is incident to  $O(\epsilon^{-O(d)})$ edges in $E_i$ (of length $\Theta(L_i) = \Theta(L_{i-1}/\epsilon)$). Thus, by choosing $\ce = \epsilon^{- c_0d}$ for some constant $c_0$, $C$ can pay for its incident spanner edges in $E_i$. Inductively, every spanner edge will be paid at the end. Since only $\ce w(\mst)$ credits are allocated at the beginning (to $\mathcal{C}_0$), the total weight of all spanner edges is $O(\ce) = O(\epsilon^{-O(d)})$.
We first observe that the packing argument in $\mathbb{R}^d$ gives an upper bound $O(\epsilon^{-d})$ in the number of edges in $E_i$ incident to a cluster $C\in \mathcal{C}_{i-1}$. Furthermore, the bound in (b) can be made as good as $\epsilon \ce L_{i-1}$. Thus, if we are careful, choosing $\ce = \Theta(\epsilon^{-(d+2)})$ suffices, which as a result, gives us lightness bound $O(\ce) = O(\epsilon^{-(d+2)})$. To shave the extra $\epsilon^{-2}$ factor, we introduce two new ideas. Firstly, by carefully constructing the hierarchical clustering and partitioning the edge set $E$, we can reduce the worst-case bound on the number of edges in $E_i$ incident to a cluster $C\in \mathcal{C}_{i-1}$ from $O(\epsilon^{-d})$ to $O(\epsilon^{-(d-1)})$. This shaves the first $\epsilon^{-1}$ factor. Secondly, we show that in most cases, each cluster $C \in \mathcal{C}_{i-1}$, after giving its credit to clusters in $\mathcal{C}_i$, has at least $\Omega(\ce L_{i-1})$ leftover credits. Note that the leftover credit bound in BLW argument is $\Omega(\ce \epsilon L_{i-1})$. Thus, the second idea helps us in shaving another $\epsilon^{-1}$ factor.

The major technical difficulty that we are faced with is in realizing the second idea. Achieving the weaker credit leftover bound $\Omega(\ce \epsilon L_{i-1})$
(as done in BLW) is already a challenge and, in fact, sometimes impossible. This is because the credit argument has several subtleties in the way credit is distributed; a more detailed explanation is provided in Section~\ref{sec:ub-spanner-Rd}. To achieve the stronger $\Omega(\ce L_{i-1})$ leftover credit bound, we employ two new insights: (1) we can loosen the credit lower bond of each cluster $C$ to be proportional to the number of edges in $E_i$ incident to $C$ and (2) the amount of the leftover credits that $C$ has is  proportional to the number of edges in $E_i$ incident to $C$. The details of this argument are presented in Section~\ref{sec:ub-spanner-Rd}.
We remark that our argument for obtaining the optimal lightness bound is elaborate and intricate, but this should be acceptable, given that the previous  lightness bound of $O(\eps^{-2d})$
required an intricate proof, spreading over a 60-paged chapter in \cite{NS07}.

\section{Preliminaries}

For a pair $x,y$ of points in $\mathbb R^d$, we denote by $xy$ the line segment between $x$ and $y$. The distance between $x$ and $y$ will be denoted by $|xy|$.  We use $B_d(x,r)$ to denote the ball of radius $r$ centered at $x$ in $\mathbb{R}^d$.

Let $G$ be a weighted graph with vertex set $V$. We shall denote the distance between $x$ and $y$ in $G$ by $d_G(x,y)$.
 Whenever $G$ is clear from the context, we may omit the subscript $G$ in the distance notation. We use $V(G)$ and $E(G)$ to denote the vertex set and edge set of $G$. Sometimes, the vertex set of $G$ is a set of points in $\mathbb{R}^d$ and the weights of edges are given by the corresponding Euclidean distances. In this case, we use the term \emph{vertex} and \emph{point} interchangeably.

 We use $u\stackrel{G}{\leadsto}v$ to denote the shortest path from $u$ to $v$ in a graph $G$. Given two paths $P,Q$ such that the last point of $P$ is the first point of $Q$, we use $P\circ Q$ to denote the \emph{composition} of $P$ and $Q$, which is the path obtained by identifying the last point of $P$ and the first point of $Q$.  We say that two paths $P$ and $Q$ are \emph{internally vertex-disjoint} if they are vertex-disjoint except at one of their endpoints. 
 
 Let $T$ be a tree.  The  (unique) path from two nodes $x,y\in T$ is denoted by $T[x,y]$. If we remove a node $x$ (resp. $y$) from $T[x,y]$, we denote the resulting path by $T(x,y]$ (resp. $T[x,y)$). The subpath obtained by removing both $x$ and $y$ from $T[x,y]$ is denoted by $T(x,y)$.

The \emph{spread} (or \emph{aspect ratio}) of a point set $P$, denoted by $\Delta(P)$, is the ratio of the largest pairwise distance to the smallest pairwise distance, i.e.,
\begin{equation}
\Delta(P) ~=~ \frac{\max\{|xy|: x,y \in P\}}{\min\{|xy|: x \ne y \in P\}}
\end{equation}

The {\em distance} between a pair $X,Y$ of point sets, denoted by $d(X,Y)$, is the minimum distance between a point in $X$ and a point in $Y$.

We call a subset $N\subseteq P$ an \emph{$\epsilon$-cover} of $P$ if for any $x \in P$, there is a point $y \in N$ such that $|xy| \leq \epsilon$. We say $N$ is an \emph{$\epsilon$-net} if it is an $\epsilon$-cover and for any two points $x\not= y \in N$, $|xy| \geq \epsilon$.

We use $[n]$ and $[0,n]$ to denote the sets $\{1,2,\ldots,n\}$ and $\{0,1,\ldots,n\}$, respectively.  We will use the following inequalities:
\begin{equation} \label{eq:useful-ineq}
\begin{split}
    x/2 ~\leq~ &\sin(x) ~\leq~ x \qquad \mbox{ when } 0 \leq x\leq \pi/2\\
    1-x^2 ~\leq~ &\cos(x) ~\leq~ 1-x^2/3 \qquad \mbox{ when } 0 \leq x \leq \pi/2
\end{split}
\end{equation}

In this work, we are mainly interested in (Steiner) spanners with stretch $(1+\epsilon)$ for some constant $\epsilon$ sufficiently smaller than $1$.  This is without loss of generality because a (Steiner) $(1+\epsilon)$-spanner is also a (Steiner) $(1+2\epsilon)$-spanner.  We use $\epsilon \ll 1/c$, for some constant $c \geq 1$, to indicate the fact that we are assuming $\epsilon$ is sufficiently smaller than $\frac{1}{c}$.

Given a  point set $P$, we use $S_{\gr}(P)$ to denote the greedy $(1+\epsilon)$-spanner of $P$.  $S_{\gr}(P)$  is obtained by considering all pairs of points in $P$ in increasing distance order and adding to the spanner edge $xy$ whenever the distance between $x$ and $y$ in the current spanner is at least $(1+\epsilon)|xy|$. When $P$ is clear from the context, we simply denote the greedy $(1+\epsilon)$-spanner of $P$ by $S_{\gr}$.

In the analysis of the greedy spanner for doubling metrics, we often rely on the packing property  that is formally stated in the following lemma.

\begin{lemma}[Packing Lemma for Doubling Metrics] \label{lm:packing-dd} Let $P$ be a subset of points in a metric $(X,\delta)$ of doubling dimension $\ddim$ that is contained in a ball of radius $R$.  If for every $x\not= y \in P$, $\delta(x,y) >  r$, then $|P|\leq \left( \frac{4R}{r}\right)^\ddim$.\footnote{The proof of this lemma can be found in many places, e.g.~\cite{Smid09}.}
\end{lemma}

A similar packing lemma holds for Euclidean metric~\cite{Verger04}.

\begin{lemma}[Packing Lemma for Euclidean Metric] \label{lm:packing-Euclidean} Let $P$ be a subset of points in Euclidean metric of dimension $d$ that is contained in a ball of radius $R$.  If for every $x\not= y \in P$, $|xy| >  r$, then $|P| =  2^{O(d)}\left( \frac{R}{r}\right)^d$.
\end{lemma}

\section{ Lower bounds for spanners}\label{sec:lb-spanners}
In this section we provide our lower bounds for spanners, which are tight for both size and lightness for any $d = O(1)$. We start with the lower bound for $\mathbb R^2$, which is our main focus, and then generalize the argument for higher constant dimension $d = O(1)$.  For simplicity of presentation, let us consider stretch $1+c\epsilon$ for some constant $c \leq  1$ independent of $\epsilon$; the same lower bounds for stretch $1+\eps$   follow by scaling.

\paragraph{Lower bounds for spanners in $\mathbb R^2$.}
Let $C$ be a unit circle on the plane $\mathbb{R}^2$ and let $P$ be a set of points of size $ k = \frac{1}{\epsilon}$ evenly placed on the boundary of $C$. The $\mst$ of $P$ has weight at most the circumference of $C$ which is at most $2\pi$. Note that for every $p\not= q \in P$, the length of the (short) arc in $C$ connecting $p$ and $q$ is at least 
$2\pi \eps$, hence $|pq| = \Omega(\eps)$; we shall use this observation to argue that:

\begin{claim}\label{clm:lb-R2}
Let $x,y \in P$ with $|xy| = \Omega(1)$. For any $z \in P$, we have $|xz| + |yz| \ge (1 + \Omega(\epsilon))|xy|$.
\end{claim}
\begin{proof}
First note that the proof is immediate if $xy$ is not the largest side in the triangle $(x,y,z)$,
as then $|xz| + |yz| \ge |xy| + \min\{|xz|,|yz|\} \ge (1 + \Omega(\epsilon))|xy|$.

We henceforth assume that $xy$ is the largest side in the triangle $(x,y,z)$.
Let $\alpha = \angle yxz$ and $\beta = \angle xyz$.  By our assumption that $xy$ is the largest side in $(x,y,z)$, we have $0 \leq \alpha,\beta \leq \pi/2$;
note also that for any $0 \leq \alpha,\beta \leq \pi/2$, by Equation~(\ref{eq:useful-ineq}), $\cos \alpha \leq 1 - \alpha^2/3$ and $\cos \beta \leq 1-\beta^2/3$. 
Our assumption also implies that $|xy| = |xz|\cos \alpha + |yz|\cos \beta$. 
Clearly, $|xz|,|yz| \le 2$. We have:
\begin{equation*}
    \begin{split}
        \frac{|xz| + |yz|}{|xy|} &= \frac{|xz| + |yz|}{|xz|\cos \alpha + |yz|\cos \beta}
        ~\geq~ \frac{|xz| + |yz|}{|xz|(1-\alpha^2/3)+ |yz|(1-\beta^2/3)}\\
				&= 1 + \frac{|xz|\alpha^2/3 + |yz|\beta^2/3}{|xz|(1-\alpha^2/3)+ |yz|(1-\beta^2/3)} 
        ~>~ 1 + \frac{|xz|\alpha^2/3 + |yz|\beta^2/3}{|xz| + |yz|} \\
        &\geq  1 +  |xz|\alpha^2/12 + |yz|\beta^2/12,
    \end{split}
\end{equation*}
where the last inequality holds as $|xz|,|yz| \le 2$.
By the triangle inequality, $\max\{|xz|,|yz|\} \ge \frac{|xy|}{2} = \Omega(1)$. Since $\alpha \geq |yz|/2$ and $\beta \geq |xz|/2$, we have $\max\{\alpha^2/12, \beta^2/12\} = \Omega(1)$. Thus, we have
$\frac{|xz| + |yz|}{|xy|} ~\geq~ 1 +  \min(|xz|\Omega(1),|yz|\Omega(1)) ~\geq~ 1 + \Omega(\epsilon).$
\QED
\end{proof}

\begin{corollary}\label{lm:spanner-P}
 Any $(1+c\epsilon)$-spanner of $P$ must have at least $\Omega(\frac{1}{\epsilon^2})$ edges and weight at least $\Omega(\frac{1}{\epsilon^2})$,
 for some constant $c <1$ independent of $\epsilon$.
\end{corollary}
\begin{proof}
Fix an arbitrary point $x \in P$ and let $F(x)$ be the set of $\frac{1}{2\epsilon}$ furthest points from $x$ in $P$.
Let $y \in F(x)$ and note that $|xy| \geq \sqrt{2}$.  By Claim~\ref{clm:lb-R2}, $|xz| + |zy| > (1+c\epsilon)|xy|$, for any point $z \in P\setminus \{x,y\}$ and
some constant $c < 1$ independent of $\epsilon$.
Thus, any $(1+c\epsilon)$-spanner $S$ of $P$ must include all edges $(x,y)$, for all $y \in F(x)$.
Summing over all $1/\eps$ points $x \in P$, there are overall $(\frac{1}{\eps} \cdot \frac{1}{2\epsilon}) / 2 = \frac{1}{4\epsilon^2}$ such edges $(x,y)$ with $x \in P, y \in F(x)$, each with weight $\Omega(1)$, thus the corollary follows.\QED
\end{proof}

We now prove Theorem~\ref{thm:lb-d} for $d = 2$. In the following we assume that $\eps \ge \frac{1}{n}$.
\vspace{3pt}
\\\emph{Lightness bound.}  Let $P^*_n$ be any set of $n$ points obtained from the aforementioned set $P$
by adding $n - \frac{1}{\epsilon}$ points at the same locations of points of $P$; e.g., we can add all $n - \frac{1}{\epsilon}$ points to coincide with a single arbitrary point of $P$.
The weight of  $\mst(P^*_n)$ remains unchanged, i.e., $O(1)$.
By Corollary~\ref{lm:spanner-P},  any $(1+c\eps)$-spanner for $P^*_n$ must have weight $\Omega(\frac{1}{\epsilon^2})$, yielding the lightness bound.
\vspace{5pt}
\\\emph{Sparsity bound.}
Let $n' = n\eps$, and take $n'$ vertex-disjoint copies of the aforementioned point set $P$, denoted by $P_1, P_2, \ldots, P_{n'}$,
where each $P_i$ is defined with respect to a separate unit circle $C_i$ and the $n'$ circles are sufficiently far   from each other;
it suffices for the circles to be horizontally aligned so that any consecutive circles are at distance 3 from each other.
Let $Q^*_n = P_1\cup P_2\cup \ldots \cup P_{n'}$.

Let $S:= S(Q^*_n)$ be any $(1+c\epsilon)$-spanner for $Q^*_n$. For each $i\in [n']$, let $S[P_i]$ be the induced subgraph of $S$ on $P_i$.
Since the circles are sufficiently far from each other, for each $i \in [n']$, no $(1+\eps)$-spanner path between any pair $x,y \in P_i$ in $S$ may contain a point in $P_j$, for any $j \not= i$,
hence $S[P_i]$ is an $(1+c\epsilon)$-spanner of $P_i$. By Corollary~\ref{lm:spanner-P},
 we conclude that $|E(S)| ~\geq~  \sum_{i\in [n']} |E(S[P_i])| \geq n' \cdot \Omega(\frac{1}{\epsilon^2}) ~=~ \Omega(\frac{n}{\epsilon})$.

Theorem~\ref{thm:lb-d} asserts the existence of a single point set to which both the lightness and sparsity lower bounds apply. We next argue that both lower bounds apply to $Q^*_n$. Although we only showed the sparsity lower bound with respect to $Q^*_n$, the lightness lower bound follows along similar lines; specifically, 
Corollary~\ref{lm:spanner-P} implies that any $(1+c \eps)$-spanner $S$ for $Q^*_n$ must have weight at least  
$w(S) ~\geq~  \sum_{i\in [n']} w(S[P_i]) \geq n' \cdot \Omega(\frac{1}{\epsilon^2}) ~=~ \Omega(\frac{n}{\epsilon})$,
where $S[P_i]$ is the induced subgraph of $S$ on $P_i$, for each $i\in [n']$.
Since $w(\mst(Q^*_n)) = O(n') = O(n \eps)$, the required lightness bound follows, which concludes the proof of Theorem~\ref{thm:lb-d}.

\paragraph{Lower bounds for spanners in higher dimensional spaces.}
Let $\mathbb{S}_{d}$ be a $d$-dimensional unit sphere centered at the origin. Our lower bounds make use of spherical codes.

\begin{definition}[Spherical Code]\label{def:spherical-code}
	A \emph{$(d,\theta)$-spherical code} $C$ is the set of unit vectors $c_1,c_2,\ldots, c_k \in \mathbb{S}_{d}$, called \emph{codewords}, such that the angle between any two vectors is at least $\theta$. 
\end{definition}
Let $A(d,\theta)$ be the size of the largest $(d,\theta)$-spherical code. A classic bound on $A(d,\theta)$ (see~\cite{Shannon59} or~\cite{Wyner65}) is:
\begin{equation}\label{eq:spherical}
    A(d,\theta) ~\geq~ (1 + o(1))\sqrt{2\pi d}\frac{\cos \theta}{\sin^{d-1}\theta}
\end{equation}

Let $P$ be a $(d,2\pi\epsilon)$-spherical code of maximum size.
Observe that $P = \Theta(\epsilon^{-d+1})$; indeed, Equation (\ref{eq:spherical}) yields $P = \Omega(\epsilon^{-d+1})$
and $P = O(\epsilon^{-d+1})$ follows from a standard packing argument (see Lemma~\ref{lm:spherical-code-up}).

Consider any $x,y \in P$ with $|xy| = \Omega(1)$. 
Let $z$ be any point in $P\setminus \{x,y\}$ and let $\tilde C$ be the circle that goes through $x,y$ and $z$. Since $|xy| = \Omega(1)$, $\tilde{C}$ has radius $\Theta(1)$.
For any $z \in P\setminus \{x,y\}$, we have $\min(|xz|, |yz|) \ge 2\pi\epsilon$,
hence we can apply Claim~\ref{clm:lb-R2} to obtain $|xz| + |zy| \geq (1+ \Omega(\eps))|xy|$, where the constant hiding in the $\Omega$-notation
might be smaller than that in the claim statement, since the claim is stated w.r.t.\ a unit circle
whereas $\tilde C$ has radius $\Theta(1)$.
It follows that any edge $(x,y)$ with $|xy| = \Omega(1)$ must be included in any $(1+\tilde c\epsilon)$-spanner for $P$, for some constant $\tilde c < 1$ independent of $d$ and $\epsilon$.
Note also that for each point $x \in P$, there are $\Omega(|P|)$ points in $P$ lying on the hemisphere opposite to $x$.

This enables us to generalize Corollary \ref{lm:spanner-P}: Any $(1+ \tilde c\epsilon)$-spanner $S$ of $P$
must have  $\Omega(|P|^2) = \Omega(\epsilon^{-2d+2})$ edges and weight $w(S) = \Omega(|P|^2) = \Omega(\epsilon^{-2d+2})$,
 for some constant $\tilde c <1$ independent of $d$ and $\epsilon$.
Since the distance between any two nearby points in $P$ is $O(\epsilon)$, $w(\mst(P))= O(\epsilon|P|)$, and the lightness of $S$ is thus
$\Omega\left(\frac{\epsilon^{-2d+2}}{\epsilon|P|}\right)   = \Omega\left(\epsilon^{-d}\right)$.
For the size bound, we again use the trick of taking $n / |P|$ copies of the same point set $P$ that are sufficiently far from each other,
and get that the spanner size is  $\Omega( \frac{n}{|P|}\cdot |P|^2) = \Omega(n \cdot \eps^{-d+1})$;
moreover, the lightness lower bound $\Omega\left(\epsilon^{-d}\right)$ applies to this extended point set as well.
For the size and lightness bounds to apply to $n$-point sets, we assume that $n \ge |P|$, i.e., $\eps = \Omega({n}^{-\frac{1}{d-1}})$.

\section{Sparse Steiner spanners}\label{sec:ub-steiner-spanner}

In this section, we prove Theorem~\ref{thm:sparse-Steiner}. Our proof strategy consists of two steps.
In the first step we prove a relaxed version of Theorem~\ref{thm:sparse-Steiner}, where the size of the spanner depends on the spread $\Delta$ of the point set,  $O(\frac{n}{\sqrt{\epsilon}} \cdot \log(\Delta))$.
In the second step, we reduce the general problem to the relaxed one proved in the first step.
For the reduction, we solve $\log \Delta$ spanner construction problems, for point sets of spread $O(\frac{1}{\epsilon})$ each,
and demonstrate that no dependency on $\Delta$, even a logarithmic one, is incurred.
This reduction employs the standard net-tree spanner, based on a hierarchical net structure, which consists of $\log \Delta$ edge sets $E_0, E_1, \ldots, E_{\log \Delta - 1}$
that we refer to as {\em ring spanners}. Each ring spanner $E_i$ connects pairs of points at distance in the range $(r_i, O(r_i / \eps))$, where $r_i$ grows geometrically with $i$. A central ingredient of the reduction is a careful replacement of the ring spanners $R_i$ by much sparser {\em Steiner} ring spanners.

\subsection{Steiner spanners for point sets of bounded spread}

In this section we handle point sets of bounded spread, which constitutes a central ingredient in the proof of Theorem~\ref{thm:sparse-Steiner}.
Specifically, we prove the following statement.
\begin{proposition}\label{lm:bounded-spread}
For any set $P$ of $n$ points in $\mathbb{R}^2$ with spread $\le \Delta$
and any $\eps = \tilde \Omega(\frac{1}{n^2})$, there is a Steiner spanner  of size
$O(\frac{n \log \Delta}{\sqrt{\epsilon}})$.
\end{proposition}

\subsubsection{An auxiliary lemma}
We shall assume that $\epsilon$ is sufficiently smaller than $1$.
In what follows $P$ is an arbitrary (fixed) set of $n$ points in $\mathbb R^2$.
Let $X$ be a point set in $\mathbb{R}^2$; abusing notation,  when $X$ is of infinite size (such as a rectangle or any other polygonal shape),
we may use $X$ as a shortcut for $X \cap P$, i.e., to denote the set of points in $X$ that belong to $P$; we may henceforth use $|X|$ as a shortcut for $|X \cap P|$.
Let $R_1$ and $R_2$ be two rectangles of the same length and width whose sides are parallel to the $x$ and $y$ axis.
We say that  $R_1,R_2$ are \emph{horizontally} (respectively, {\em vertically}) {\em parallel} if there is a vertical (respectively, horizontal) line going through the left (respectively, top) sides of both rectangles.  The following lemma is crucial in our proof.

\begin{lemma}\label{lm:rec-spanner} Let $R_1, R_2$ be two horizontally (respectively, vertically) parallel rectangles of width (resp., length) $W$
and the same length (resp., width).  Let $d(R_1,R_2) = \frac{W}{\ell}$ where $\ell > 1$. There is a Steiner spanner $S$ with $O(\frac{\ell}{\sqrt{\epsilon}}( |R_1| + |R_2|))$ edges such that for any point $p\in R_1, q\in R_2$, $d_S(p,q) \leq (1 + \epsilon)|pq|$, assuming $\epsilon$ is sufficiently smaller than $1$.
\end{lemma}
\begin{proof}
By symmetry, it suffices to prove the lemma for two horizontally parallel rectangles.  By scaling, we may assume that $W = 1$. Let $\dl,\dr$ be two vertical lines that contain the left sides and right sides of $R_1$ and $R_2$, respectively. Let $L$ be the horizontal line segment of length $1$
with endpoints touching $\dl$ and $\dr$ and that is within distance $1 / (2\ell)$ from both $R_1$ and $R_2$
(see Figure~\ref{fig:two-rec}). We then place a set $X$ of  $\frac{\ell}{\sqrt{\epsilon}}$ evenly spaced Steiner points along $L$, and take to the Steiner spanner $S$ all edges that connect each of the Steiner points with all the points in $R_1\cup R_2$. The vertex set of $S$ is $X \cup R_1\cup R_2 $ and its edge set is of size
$|X|(|R_1| + |R_2|) = \frac{\ell}{\sqrt{\epsilon}}(|R_1| + |R_2|)$.  

We next prove the stretch bound for an arbitrary pair $p,q$ of points with $p \in R_1,q \in R_2$,
assuming $\epsilon$ is sufficiently smaller than $1$.
Let $y$ be the intersection of the line segments $pq$ and $L$
and let $x$ be the closest point of $X$ to $y$. Since the distance between consecutive points of $X$ along $L$ is $\frac{\sqrt{\epsilon}}{\ell}$, we have $|xy| \leq \frac{\sqrt{\epsilon}}{\ell}$. Note also that $|py| \geq \frac{d(R_1,R_2)}{2} = \frac{1}{2\ell}$.
Defining $\alpha = \angle xpy$, we conclude that
\begin{equation*}
\sin(\alpha)  ~=~ \frac{|xy| \sin \angle pxy}{|py|} ~\leq~ \frac{|xy|}{|py|} \leq 2\ell |xy| ~\leq~   2\sqrt{\epsilon},
\end{equation*}
hence $\alpha \leq 4\sqrt{\epsilon}$  by Equation~(\ref{eq:useful-ineq}); to apply Equation~(\ref{eq:useful-ineq}),
we rely on the fact that $\alpha < \pi/2$, which holds since $\eps$ is sufficiently small.
Let $x'$ be the projection of $x$ onto $pq$.
We have:
\begin{equation*}
|px|  ~=~ \frac{|px'|}{\cos(\alpha)} ~\leq~ \frac{|px'|}{1-\alpha^2} \leq\frac{|px'|}{1-16\epsilon} ~\le~ (1 + O(\epsilon)) |px'|,
\end{equation*}
where the first inequality follows from Equation~(\ref{eq:useful-ineq}) (recall that $\alpha < \pi/2$) and
the last inequality holds since $\eps$ is sufficiently small.
By symmetry, we have $|xq|  \le (1+O(\epsilon))|x'q|$. Since $S$ contains both edges $(p,x)$ and $(x,q)$, it follows that
$$d_S(p,q) ~\le~ |px| + |xq| ~\le~ (1 + O(\epsilon))(|px'| + |x'q|) ~=~ (1+O(\epsilon))|pq| ~\leq~ (1 + \epsilon')|pq|,$$
where $c\epsilon = \epsilon'$ and $c$ is the constant hiding in the $O$-notation above.
Thus we obtain  a spanner with stretch $1+\epsilon'$ and $O(\frac{\ell}{\sqrt{\epsilon'}} (|R_1| +|R_2|))$ edges, and the required result now follows by scaling. \QED
\end{proof}

 \begin{figure}[!htb]
        \center{\includegraphics[width=0.4\textwidth]
        {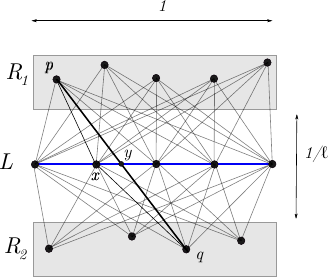}}
        \caption{Illustration for the proof of Lemma~\ref{lm:rec-spanner}. The solid blue line is $L$.}
        \label{fig:two-rec}
      \end{figure}

\subsubsection{Proof of Proposition \ref{lm:bounded-spread}} \label{s:boundedspread}
This section is devoted to the proof of Proposition \ref{lm:bounded-spread}.

Since the spread is $\le \Delta$, we may assume without loss of generality that the point set $P$ is contained in a unit axis-parallel square, $Q_0$, of side length $O(\Delta)$ and the minimum pairwise distance is $1$.

We partition all pairs of points $P\times P$ into $I = O(\log \Delta)$ sets $\mathcal{P}_1, \ldots, \mathcal{P}_{I}$ where $\mathcal{P}_i = \{(x,y) \in P\times P: 2^{i-1} \leq |x y| < 2^{i}\}$. The following definition of a ring spanner will be used here and in subsequent sections.
\begin{definition} [Ring Spanner] \label{def:ring}
Fix $c_1,c_2$ such that $0 < c_1 < c_2$ and let $t \ge 1$ be a stretch parameter.
We say that a spanner (i.e., an edge set) $R$ is a \emph{($c_1,c_2)$-ring  $t$-spanner} for a point set $P$ if $d_R(p,q) \leq t |pq|$, for any $p,q \in P$ with $c_1 \leq |pq| \leq c_2$. That is, for every $p \in P$, the ring spanner $R$ preserves distances to within a factor of $t$ between $p$ and every point $q$  in the annulus (or ring) $B_2(p,c_2)\setminus B_2(p,c_1)$.
\end{definition}
For each index $i$, we will construct a  $(2^{i-1},2^{i})$-ring Steiner $(1+\eps)$-spanner for $P$ of size $O(\frac{n}{\sqrt{\epsilon}})$, denoted by $S_i$; in particular, $S_i$ will ``handle'' all pairs in the set $\mathcal{P}_i$, i.e., 
we will have $d_{S_i}(x,y) \leq (1+\epsilon) |xy|$ for every pair $(x,y) \in \mathcal{P}_i$. Thus, $S = \cup_{i=1}^{I} S_i$ will provide a Steiner $(1+\eps)$-spanner for $P$ of the required size, thus completing the proof of Proposition~\ref{lm:bounded-spread}.

Fix an arbitrary index $i \in I$; we next describe the construction of $S_i$.
If $\mathcal{P}_i =\emptyset$, then $S_i = \emptyset$ trivially satisfies the requirements. We may henceforth assume that $\mathcal{P}_i \not= \emptyset$,
in which case the side length of $Q_0$ is at least $2^{i-1}/\sqrt{2}$. We first divide $Q_0$ into subsquares of side length $3 \cdot 2^{i}$.
(Since the side length of $Q_0$ is at least $2^{i-1}/\sqrt{2}$, we can extend the side length of $Q_0$ by a constant factor so that it is divisible by $3\cdot 2^{i}$.)
We then extend each subsquare in four directions by an additive factor of $2^{i}$ so that each extended subsquare has side length $5\cdot 2^i$.
Let $\mathcal{Q} = \mathcal{Q}_i$ be the set of all extended subsquares; we omit the subscript $i$ (in $\mathcal{Q}_i$) to avoid cluttered notation in what follows. (See Figure~\ref{fig:overlappingSQ}(a).) We make two simple observations:

\begin{observation}\label{obs:covering} For every $(x,y) \in \mathcal{P}_i$, there is at least one square in $\mathcal{Q}$ that contains both $x$ and $y$.
\end{observation}

\begin{observation}\label{obs:low-ply} Every point of $P$ is contained in at most 4 different  squares in $\mathcal{Q}$.
\end{observation}

 \begin{figure}[!h]
        \center{\includegraphics[width=1.0\textwidth]
        {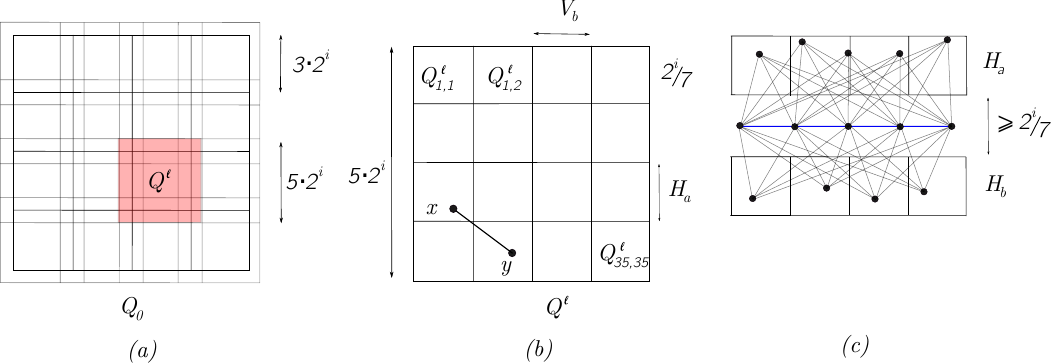}}
        \vspace{-0.4cm}
        \caption{(a) $Q_0$ is divided into extended subsquares of side length $5\cdot 2^i$. (b) $Q^{\ell}$ is further divided into $35\cdot 35$ subsquares of side length $2^{i}/7$. Then for two points $x$ and $y$ that are in two adjacent vertical bands and horizontal bands, their distance is at most $2\sqrt{2}\cdot 2^i/7 < 2^{i-1}$. (c) A Steiner spanner is constructed for every pair of non-adjacent horizontal/vertical bands. }
        \label{fig:overlappingSQ}
      \end{figure}

Let $Q^{\ell}$,  $1\leq \ell \leq |\mathcal{Q}|$, be a square in $\mathcal{Q}$. Let $\mathcal{P}_i^{\ell}\subseteq \mathcal{P}_i$ be the set of pairs of points
from $\mathcal{P}_i$ that are contained in $Q^{\ell}$. We  construct a spanner $S_{i}^{\ell}$ for $\mathcal{P}_i^{\ell}$ in the following two steps.

\paragraph{Step 1} Divide $Q^{\ell}$ into subsquares of side length $\frac{2^i}{7}$, denoted by $Q^{\ell}_{a,b}$, where $1\leq a,b\leq 35$; the first subscript index $a$ goes from top to bottom, the second subscript index $b$ goes from left to right. We call $H_{a} = \cup_{1\leq b\leq 35} Q^{\ell}_{a,b}$ a \emph{horizontal band} of $Q^{\ell}$ and  $V_{b} = \cup_{1\leq a\leq 35} Q^{\ell}_{a,b}$ a \emph{vertical band}. See Figure~\ref{fig:overlappingSQ}(b).

\paragraph{Step 2} For each pair of non-adjacent horizontal bands $H_{a}, H_{b}$, where $1\leq a < b \leq 35, b\not= a+1$, we apply  the construction in Lemma~\ref{lm:rec-spanner} with $W = 5 \cdot 2^i$ and $\ell = \frac{W}{d(H_{a},H_{b})}$ to obtain a  \emph{horizontal Steiner spanner} $HS_{a,b}$ (Figure~\ref{fig:overlappingSQ}(c)). Note that $1 < \ell = O(1)$ since $\frac{2^i}{7} \leq d(H_a,H_b) < 5\cdot 2^i$.  Thus, $HS_{a,b}$ has:
\begin{equation} \label{eq:HS-size}
|E(HS_{a,b})| = O\left(\frac{|H_{a}| + |H_b|}{\sqrt{\epsilon}}\right).
\end{equation}
Similarly, we construct a \emph{vertical Steiner spanner} $VS_{a,b}$ for each pair of non-adjacent vertical bands $V_a,V_b$, where  $1\leq a < b \leq 35, b\not= a+1$. 
Thus, $VS_{a,b}$ has:
\begin{equation} \label{eq:VS-size}
|E(VS_{a,b})| = O\left(\frac{|V_{a}| + |V_b|}{\sqrt{\epsilon}}\right).
\end{equation}
The spanner for $P_i^{\ell}$ is given by:

\begin{equation}\label{eq:S-i-ell}
S_i^{\ell} = \left(  \bigcup_{\substack{1\leq a < b \leq 35 \\ b\not= a+1}} HS_{a,b} \right) \cup \left(  \bigcup_{\substack{1\leq a < b \leq 35 \\ b\not= a+1}}VS_{a,b}  \right).
\end{equation}

This completes the description of the spanner construction $S_{i}^{\ell}$.

The spanner $S_i$ for $\mathcal{P}_i$ is given by:

\begin{equation}\label{eq:S-i}
S_i = \bigcup_{1\leq \ell \leq |\mathcal{Q}|} S_i^{\ell},
\end{equation}
and the ultimate spanner $S$ for all pairs $P\times P$ is given by:

\begin{equation}\label{eq:S}
S = \bigcup_{1\leq i \leq I} S_i.
\end{equation}

\paragraph{Stretch analysis.} Let $x,y$ be two arbitrary points of $P$ and let $i\in I$ be the index satisfying $(x,y) \in \mathcal{P}_i$.  By Observation~\ref{obs:covering}, there is an extended subsquare $Q^{\ell} \in \mathcal{Q}$ containing both $x$ and $y$. We argue that $x$ and $y$ must belong to two non-adjacent horizontal bands and/or two non-adjacent vertical bands.  Indeed, otherwise (see Figure~\ref{fig:overlappingSQ}(b)) $x$ and $y$ are contained in the same square of side length $2\cdot\frac{2^i}{7}$ consisting of four subsquares of ${Q}^{\ell}$, which implies $|xy| \leq \sqrt{2} \cdot (2\cdot\frac{2^i}{7}) < 2^{i-1}$, a contradiction to the fact that $(x,y) \in \mathcal{P}_i$. It follows that the distance between $x$ and $y$ is preserved to within a factor of $1+\eps$ by either a horizontal Steiner spanner and/or a vertical Steiner spanner.

\paragraph{Size analysis.} By Equations~\ref{eq:HS-size},~\ref{eq:VS-size} and~\ref{eq:S-i-ell}, we have:
\begin{equation*}
|E(S_i^\ell)| = O\left(\frac{\sum_{1\leq a\leq 35}|H_a|}{\sqrt{\epsilon}} + \frac{\sum_{1\leq b\leq 35}|V_b|}{\sqrt{\epsilon}}\right)  = O\left(\frac{|Q^{\ell}|}{\sqrt{\epsilon}}\right).
\end{equation*}

Thus, by Equation~\ref{eq:S-i} and Observation~\ref{obs:low-ply}, $|E(S_i)| = O(\frac{\sum_{Q^{\ell} \in \mathcal{Q}}|Q^{\ell}|}{\sqrt{\epsilon}}) = O(\frac{n}{\sqrt{\epsilon}})$.
Using Equation~\ref{eq:S}, it follows that $|E(S)| = O(\frac{nI}{\sqrt{\epsilon}}) = O(\frac{n\log \Delta}{\sqrt{\epsilon}})$, which concludes the proof of
Proposition \ref{lm:bounded-spread}.

\subsection{From bounded spread to general point sets}
In this section we provide the reduction from the general case to point sets of bounded spread.
We start (Section \ref{s:overview}) with an overview of the net-tree spanner construction, on which our reduction builds.

\subsubsection{The net-tree spanner: A short overview} \label{s:overview}

In this section we describe the {\em net-tree spanner} construction, which has several variants (see, e.g.,  \cite{GGN04,CGMZ16,GR082}).
For concreteness we consider the constructions of \cite{GGN04,CGMZ16}, which were discovered independently but are similar to each other, and apply to the wider family of   doubling metrics.
(The construction of \cite{GGN04} was presented for Euclidean spaces.)
We will later (Section \ref{subsub}) demonstrate that,  in Euclidean spaces, the   constructions of \cite{GGN04,CGMZ16} can be strengthened via the usage of Steiner points, to obtain a quadratic improvement to the spanner size.
As will be shown, this improvement requires several new insights.

Let $X$ be a point set in the Euclidean metric of dimension $d$.
By scaling, we   assume that the minimum pairwise distance in $X$ is 1, thus the spread $\Delta = \Delta(X)$ coincides with the \emph{diameter} of $X$,
i.e., $\Delta  = \max_{u, v \in X}|uv|$.
Fix any $r > 0$ and any set $Y \subseteq X$; $Y$ is called an \emph{$r$-net} for $X$ if (1) $|y, y'| \ge r$, for any  $y \ne y' \in Y$, and (2)
for each $x \in X$, there is $y \in Y$ with $|x, y| \le r$; such a net can be constructed by a greedy algorithm.

\vspace{-8pt}
\paragraph{Hierarchical Nets.}  
Write $\ell = \ceil{\log_2 \Delta}+1$, and let $\{N_i\}_{i \geq 0}^{\ell}$ be a sequence
of {\em hierarchical nets}, where $N_0 = X$, and for each $i \in [\ell]$, $N_i$ is
an $2^{i}$-net for $N_{i-1}$.
For each $i \in [0,\ell]$, $N_i$ is called the \emph{$i$-level net}.  
Note that $N_0 = X \supseteq N_1 \supseteq \ldots \supseteq N_\ell$, and $N_\ell$ contains exactly one point.
The same point of $X$ may have instances in many nets (any point of $N_i$ is necessarily also a point of $N_j$,
for each $j \in [0,i]$).

The hierarchical nets induce a hierarchical tree $T = T(X)$, called \emph{net-tree}; this tree is not required for the construction itself,
but is rather used in the stretch analysis; we refer to \cite{GGN04,CGMZ16} for the details.  

\vspace{-4pt}
\paragraph{Spanner via Cross Edges.}
The spanner   $H = H(X)$ of \cite{GGN04,CGMZ16} is obtained by adding,
for each   $i \in [0,\ell-1]$, a set $E_i$ of edges between all points of $N_i$ that are within distance $(4 + \frac{32}{\eps})2^{i}$ from each other, called \emph{cross edges}.
That is,
$E_i := \left\{(p,q) ~\vert~ p,q \in N_i, |pq| ~\le~ \left(4 + \frac{32}{\eps}\right)2^{i} \right\}$.

\begin{lemma}[Theorem 5.6~\cite{CGMZ16}, Theorem 3.2~\cite{GGN04}]\label{lm:cross-edge-spanner} Let $H = \bigcup_{i=0}^{\ell-1} E_i$. Then $H$ is a $(1+\eps)$-spanner for $X$ and $E|(H)| = n \cdot \eps^{-O(d)}$.
\end{lemma}

The bound on the number of edges in Lemma~\ref{lm:cross-edge-spanner} is weaker than the state-of-the-art in Euclidean spaces by a factor of $1 / \eps$. Our goal is to obtain a quadratic improvement over the state-of-the-art size bound in Euclidean spaces, namely $O(n \cdot \eps^{(-d+1)/2})$, using Steiner points.

\subsubsection{The reduction} \label{subsub}
To improve the size of the net-tree spanner $H$, we will improve the size bound of each edge set $E_i$.
We shall focus on 2-dimensional point sets, but our construction naturally generalizes for higher dimensions, as discussed at the end of this section.
Packing arguments yield $|E_i| \le |N_i| \cdot \eps^{-2}$, and our goal is to replace $E_i$ by an edge set $E'_i$ of size $ |N_i| \cdot O(\frac{1}{\sqrt{\epsilon}} \log \frac{1}{\epsilon})$ without increasing the pairwise distances by  much.

\paragraph{Ring Steiner spanners.}
Recall Definition \ref{def:ring} from Section \ref{s:boundedspread} of a ring spanner.
Note that the edge set $E_i$, which handles pairwise distances in the range $[2^i, \left(4 + \frac{32}{\eps}\right)2^{i}]$,
 provides a $(2^i,\left(4 + \frac{32}{\eps}\right)2^{i})$-ring 1-spanner for $N_i$. Recall that $|E_i| = |N_i| \cdot O(\eps^{-2})$.
We next show that $E_i$ can be replaced by a significantly sparser set $E'_i$ that uses Steiner points, by building on the result for bounded spread.

\begin{lemma} \label{subsets}
For each $i$, there exist subsets $N^1_i,N^2_i, \ldots, N^k_i$ of $N_i$ that form a covering, i.e., $\bigcup_{j=1}^k N^j_i = N_i$, such that
(1) for each $j \in [k]$, $N^j_i$ has spread $O(1/ \eps)$, (2) each point of $N_i$ belongs to at most four subsets from $N^1_i,N^2_i, \ldots, N^k_i$, and
(3) for each edge $(p,q) \in E_i$, there exists an index $j$ such that $p,q \in N^j_i$.
\end{lemma}
\begin{proof}
Let $B$ be the bounding box of $N_i$,
define $\tau_i = \left(4 + \frac{32}{\eps}\right)2^{i}$, and assume without loss of generality that the side lengths of $B$ are divisible by $2\tau_i$.
We partition $B$ into squares of side length $2\tau_i$ each. We extend each square equally in four directions to obtain a square of side length $3\tau_i$.
Let  $N^1_i,N^2_i, \ldots, N^k_i$ be the nonempty point sets lying in the extended squares.
Clearly, $N^1_i,N^2_i, \ldots, N^k_i$ form a covering of $N_i$.
Since $N_i$ is a $2^i$-net for $N_{i-1}$, every two points in $N_i$ are at distance at least $2^i$ from each other; thus for each $j \in [k]$,
every two points in $N^j_i$ are at distance at least $2^i$  and at most $\sqrt{2} \cdot 3\tau_i$ from each other, and item (1) holds.
Note that the overlapping region of any pair of neighboring extended squares is a rectangle of side lengths $\tau_i$ and $3\tau_i$, which implies not only item (2), but also the fact that for any pair of points within distance $\tau_i$ from each other, there is at least one extended square to which they both belong; since $|pq| \le \tau_i$ for each   $(p,q) \in E_i$, item (3) holds as well.   
\QED
\end{proof}

Fix any $i \in [0,\ell-1]$,
and consider the subsets $N^1_i,N^2_i, \ldots, N^k_i$ of $N_i$ guaranteed by Lemma \ref{subsets}.
For each $j \in [k]$, we construct a spanner $S^j_i$ for $N^j_i$ as follows.  
If $|N^j_i| \le \frac{2}{\sqrt{\epsilon}} \cdot \log \frac{1}{\epsilon}$,
we take $S^j_i$ to be the complete graph over $N^j_i$, and get a  1-spanner for $N^j_i$
(without Steiner points) with ${|N^j_i| \choose 2} \le \frac{|N^j_i|}{\sqrt{\epsilon}}\log \frac{1}{\epsilon}$ edges.
Otherwise $|N^j_i| > \frac{2}{\sqrt{\epsilon}} \log \frac{1}{\epsilon}$,
and we take $S^j_i$ to be the Steiner $(1+\eps)$-spanner for $N^j_i$ provided by Proposition \ref{lm:bounded-spread}. 
Item (1) of Lemma \ref{subsets} implies that the spread of each $N^j_i$ is $O(1/\epsilon)$, thus by Proposition \ref{lm:bounded-spread} the number of edges in $S^j_i$ in this case is also bounded by
$O(\frac{|N^j_i|}{\sqrt{\epsilon}}\log \frac{1}{\epsilon})$.
Define $E'(S^j_i)$ to be the edge set of $S^j_i$,
and let $E'_i$ be the set of edges in the union of all the spanners $S^j_i$, i.e., $E'_i = \bigcup_{j=1}^k E'(S^j_i)$.

 \begin{corollary}   \label{edgesize}
The edge set $E'_i$ is defined over a superset $N'_i = N_i \cup S_i$ of $N_i$, where $S_i$ is a set of Steiner points,
such that $|E'_i| = O(\frac{|N_i|}{\sqrt{\epsilon}} \log \frac{1}{\epsilon})$ and $E'_i$ is a  $(2^i,\left(4 + \frac{32}{\eps}\right)2^{i})$-ring $(1+\eps)$-spanner for $N_i$.
\end{corollary}
\begin{proof}
By item (2) of Lemma \ref{subsets}, we have $\sum_{j=1}^k |N^j_i| \le 4 |N_i|$.
It follows that
$$|E'_i| ~=~
\sum_{j=1}^k |E'(S^j_i)| ~=~ \sum_{j=1}^k O\left(\frac{|N^j_i|} {\sqrt{\epsilon}} \log \frac{1}{\epsilon}\right)
~=~  O\left(\frac{4|N_i|}{\sqrt{\epsilon}} \log \frac{1}{\epsilon}\right).$$

To show that $E'_i$ is a  $(2^i,\left(4 + \frac{32}{\eps}\right)2^{i})$-ring $(1+\eps)$-spanner for $N_i$,
consider any pair $p,q \in N_i$ such that $2^i \leq |pq| \leq \left(4 + \frac{32}{\eps}\right)2^{i}$.
We have $(p,q) \in E_i$ by construction, thus item (3) of Lemma \ref{subsets} implies that
 there exists an index $j$ such that $p,q \in N^j_i$. Hence there is a $(1+\eps)$-spanner path between $p$ and $q$ in the   spanner $S^j_i$ for $N^j_i$,
 and thus also in the superset $E'_i$ of $E(S^j_i)$, i.e.,  $d_{E'_i}(p,q) \leq (1+\epsilon)|pq|$. \QED
\end{proof}
\paragraph{A sparser Steiner spanner via ring Steiner spanners.}
Denote by $H'$ the spanner obtained as the union of all the Steiner ring spanners $E'_i$, i.e., $H' = \bigcup_{i=0}^{\ell-1} E'_i$.
To complete the reduction from the general case to the case of bounded spread,   thus finishing the proof of  Theorem ~\ref{thm:sparse-Steiner},
we argue that $H'$ is a Steiner $(1+O(\eps))$-spanner for $P$ with $O(\frac{n}{\sqrt{\epsilon}}\log^2 \frac{1}{\epsilon})$ edges.
(One can reduce the stretch down to $1+\eps$ by scaling.)
\vspace{4pt}
\\\emph{Stretch analysis.~} For each $0 \le i \le \ell - 1$, $E_i$ is a $(2^i,\left(4 + \frac{32}{\eps}\right)2^{i})$-ring 1-spanner for $N_i$, i.e., all   distances in
$[2^i,\left(4 + \frac{32}{\eps}\right)2^{i}]$ are preserved precisely by $E_i$; on the other hand, the stretch bound of $E'_i$ is $(1+\eps)$.
Since $H$, obtained as the union of all ring 1-spanners $E_i$, is a $(1+\eps)$-spanner for $P$ by Lemma~\ref{lm:cross-edge-spanner}, the stretch of $H'$, the spanner obtained as the union of all ring $(1+\eps)$-spanners $E'_i$, will be bounded by $(1+\eps)^2 = 1 + O(\eps)$.
\vspace{4pt}
\\\emph{Size analysis.~} We next prove that the size bound of $H'$ is   in check.
\begin{lemma}
$H'$ consists of at most $O(\frac{n}{\sqrt{\epsilon}}\log^2 \frac{1}{\epsilon})$ edges.
\end{lemma}
\begin{proof}
Denote the edge set of $H'$ by $E'$. We apply a charging argument; it will be instructive to consider another edge set $\tilde E$ with $|\tilde E| \ge |E'|$,
which has two properties that are useful for analysis purposes.

Consider the edge set $E'_i = \bigcup_{j=1}^k E'(S^j_i)$, for any level $0 \le i \le \ell - 1$.
We focus on an arbitrary index $j \in [k]$, and recall that $S^j_i$ is the $(1+\eps)$-spanner for $N^j_i$ defined above,
having at most $O(\frac{|N^j_i|}{\sqrt{\epsilon}} \log \frac{1}{\epsilon})$ edges.  
This size bound on $S^j_i$ implies that the average degree, say $D$, of a point in $N^j_i$ due to edges of $E'(S^j_i)$ is at most $O(\frac{1}{\sqrt{\epsilon}} \log \frac{1}{\epsilon})$. A priori, however, the {\em maximum} degree of a point of $N^j_i$ due to edges of $E'(S^j_i)$
could be huge and, moreover, there could be many edges in $S^j_i$ that are incident on Steiner points.
We will consider another edge set $\tilde E(S^j_i)$ of at least the same  size as $E'(S^j_i)$, defined over $N^j_i$ (i.e., with no Steiner points),
where the maximum degree of a point in $N^j_i$ due to edges of $\tilde E(S^j_i)$ does not exceed $D$.

 We distinguish between two cases.
In the case that $|N^j_i| \le D$, recall that $S^j_i$ is the complete graph over $N^j_i$, of maximum degree $|N^j_i| -1 < D$,
hence we can take $\tilde E(S^j_i)$ to be $E'(S^j_i) = {N^j_i \choose 2}$.
Otherwise $|N^j_i| > D$,
and we take $\tilde E(S^j_i)$ to be any edge set over $N^j_i$ that induces a $D$-regular graph; such an edge set clearly exists, and its size
is no smaller than that of the original edge set $E'(S^j_i)$. (It suffices for all vertices to have degree $\Theta(D)$, i.e., strict regularity is not needed.)

Observe that each edge of $\tilde E(S^j_i)$ has both endpoints in $N^j_i$, which are at distance at most $\sqrt{2} \cdot 3\tau_i$ from each other by construction,
where $\tau_i = \left(4 + \frac{32}{\eps}\right)2^{i}$ is defined in the proof of Lemma \ref{subsets}.

Define $\tilde E_i = \bigcup_{j=1}^k \tilde E(S^j_i)$.  Although $|\tilde E(S^j_i)| \ge |E'(S^j_i)|$ for each  $j \in [k]$, it may a priori be that
$|\tilde E_i| < |E'_i|$, due to potential intersections between the different edge sets $\tilde E(S^j_i)$, $j \in [k]$.
To overcome this technicality, we consider $\tilde E_i$ as a {\em multi-graph}, in which edges may appear multiple times.
By Item (2) of Lemma \ref{subsets}, each point of $N_i$
belongs to at most four edge sets from $\tilde E(S^1_i),\tilde E(S^2_i), \ldots, \tilde E(S^k_i)$, hence the degree of each point due to all edges (with all their multiplicities) of $\tilde E_i$ is at most $4D$.

Define $\tilde E = \bigcup_{i = 0}^{\ell-1} \tilde E_i$; as before, we consider this edge set $\tilde E$ as a multi-graph.
It is easy to verify that the resulting edge set $\tilde E$ satisfies $|\tilde E| \ge |E'|$. Next, we upper bound the size of $\tilde E$.

Following  \cite{CGMZ16}, for each point $p \in P$, we define $i^*(p):= \max\{i \in [0,\ell] ~\vert~ p \in N_i\}$.
To upper bound the size of $\tilde E$,  we orient each edge $(p,q) \in \tilde E$ from $p$ towards $q$ if $i^*(p) < i^*(q)$; if $i^*(p) = i^*(q)$, the edge $(p,q)$ is oriented arbitrarily.
We next bound the out-degree of an arbitrary point $p$ by all edges of $\tilde E$. Let $i$ be the minimum index such that $p$ has at least one outgoing edge in $\tilde E_i$, leading to some point $q$. We know that $|pq| \le \sqrt{2} \cdot 3\tau_i$; take $\mu$ such that $\mu \cdot 2^i = \sqrt{2} \cdot 3\tau_i + 1$, and note that
$|pq| < \mu \cdot 2^i, \mu = \Theta(1 / \eps)$.
Since $N_{i + \lceil \log \mu \rceil }$ is a $2^{i + \lceil \log \mu \rceil}$-net (of $N_{i + \lceil \log \mu \rceil - 1}$),
any two points in  $N_{i + \lceil \log \mu \rceil }$ are at distance  at least $\mu \cdot 2^i$ from each other.
Since $\mu \cdot 2^i > |pq|$ and $i^*(p) < i^*(q)$, it follows that $p$ cannot belong to $N_{i + \lceil \log \mu \rceil}$.

Thus, $p$ may only belong to the $\lceil \log \mu \rceil$ nets $N_i, N_{i+1},\ldots,N_{i + \lceil \log \mu \rceil - 1}$.
Observe that the $\lceil \log \mu \rceil$ edge sets $\tilde E_i, \tilde E_{i+1},\ldots,\tilde E_{i + \lceil \log \mu \rceil - 1}$ are defined over the nets $N_i, N_{i+1},\ldots,N_{i + \lceil \log \mu \rceil - 1}$, respectively, while all the other edge sets of $\tilde E$ are defined over different nets.
It follows that the out-degree of $p$ may increase only due to these $\lceil \log \mu \rceil$ edge sets.
In each of these edge sets the degree of $p$, let alone its out-degree,
is at most $4D$,  
hence the out-degree of $p$ due to the entire edge set $\tilde E$ is at most
$\lceil \log \mu \rceil \cdot 4D = O(\frac{1}{\sqrt{\epsilon}} \log^2\frac{1}{\epsilon})$.
Having shown that the out-degree of any point $p \in P$ due to the edge set $\tilde E$ (defined as a multi-graph) is $O(\frac{1}{\sqrt{\epsilon}} \log^2 \frac{1}{\epsilon})$,
we conclude that the size of $\tilde E$, and thus of $E'$, is bounded by $O(\frac{n}{\sqrt{\epsilon}} \log^2\frac{1}{\epsilon})$.
\QED
\end{proof}

\paragraph{Extension to any constant dimension.} The 2-dimensional construction presented here naturally generalizes to $\mathbb{R}^d$, for any constant $d$.
We shall only highlight the key components of this generalization.

Two hyperrectagles $R_1 =  [0,H]\times [0,W]^{d-1}, R_2 = [H+\frac{W}{\ell},  2H+\frac{W}{\ell}]\times [0,W]^{d-1}$, for some numbers $H, W > 0$, are called \emph{parallel hyperrectangles}; $R_1$ and $R_2$ have $d-1$ sides of length $W$ and another side of  length $H$,
and the distance between $R_1$ and $R_2$ is $\frac{W}{\ell}$.
The $d$-dimensional analogue of Lemma~\ref{lm:rec-spanner} is to construct a Steiner spanner that handles all pairs of points from $R_1$ and $R_2$ with at most $O(\frac{\ell}{\epsilon^{(d-1)/2}}(|R_1| + |R_2|))$  edges; we employ the same argument: scale the metric so that $W = 1$, and then place a grid of $O(\frac{\ell}{\epsilon^{(d-1)/2}})$ Steiner points in the $d-1$ dimensional hypercube $L_d = [H + \frac{1}{2\ell},  H + \frac{1}{2\ell}]\times [0,1]^{d-1}$ that is aligned with $R_1$ and $R_2$ in $(d-1)$ dimensions and  separates $R_1$ and $R_2$  in the middle of the remaining dimension, and finally connect all   Steiner points with all points in $R_1$ and $R_2$;
$L_d$ is the $d$-dimensional analogue of the separating segment $L$ in the proof of Lemma~\ref{lm:rec-spanner}.
Next, by using the same partitioning approach, we can construct a Steiner spanner for point sets of spread at most $\Delta$, with at most $O(\frac{n}{\epsilon^{(d-1)/2}}\log(\Delta))$ edges; this is the $d$-dimensional analogue of Proposition~\ref{lm:bounded-spread}.
Finally, the reduction from the general case to the case of bounded spread is carried out in a very similar way, by building on the net-tree spanner and replacing cross edges by Steiner ring spanners.
That is, as before, for every level $i$, we replace each edge set $E_i$, where  $E_i := \left\{(p,q) ~\vert~ p,q \in N_i, |pq| ~\le~ \left(4 + \frac{32}{\eps}\right)2^{i} \right\}$,  by  a  $(2^i, (4 + \frac{32}{\epsilon})2^i)$-ring $(1+\epsilon)$-spanner for $N_i$;
the $2$-dimensional treatment for this part extends easily to any dimension.
As a result, we get Steiner spanners for general point sets with at most $O(\frac{n}{\epsilon^{(d-1)/2}}\log^2 \frac{1}{\epsilon})$ edges.

\section{Lower bounds for sparse Steiner spanners in $\mathbb{R}^2$}\label{sec:lb-steiner-spanners}

In this section we prove Theorem~\ref{thm:lb-Steiner-R^2}. Let $U$ be a unit square with four sides $N, E, S, W$. Let $P_1$ be any set of evenly spaced points along $N$ such that the distance between two consecutive points of $P_1$ along $N$ is $c\sqrt{\epsilon\log(\frac{1}{\epsilon})}$, for a sufficiently large constant $c$.
To simplify the argument, we remove the two furthest points of $P_1$, so that every point is at distance $\ge c\sqrt{\epsilon\log(\frac{1}{\epsilon})}$ from the corners of $U$.
We define the set of points $P_2$ on $S$ similarly. Let $P = P_1\cup P_2$. See Figure~\ref{fig:lb-square}(a) for an illustration.  Our goal is to show that:

\begin{proposition}\label{lm:lb-weight} Any Steiner $(1+\epsilon)$-spanner $ST_P$ of $P$ must have $w(ST_P) = \Omega(\frac{1}{\epsilon\log \frac{1}{\epsilon}})$.
\end{proposition}

Before proving Proposition~\ref{lm:lb-weight}, we show that it implies Theorem~\ref{thm:lb-Steiner-R^2}.

\begin{claim} If Proposition~\ref{lm:lb-weight} is true, then Theorem~\ref{thm:lb-Steiner-R^2} holds.
\end{claim}
\begin{proof}
Assume that $w(ST_P) = \Omega(\frac{1}{\epsilon\log \frac{1}{\epsilon}})$. Since $w(\mst(P_1)) = O(1)$, the lightness of $ST_P$ is $\Omega(\frac{1}{\epsilon\log \frac{1}{\epsilon}})$.
Let $x_1$ and $x_2$ be any two points in $P_1$ and $P_2$, respectively. Since $|x_1x_2| = O(1)$, the shortest path between $x_1$ and $x_2$ in $ST_P$ must have length $O(1)$.
Thus, denoting by $e_{\max}$ the edge of maximum weight in $E(ST_P)$, we have
\begin{equation} \label{eq:lb-S-size}
|E(ST_P)| ~\geq~ \frac{w(ST_P)}{w(e_{\max})} ~=~ \Omega(\frac{1}{\epsilon\log \frac{1}{\epsilon}}).
\end{equation}

Recall that $|P| = O(\frac{1}{\sqrt{\epsilon \log (\frac{1}{\epsilon})}})$. Thus, $S$ has more edges than the number of points by $\Omega(\frac{1}{\sqrt{\epsilon\log \frac{1}{\epsilon}}})$  factors. 
We have achieved the required lightness and sparsity bounds, but we are not done; to show that Theorem~\ref{thm:lb-Steiner-R^2} holds, we next extend this argument to an $n$-point set, for any $n$
and $\eps$ with $\eps = \tilde \Omega(\frac{1}{n^2})$.

Let $g(\epsilon) = \frac{1}{\sqrt{\epsilon \log (\frac{1}{\epsilon})}}$ and $\alpha$ be such that $|P| = \alpha g(\epsilon)$. For simplicity of presentation, we assume that $n$ is divisible by $\alpha g(\epsilon)$, otherwise, we can always increase $n$ by at most $\alpha g(\epsilon)$  to guarantee this property. We make $k =  \frac{n}{\alpha g(\epsilon)}$
vertex-disjoint copies of the aforementioned point set $P$, denoted by $P_1, P_2, \ldots, P_k$,
where each $P_i$ is defined with respect to a separate unit square $U_i$,
where the squares $U_1,\ldots, U_k$ are horizontally aligned so that the distance between any two nearby squares is $3$ (see Figure~\ref{fig:lb-square}(b)).
Let $Q = P_1\cup P_2\cup \ldots \cup P_k$ and note that $|Q| = n$.
Let $S_Q$ be any Steiner $(1+\epsilon)$-spanner of $Q$
and let $S_Q[P_i]$ be any inclusion-wise minimal subgraph of $S_Q$ that provides a $(1+\eps)$-spanner for $P_i$, for each $i \in [k]$.
Since $d(U_i,U_j) \geq 3$ for every $i\not= j$, $S_Q[P_i]$ and $S_Q[P_j]$ must be vertex-disjoint (let alone edge-disjoint) by their minimality.
By Proposition~\ref{lm:lb-weight}, $w(S_Q[P_i]) = \Omega(\frac{1}{\epsilon\log \frac{1}{\epsilon}})$, thus $w(S_Q) =  \Omega(\frac{k}{\epsilon\log \frac{1}{\epsilon}})$. By construction,
 we have $w(\mst(Q)) \leq O(k)$, thus $w(S_Q) = \Omega(\frac{1}{\epsilon\log \frac{1}{\epsilon}})w(\mst(Q)))$, which proves the lightness bound.
For the size bound,  Equation~(\ref{eq:lb-S-size}) yields $|E(S_Q[P_i])| \geq \Omega(\frac{1}{\epsilon\log \frac{1}{\epsilon}})$, thus
$|E(S_Q)| \geq k \Omega(\frac{1}{\epsilon\log \frac{1}{\epsilon}}) = \frac{n}{\alpha g(\epsilon)} \Omega(\frac{1}{\epsilon\log \frac{1}{\epsilon}})= \Omega(\frac{n}{\sqrt{\epsilon \log \frac{1}{\epsilon}}}).$ \QED
\end{proof}

 \begin{figure}[!htb]
        \center{\includegraphics[width=1.0\textwidth]
        {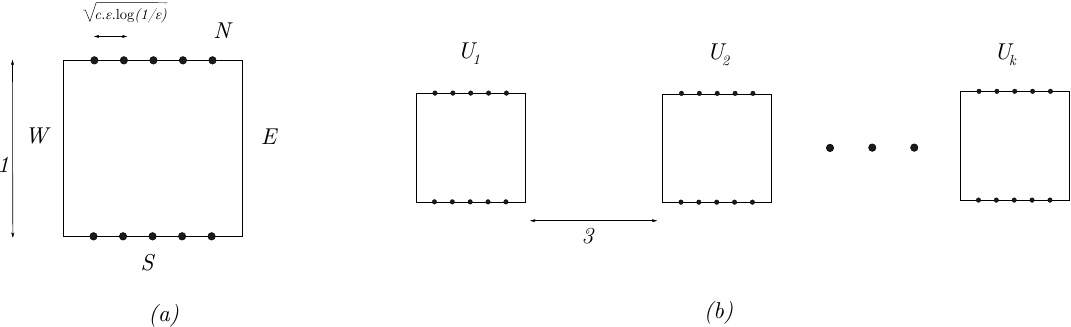}}
        \vspace{-0.4cm}
        \caption{(a) The point set $P$ lying on the $N$ and $S$ sides of the unit square in our lower bound proof. (b) Extending the example in (a)
        to the case where the point set is arbitrarily large.}
        \label{fig:lb-square}
      \end{figure}

In what follows we prove Proposition~\ref{lm:lb-weight}. Assume that $ST_P$ is a Steiner $(1+\epsilon)$-spanner for $P$ of minimum weight. Observe that one can ``planarize'' $ST_P$ without increasing its weight: whenever two edges of $ST_P$ intersect at a point on the plane that is not a point of $P$, we add the crossing point to the set of vertices of $ST_P$.
We argue that the spanner must stay inside $U$.
\begin{claim}\label{clm:inside} $ST_P \subseteq U$.
\end{claim}
\begin{proof}
For any point $p$, define its {\em projection onto $U$}, denoted by $\pr(p)$, as follows. If $p \in U$, then $\pr(p) = p$, otherwise $\pr(p)$ is the closest point on the boundary of $U$. Observe that for every pair $p,q$ of points, $|\pr(p) \pr(q)| \le |p q|$.
Thus the spanner obtained from $ST_P$ by replacing every edge $(p,q) \in ST_P$ with the ``projected'' edge $(\pr(p),\pr(q))$
has stretch and weight no greater than those of $ST_P$.
\QED
\end{proof}
  Let $x_1$ and $x_2$ be two points in $P_1$ and $P_2$, respectively. We define a \emph{bell} of radius $r$ of the line segment $x_1x_2$, denoted by $\be(x_1x_2,r)$, to be the set of points  in $U$ at distance at most $r$ from the line segment $x_1 x_2$.
 We call the boundary line segments of $\be(x_1x_2,r)$ connecting the $N$ and $S$ sides of $U$ the \emph{long boundaries} of the bell, and the other two boundary line segments are called the \emph{short boundaries}.
Since we made sure that every point of $P$ is at distance $\ge c\sqrt{\epsilon\log(\frac{1}{\epsilon})}$ from the corners of $U$, we have  $\be(x_1,x_2)\subseteq  U$,
which in particular means that all the bells (including the two extreme ones) are of precisely the same size.

 Let $Q_x$ be an arbitrary shortest path between $x_1$ and $x_2$ in $ST_P$. We claim that:

\begin{claim}\label{clm:bell} $Q_x \subseteq \be(x_1x_2, 2\sqrt{\epsilon})$.
\end{claim}
\begin{proof}
Suppose for contradiction that $Q_x$ contains a point outside $\be(x_1x_2,\sqrt{2\epsilon})$. By Claim~\ref{clm:inside}, $Q_x$ must intersect a long boundary, say $L$, of $\be(x_1x_2,\sqrt{2\epsilon})$ at a point $t$. Let $p$ be the reflection point of $x_1$ over the line defined by $L$ (see Figure~\ref{fig:bell}(a)). By the triangle inequality, we have
\begin{equation*}
\begin{split}
    w(Q_x) &\geq~ |tx_1|  + |tx_2|  ~=~ |tp| + |tx_2| \geq |px_2|
    ~=~ \sqrt{|x_1x_2|^2 + |x_1p|^2}
    ~=~ \sqrt{|x_1x_2|^2 + 16\epsilon}\\
    &\geq~ |x_1x_2|\sqrt{1 + 8\epsilon} \qquad \mbox{since } |x_1x_2|\leq\sqrt{2}\\
    &>~ (1+\epsilon)|x_1x_2| \qquad \mbox{~~since }\epsilon < 1,
\end{split}
\end{equation*}
which contradicts the fact that $Q_x$ is a $(1+\epsilon)$-spanner path for the pair $x_1,x_2$.
\QED
\end{proof}

Since the acute angle between $x_1x_2$ and $N$ is at least $\pi/4$, we have:
\begin{observation}\label{obs:bell-on-boundary}
The segments $\be(x_1x_2,2\sqrt{\epsilon})\cap N$, $\be(x_1x_2,2\sqrt{\epsilon})\cap S$ have length at most  $4\sqrt{2\epsilon}$ each.
\end{observation}

\begin{figure}[!htb]
        \center{\includegraphics[width=0.6\textwidth]
        {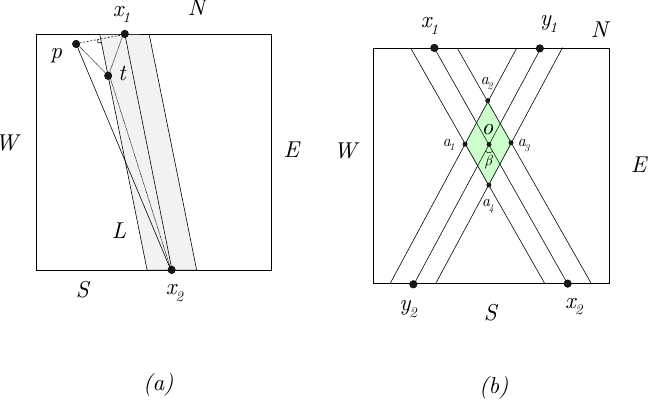}}
        \vspace{-0.4cm}
        \caption{(a) The shaded region is the bell $\be(x_1x_2,2\sqrt{\epsilon})$. (b) The parallelogram $A$ in the proof of Lemma~\ref{lm:intersect-length} is green shaded.}
        \label{fig:bell}
      \end{figure}
Let $y_1,y_2$ be two points in $P_1$ and $P_2$, respectively, such that $|\{y_1,y_2\}\cap \{x_1,x_2\}|\leq 1$. By Observation~\ref{obs:bell-on-boundary} and since the minimum pairwise distance between points in $P$ is at least $c\sqrt{\epsilon \log(\frac{1}{\epsilon})}$ for some sufficiently big constant $c$, we have:
\begin{observation}\label{obs:bell-disjoint}
If $x_1x_2 \cap y_1y_2 = \emptyset$, then  $\be(x_1x_2, 2\sqrt{\epsilon})\cap  \be(y_1y_2,2\sqrt{\epsilon}) = \emptyset$
\end{observation}

Let $Q_y$ be an arbitrary shortest path between $y_1$ and $y_2$ in $ST_P$.  For a given path $Q$, we will use $Q[a,b]$ to denote the subpath between $a$ and $b$ of $Q$.  Recall that $ST_P$ is planarized, and so $Q_x$ and $Q_y$ may only intersect at points that are vertices of $ST_P$. We want to upper bound the sum of weights of all subpaths shared by $Q_x$ and $Q_y$ (if any), denoted by  $w(Q_x\cap Q_y)$.
The next observation shows that this sum is maximized when $Q_x$ and $Q_y$ share a single subpath.

\begin{observation}\label{obs:single-path}
If $(Q_x\cap Q_y) \not= \emptyset$, then $w(Q_x\cap Q_y)$ is maximized when $Q_x\cap Q_y$ is a single path in $ST_P$.
\end{observation}
\begin{proof}
Let $p$ and $q$ be the first and the last points along $Q_y$ that belong to $Q_x\cap Q_y$, respectively.
Since $Q_x$ and $Q_y$ are shortest paths between $x_1$ and $x_2$ and between $y_1$ and $y_2$ in $ST_P$, respectively, we have $w(Q_x[p,q]) = w(Q_y[p,q])$. Thus, we can replace $Q_y[p,q]$ by $Q_x[p,q]$ to obtain another shortest path $Q'_y$ in $ST_P$ between $y_1$ and $y_2$ such that $Q_x \cap Q'_y$ is a single path and $w(Q_x\cap Q_y) \leq w(Q_{x}\cap Q'_{y})$. \QED
\end{proof}
We define the distance between two pairs $\{x_1,x_2\}$ and $\{y_1,y_2\}$, denoted by $d(\{x_1,x_2\},\{y_1,y_2\})$, to be $\max\{|x_1y_1|, |x_2y_2|\}$. The following Lemma is central to the proof of Proposition~\ref{lm:lb-weight}.
\begin{lemma}\label{lm:intersect-length}
If $d(\{x_1,x_2\}, \{y_1,y_2\}) = j\sqrt{\epsilon}$ for some sufficiently large $j$, then $w(Q_x\cap Q_y) = O(\frac{1}{j^2})$.
\end{lemma}
\begin{proof}
We assume  without loss of generality that $|x_2y_2| \ge |x_1y_1|$.
If $x_1x_2 \cap y_1y_2 = \emptyset$, then $w(Q_x \cap Q_y) = 0$ by Observation \ref{obs:bell-disjoint}, and we are done.
We henceforth assume that $x_1x_2 \cap y_1y_2 \ne \emptyset$ and let $o = x_1x_2\cap y_1y_2$. Since $|x_2y_2|\geq |x_1y_1|$, we have
$d(o,S),|ox_2|,|oy_2| \geq \frac{1}{2}$. Let $\beta = \angle x_2oy_2$.

 Let $A$ be the parallelogram given by $A = \be(x_1x_2, 2\sqrt{\epsilon})\cap \be(y_1y_2,2\sqrt{\epsilon})$, and let $a_1,a_2,a_3,a_4$ be the vertices of $A$, where $a_1,a_2,a_3,a_4$ are closest to the left, top ($N$), right and bottom ($S$) sides of the square $U$, respectively (see Figure~\ref{clm:bell}(b)).
We now bound $|oa_4|$. Since $\pi/4 \le \angle ox_2y_2 \leq \pi/2$, we have:
\begin{equation}\label{eq:beta}
    \sin(\beta) = \frac{(\sin \angle ox_2y_2) |x_2y_2|}{|oy_2|}\geq  \frac{1}{\sqrt{2}}\frac{|x_2y_2|}{|oy_2|}  \geq |x_2y_2|/2  = \frac{j\sqrt{\epsilon}}{2}.
\end{equation}

and
\begin{equation}\label{eq:beta-upper}
\sin(\beta) ~\leq~ \frac{|x_2y_2|}{|oy_2|} ~\leq~ 2j \sqrt{\epsilon}
\end{equation}

\begin{wrapfigure}{r}{0.2\textwidth}
	\vspace{-20pt}
	\begin{center}
		\includegraphics[width=0.2\textwidth]{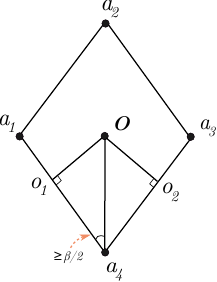}
	\end{center}
	\vspace{-10pt}
\end{wrapfigure}

Thus, by Equation~\ref{eq:useful-ineq},
\begin{equation}\label{eq:beta-bound}
\frac{j\sqrt{\epsilon}}{2} \leq \beta \leq 4j \sqrt{\epsilon}.
\end{equation}
 Let $o_1,o_2$ be the projections of $o$ onto the lines that go through the line segments $a_1a_4$ and $a_4a_3$, respectively. Since $\angle o_1a_4o + \angle oa_4o_2 = \beta$, at least one among
$\angle o_1a_4o$ and $\angle oa_4o_2$, without loss of generality  $\angle o_1a_4o$, must have degree at least $\beta / 2$.
Thus $\beta / 2 \le \angle o_1a_4o \le \pi/2$, and so
$\sin \angle o_1a_4o \ge \sin (\beta / 2) \ge (\sin \beta) / 2$, which implies that
\begin{equation*}
    |oa_4| ~=~ \frac{|oo_1|}{\sin \angle o_1a_4o} ~\leq~ \frac{2\sqrt{\epsilon}}{\sin (\beta/2)} ~\leq~ \frac{4\sqrt{\epsilon}}{\sin \beta} ~\leq~ \frac{8\sqrt{\epsilon}}{j\sqrt{\epsilon}} ~=~ \frac{8}{j}.
\end{equation*}
By the triangle inequality, we conclude that
    $d(a_4,S) \geq d(o,S) -|oa_4| \geq \frac{1}{2} - \frac{8}{j} \geq \frac{1}{4}$,
for any $j \ge 32$.

By Observation \ref{obs:single-path}, $Q_x\cap Q_y$ is a single path. Let $p$ and $q$ be its endpoints.  By Claim~\ref{clm:bell}, $p,q \in A$. Thus, $d(p,S)\geq d(a_4,S) \geq \frac{1}{4}$ and $d(q,S) \geq d(a_4,S) \geq 1/4$.
If $p = q$, then $w(Q_x \cap Q_y) = 0$ and Lemma~\ref{lm:intersect-length} holds.
Moreover, if $Q_x[p,q] \leq \sqrt{2}\epsilon$, then again the lemma must hold, since the fact that $|x_1 y_1|,|x_2 y_2| \le 1$ yields $j \leq  \frac{1}{\sqrt{\epsilon}}$.
We henceforth assume that $p \neq q$ and
\begin{equation} \label{eq:trivial-lb-Qx-Qy}
Q_x[p,q]  > \sqrt{2}\epsilon.
\end{equation}

Let $L_{x}$ (respectively, $L_y$) be the line going through $p$ and parallel to $x_1x_2$  (resp., $y_1y_2$). Let $L'_{x}$ (respectively, $L'_{y}$) be the line going through $q$ and parallel to $x_1x_2$  (resp., $y_1y_2$). Note that $\angle L_x p L_y = \angle L_x' qL_y' = \beta$.
By construction, it is readily verified that all lines $L_x,L_y,L'_x,L'_y$ intersect $S$; we henceforth define $x = L_x\cap S, y = L_y\cap S, x' = L_x' \cap S, y' = L_{y}'\cap S$.

\begin{claim}\label{clm:ypy2-and-more} All angles $\angle ypy_2, \angle xpx_2, \angle y'qy_2, \angle x'qx_2$ are at most $32\sqrt{2\epsilon}$.
\end{claim}
\begin{proof}
By symmetry, it suffices to bound $\angle ypy_2$.  Since $y \in \be(y_1y_2,2\sqrt{\epsilon})\cap S$,
Observation~\ref{obs:bell-on-boundary} yields $|y_2y| \leq 4\sqrt{2\epsilon}$. Thus, $\sin (\angle y_2 p y)\leq \frac{|yy_2|}{|py_2|} \leq \frac{|yy_2|}{d(p,S)} \leq 4|yy_2|  = 16\sqrt{2\epsilon}$. By Equation~\ref{eq:useful-ineq}, $\angle y_2 p y \leq 32\sqrt{2\epsilon}$. \QED
\end{proof}
\begin{claim}\label{clm:x2py2-lb}
$\beta/2 ~\leq ~\angle x_2py_2, \angle x_2qy_2 ~\leq~  2\beta$.
\end{claim}
\begin{proof}
By symmetry, it suffices to bound $\angle x_2py_2$.
Recall that $ j\sqrt{\epsilon}/2 \leq \beta$.  By Claim~\ref{clm:ypy2-and-more}, when $j$ is sufficiently large, it holds that:~~
$\angle x_2py_2 ~\leq~   \angle x_2px + \angle xpy +  \angle y_2py ~\leq~ \beta + 64\sqrt{2\epsilon} ~\leq~ 2\beta,$
\\$\angle x_2py_2 ~\geq~  \angle xpy  -  \angle x_2px -  \angle y_2py ~\geq~ \beta - 64\sqrt{2\epsilon} ~\geq~ \beta/2$.
\QED
\end{proof}

Let $z_1,z_2$ be two points in $P_1$ and $P_2$, respectively.  
Let $L_S$ be the line containing $S$ side of square $U$. Let $Q_z$ be the shortest path in the spanner between $z_1$ and $z_2$ in $ST_P$.  Let $a$ be any point on $Q_z$. We define the \emph{admissible triangle} of $a$ w.r.t.\ $z_2$ to be the triangle $w_1 aw_2$ such that  (a) $w_1,w_2 \in L_S$, (b)  $az_2$ is the bisector of the triangle $w_1aw_2$ and (c) $\angle w_1aw_2 = \beta/4$ (see Figure~\ref{fig:lb-ad}(a)).

\begin{figure}[!htb]
	\center{\includegraphics[width=0.9\textwidth]
		{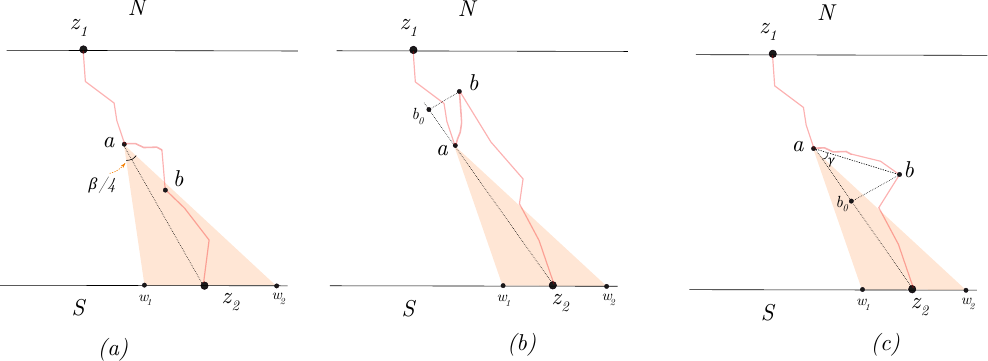}}
	\vspace{-0.4cm}
	\caption{ The light red path is $Q_z$. (a) The light orange shaded region is the admissible triangle of $a$ w.r.t.\ $z_2$. (b) The case where $b_0$ does not belong to the line segment $az_2$. (c) The case where $b$ does not belong to the admissible triangle of $a$ w.r.t.\ $z_2$.}
	\label{fig:lb-ad} 
\end{figure}

\begin{claim} \label{clm:admissible} Let $a,b$ be two points in $Q_z$ such that (1) $a \in Q_z[z_1,b]$ and (2) $w(Q_z[a,b]) >  \max\{\frac{768\sqrt{2}}{j^2}, \sqrt{2} \eps\}$. Then $b$ is in the admissible triangle of $a$  w.r.t.\ $z_2$.
\end{claim}
\begin{proof}
	Let $b_0$ be the projection of $b$ on $az_2$. Observe that $|bz_2| \geq |b_0z_2|$ by the definition of $b_0$. We claim that $b_0$ must belong to the line segment $az_2$ since otherwise, we have $|b_0z_2| \geq |az_2|$ (see Figure~\ref{fig:lb-ad}(b)). Since  $|z_1z_2| \leq \sqrt{2}$ and $w(Q_z[a,b]) > \sqrt{2} \epsilon$ by assumption (2) of the claim, we have:	
	\begin{equation*}
		\begin{split}
			w(Q_z) ~&\geq~ |z_1a| + w(Q_z[a,b]) + |bz_2|~\geq~    |z_1a| + Q_x[a,b] + |b_0z_2|  \\
			~&\geq~   |z_1a| + w(Q_z[a,b]) + |az_2| ~\geq~ |z_1z_2| +  w(Q_z[a,b])\\
			&> |z_1z_2| +  \sqrt{2}\epsilon \geq (1+\epsilon)|z_1z_2|,
		\end{split}
	\end{equation*}
	contradicting the fact that $Q_z$ is a $(1+\epsilon)$-spanner path for the pair $z_1,z_2$.

	Suppose that $b$ is not in the admissible triangle of $a$ w.r.t.\ $z_2$ (see Figure~\ref{fig:lb-ad}(c)). Then $\gamma  = \angle baz_2 \geq \beta/8 \geq \frac{j\sqrt{\epsilon}}{16}$ by Equation~\eqref{eq:beta-bound}. 
	We now show that:
	\begin{equation}\label{eq:ab-length}
		w(Q_z[a,b]) - |ab_0| \leq \sqrt{2}\epsilon
	\end{equation}
	
	If Equation~\eqref{eq:ab-length} is not true, by the triangle inequality and the fact that $|z_1z_2| \le \sqrt{2}$,
	\begin{equation*}
		\begin{split}
			w(Q_z) & \geq |z_1a| + w(Q_z[a,b]) + |bz_2|\\
			&= (|z_1a| + |ab_0| + |bz_2|) + (w(Q_z[a,b])- |ab_0|)\\
			&\geq  (|z_1a| + |ab_0| + |b_0z_2|) + (w(Q_z[a,b])- |ab_0|)\\
			&= |z_1z_2| + (w(Q_z[a,b])- |ab_0|) > |z_1z_2| + \sqrt{2} \eps \geq (1+\eps)|z_1z_2|
		\end{split}
	\end{equation*}
	contradicting the fact that $Q_z$ is a $(1+\epsilon)$-spanner path for the pair $z_1,z_2$.  Thus Equation~\eqref{eq:ab-length} must hold.
	
	Observe that $|ab_0|\leq |ab| \cos \gamma \leq Q_z[a,b]\cos \gamma$. By Equation~\eqref{eq:ab-length}, we thus have $Q_z[a,b] ~\leq~ \sqrt{2}\epsilon + Q_z[a,b] \cos \gamma$, which yields
		$Q_z[a,b] ~\leq~ \frac{\sqrt{2}\epsilon}{1-\cos\gamma} ~\leq~\frac{3\sqrt{2}\epsilon}{\gamma^2} ~\le~  \frac{768 \sqrt{2}}{j^2}$,  which contradicts assumption (2) in the claim. \QED
\end{proof}

\begin{figure}[!htb]
	\center{\includegraphics[width=0.6\textwidth]
		{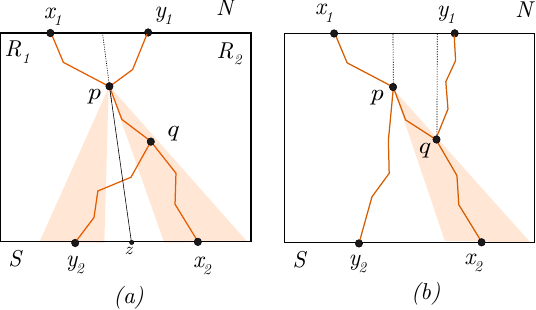}}
	\vspace{-0.4cm}
	\caption{ The light orange shaded regions are the admissible triangles. (a) Case 1: $p \in Q_y[y_1,q]$ and  (b) Case 2: $q \in Q_y[y_1,p]$.}
	\label{fig:intersect}
\end{figure}

Suppose for contradiction that Lemma~\ref{lm:intersect-length} does not hold, and specifically, that $w(Q_x[p,q]) > \frac{768\sqrt{2}}{j^2}$. This implies $w(Q_x[p,q]) > \max\{\frac{768\sqrt{2}}{j^2}, \sqrt{2} \epsilon\}$ by Equation~\eqref{eq:trivial-lb-Qx-Qy}.   Without loss of generality, we assume that $p \in Q_x[x_1,q]$. There are only two cases with respect to the position of $p$ and $q$ on $Q_y$: (1) $p \in Q_y[y_1,q]$ or (1) $q \in Q_y[y_1,p]$. See Figure~\ref{fig:intersect}.
\vspace{0.2cm}
\\\emph{Case 1: $p \in Q_y[y_1,q]$.}  By applying Claim~\ref{clm:admissible} with $(z_1,z_2)= (y_1,y_2), a = p, b = q$, it holds that $q$ is in the admissible triangle of $p$ w.r.t.\ $y_2$. Also by applying Claim~\ref{clm:admissible} with $(z_1,z_2) = (x_1,x_2), a = p, b = q$, it holds that $q$ is in the admissible triangle of $p$ w.r.t.\ $x_2$. Let $pz$ be the bisector of the angle $x_2py_2$ where $z \in S$ (see Figure~\ref{fig:intersect}(a)). Let $R_1$ and $R_2$ be the left and the right regions, respectively, of the square separated by the line containing $pz$. By Claim~\ref{clm:x2py2-lb}, $\angle zpx_2 ~= ~ \angle x_2py_2/2 ~\geq~  \beta/4$.   Since $q$  is in the admissible triangle of $p$ w.r.t.\ $x_2$, $\angle x_2 p q  \leq \beta/8$. Thus, $q \in R_2$.  By the same argument, since $q$ is  in the admissible triangle of $p$ w.r.t.\ $y_2$, $q$ must be in $R_1$, which is  a contradiction.
\vspace{0.2cm}
\\\emph{Case 2: $q \in Q_y[y_1,p]$.}  By applying Claim~\ref{clm:admissible} with $(z_1,z_2) = (x_1,x_2), a = p, b = q$,  it holds that $q$ is in the admissible triangle of $p$ w.r.t.\ $x_2$. By Claim~\ref{clm:ypy2-and-more} the acute angle between $px_2$ and $S$ is at least $\pi/4 - \angle xpx_2 ~ \geq \pi/4 - 32\sqrt{2\epsilon} ~>~ \pi/6$ when $\epsilon$ is sufficiently smaller than $1$. Since $\beta \leq \pi/2$ and $q$ is in the admissible triangle of $p$ w.r.t.\ $x_2$ , $ \angle qpx_2 ~\leq~   \beta/8 ~<~ \pi/6$. It follows that $d(p,N) < d(q,N)$. (See Figure~\ref{fig:intersect}(b).)

 By applying Claim~\ref{clm:admissible} with $(z_1,z_2) = (y_1,y_2), a= q,  b = p$, it holds that $p$ is in the admissible triangle of $q$ w.r.t.\ $y_2$. Thus by a symmetric argument as above, we get $d(q,N) < d(p,N)$, which is a contradiction.

This completes the proof of Lemma~\ref{lm:intersect-length}. \QED
\end{proof}

\paragraph{Proof of Proposition~\ref{lm:lb-weight}.}
Consider a pair $x_1\in P_1,x_2\in P_2$ of points. Let $Q_x$ be a shortest path between $x_1$ and  $x_2$ in $ST_P$.  We say that the pair $x_1,x_2$ contributes a {\em positive cost} of $w(Q_x)$ to $w(ST_P)$, and note that $w(Q_x) \ge 1$. 
Denote by $\bar P(x_1,x_2)$ the set of pairs $y_1 \in P_1, y_2 \in P_2$, such that $|\{y_1,y_2\}\cap \{x_1,x_2\}|\leq 1$.
For every pair $y_1, y_2 \in \bar P(x_1,x_2)$,
we charge a {\em negative cost} of $w(Q_x\cap Q_y)$ to $w(ST_P)$, and associate it with the pair $x_1,x_2$.
The sum of negative costs associated with pair $x_1,x_2$, denoted by $\ngc(x_1,x_2)$, is given by $\ngc(x_1,x_2) = \sum_{y_1,y_2 \in \bar P(x_1,x_2)}  w(Q_x \cap Q_y).$
\begin{observation} \label{conclude}
$w(ST_P) \ge \sum_{x_1 \in P_1, x_2 \in P_2} \left( w(Q_x) - \ngc(x_1,x_2)\right).$
\end{observation}
\begin{proof}
Observe that $w(ST_P) \geq w(\cup_{x_1 \in P_1, x_2\in P_2} Q_{x})$.  By the inclusion-exclusion principle,
\begin{equation*}
w(ST_P)   ~\geq~ \sum_{x_1 \in P_1, x_2 \in P_2}\left( w(Q_x)   - \sum_{y_1, y_2 \in \bar P(x_1,x_2)} w(Q_x\cap Q_y)\right)
~=~  \sum_{x_1 \in P_1, x_2 \in P_2}\left( w(Q_x)  - \ngc(x_1,x_2) \right) \inQED
\end{equation*}
\end{proof}

We next upper bound $\ngc(x_1,x_2)$.  By definition of $P$ and $\bar P(x_1,x_2)$, for any pair $y_1,y_2 \in \bar P(x_1,x_2)$,
we can write $d(\{x_1,x_2\}, \{y_1,y_2\}) = cj\sqrt{\epsilon\log(\frac{1}{\epsilon})}$ for some index $j$ satisfying $1\leq j\leq \frac{1}{c\sqrt{\epsilon\log(\frac{1}{\epsilon})}}$.
Fix an arbitrary index $j$ such that $1\leq j\leq \frac{1}{c\sqrt{\epsilon\log(\frac{1}{\epsilon})}}$, and note that there are at most $4j$ pairs $y_1,y_2 \in \bar P(x_1,x_2)$ such that $d(\{x_1,x_2\}, \{y_1,y_2\}) = cj\sqrt{\epsilon\log(\frac{1}{\epsilon})}$.
Take $k$ so that $k \sqrt{\eps}= cj\sqrt{\epsilon\log(\frac{1}{\epsilon})}$; by Lemma~\ref{lm:intersect-length}, the total contribution to $\ngc(x_1,x_2)$ by all such pairs is at most
$$O\left(\frac{1}{k^2}\right)   4j  ~=~ O\left(\frac{1}{\log (\frac{1}{\epsilon})} \cdot \frac{1}{c^2j^2}\right) 4j ~=~ O\left(\frac{1}{\log (\frac{1}{\epsilon})}\cdot \frac{1}{c^2j}\right).$$
Summing over all possible values of $j$
and using the fact that $c$ is sufficiently large,
 we get that  $\ngc(x_1,x_2)$ satisfies
\begin{equation*}
   \ngc(x_1,x_2) ~\le~ O\left(\frac{1}{c^2\log(\frac{1}{\epsilon})} \sum_{j=1}^\frac{1}{c\sqrt{\epsilon\log(\frac{1}{\epsilon})}} \frac{1}{j}\right) ~=~  O\left(\frac{1}{c^2\log(\frac{1}{\epsilon})} \log\left(\frac{1}{c\sqrt{\epsilon\log(\frac{1}{\epsilon})}}\right) \right)  ~\leq~ 1/2,
\end{equation*}
Hence $w(Q_x) - \ngc(x_1,x_2) \ge \frac{1}{2}$, and by Observation~\ref{conclude}
$$w(ST_P) ~\ge~ \sum_{x_1 \in P_1, x_2 \in P_2} \left( w(Q_x) - \ngc(x_1,x_2)\right) ~\geq~
\frac{|P_1||P_2|}{2} ~=~ \Omega\left(\frac{1}{\epsilon\log(1/\epsilon)}\right). \inQED$$

\section{Upper bounds for greedy light spanners in the Euclidean space} \label{sec:ub-spanner-Rd}

At a high level, we follow the lightness analysis framework of the greedy spanner by Borradaile, Le and Wulff-Nilsen~\cite{BLW19,BLW17};
 in what follows, we abbreviate  Borradaile, Le and Wulff-Nilsen~\cite{BLW19,BLW17} as BLW. We will focus on the  presentation of the BLW framework for
 \emph{doubling metrics}~\cite{BLW19}; we directly adapt their analysis to our setting. The \emph{doubling dimension} of a metric space is the smallest value $\rho$ such that every ball in the metric space can be covered by $2^\rho$ balls of half the radius; a metric space is called \emph{doubling} if its doubling dimension is bounded by some constant.
The doubling dimension is a generalization of the Euclidean dimension for arbitrary metric spaces, as the Euclidean space $\mathbb{R}^d$ equipped with an $\ell_p$ norm has doubling dimension $\Theta(d)$~\cite{GKL03}.
Thus, one can directly transfer the BLW analysis to $\mathbb{R}^d$. However, the lightness bound obtained in~\cite{BLW19} is $\epsilon^{-O(d)}$, with an unspecified constant behind the $O$-notation.

On the other hand, the goal here is to establish the precise constant in the exponent of $\epsilon$, specifically $O(\epsilon^{-d})$. This bound is optimal by the lower bound in Theorem~\ref{thm:lb-d}. To achieve the goal, we reformulate and refine the BLW analysis by providing several new insights that apply to Euclidean spaces. Our first insight is that by carefully tailoring the BLW analysis directly to $\mathbb{R}^d$, we can obtain a lightness bound of $O(\epsilon^{-(d+2)})$. This bound is better  than the  current best upper bound $O(\epsilon^{-2d})$  by Narasimhan and Smid~\cite{NS07} when $d\geq 3$. However, in $\mathbb{R}^2$, which is arguably the most important case, these two lightness bounds coincide at $\epsilon^{-4}$, which is far from the lower bound of $\epsilon^{-2}$. To shave the superfluous factor of $\epsilon^{-2}$, we employ several new insights  to improve BLW analysis. 
Before sketching the high-level ideas of our analysis, we briefly review the BLW approach, which we tailor to $\mathbb{R}^d$.  We denote by $S_{\gr} = S_\gr(P)$ the greedy spanner of the point set $P$.
 
\subsection{A brief review of BLW approach}\label{subsec:BLW}

Let $\bar{w} = \frac{w(\mst)}{n-1}$ be the average weight of an $\mst$ edge. BLW constructed a \emph{clustering hierarchy} $\mathcal{L}_0, \mathcal{L}_1,\ldots$:  a cluster in $\mathcal{L}_i$ is a subset of points\footnote{Clusters in BLW approach are defined to be subgraphs of the greedy spanner. Here we follow the standard definition of a cluster. } and is the union of some clusters in $\mathcal{L}_{i-1}$. Furthermore, each cluster in $\mathcal{L}_i$, called a \emph{level-$i$ cluster}, has diameter $O(L_i)$ where $L_i = \frac{L_{i-1}}{\epsilon}$ and $L_{0}  = \bar{w}$. (A diameter of a cluster is the maximum distance between any two points in $C$.)  Note here that the diameter of a level-$i$ cluster could be much smaller than $L_i$. The set of edges of the spanner is partitioned according to the clustering hierarchy: level-$i$ spanner edges have length $\Theta(L_i)$.

 Let $C$ be a level-$(i-1)$ cluster. We say that a level-$i$ edge of $S_{\gr}$ is incident to $C$ if one of its endpoints is in $C$.   By the standard packing argument, BLW showed that $C$ is incident to at most $\epsilon^{-O(d)}$ level-$i$ edges. We observe that the same packing argument in $\mathbb{R}^d$ gives a sharper upper bound of $O(\epsilon^{-d})$ on the number of level-$i$ edges incident to $C$.

They then introduce a charging argument to bound the weight of all edges via credits. Assume that somehow each level-$(i-1)$ cluster $C$ gets an amount of credit proportional to $L_{i-1}$,  namely $\Omega(\ce L_{i-1})$ for some parameter $\ce$ that depends on $\eps$ and $d$; the precise value of $\ce$ will be determined later. Recall the total weight of (at most $\epsilon^{-d}$) level-$i$ spanner edges incident to $C$ is $O(\epsilon^{-d}L_i) = O(\epsilon^{-(d+1)}L_{i-1})$. Thus if $\ce = \Omega(\epsilon^{-(d+1)})$, $C$'s credit is enough to pay for all of its incident level-$i$ spanner edges.

Roughly speaking, BLW connects the credits of clusters with the weight of the $\mst$ as follows. First, each cluster of $\mathcal{C}_0$ is allocated roughly $\ce \bar{w}$ credits, so that the total allocated credit is $\ce w(\mst)$, where $\ce$ is the aforementioned parameter whose value will be determined later. Importantly, credits are allocated once during the entire course of the analysis. That is, the total credit $\ce w(\mst)$ will be used to pay for \emph{every} spanner edge, and hence, $\ce$ is the lightness upper bound. Recall that $L_0 = \Theta(\bar{w})$, thus every level-$0$ cluster $C$ has $\Omega(\ce L_0)$ credits. As mentioned above (for a general level $i$), level-$0$ clusters pay for level-$1$ spanner edges, and they can afford this payment if $\ce$ is sufficiently large. (Note that level-$0$ spanner edges have total weight only $O(\epsilon^{-d})w(\mst)$, so we can ignore them from the charging argument.) However, to pay for level-$2$ spanner edges, level-$1$ clusters need to have credits. Since all the credit is allocated to level-$0$ clusters, level-$1$ clusters need to take partial credit out of level-$0$ clusters. Therefore, level-$0$ clusters are not allowed to use all of their credits to pay for level-$1$ spanner edges; they can only use the remaining credit (after level-$1$ clusters took out some) to do so.  The crux of the argument is to carefully balance the amount of credit level-$1$ clusters take from level-$0$ clusters so that they (level-$1$ clusters) have enough credit to pay for level-$2$  spanner edges, while the amount of leftover credit of level-$0$ clusters still suffice to pay for the level-$1$ spanner edges, and of course, to be able to apply the same principle to any level.
The key technical contribution of the BLW analysis is to achieve the balance by inductively guaranteeing the following two invariants at all levels $i\geq 1$:
\begin{itemize}[nolistsep,noitemsep]
\item[(a)]  Each cluster $X \in \mathcal{L}_{i}$ has at least $\Omega(\ce L_i)$ credits, which were taken from the clusters in $\mathcal{L}_{i-1}$.
\item[(b)] Each cluster $C \in \mathcal{L}_{i-1}$, after its credit was partly taken by   clusters in $\mathcal{L}_{i}$, has at least $\Omega(\ce \epsilon^{a} L_{i-1})$ leftover credits for some constant $a \geq 1$.
That is, $C$'s remaining credit is  at least an $\epsilon^{a}$ fraction of its total credit.
\end{itemize}

We note that the bound in Invariant (b) above of the BLW approach can be made as big as $\Omega(\ce \epsilon L_{i-1})$, but not bigger, i.e., the constant $a$ in the invariant needs to be at least 1.  Given this restriction and assuming the two invariants  are guaranteed at all levels, one can choose $\ce = \Theta(\epsilon^{-(d+2)})$ so that the remaining credit of $C$ is enough to pay for its incident level-$i$ spanner edges.  This yields a lightness upper bound of $O(\epsilon^{-(d+2)})$.  However, one highly nontrivial technical problem is that BLW cannot always guarantee invariant (b) for {\em all} the clusters in $\mathcal{L}_{i-1}$; we will elaborate more on this problem in the next section. Our main goal is to shave the $+2$ in the exponent of $\epsilon$ in the lightness bound.

\subsection{The high-level ideas of our analysis}\label{subsec:our-high-lv}

We employ the following two-step strategy.
\begin{itemize}[noitemsep]
\item{(Step 1)} We show that each cluster $C\in \mathcal{L}_{i-1}$ is incident to at most $O(\epsilon^{-d+1})$ level-$i$ spanner edges.
This shaves a $1/\epsilon$ factor from the naive bound $O(\epsilon^{-d})$ obtained by the standard packing argument.
\item{(Step 2)} We show that each cluster $C \in \mathcal{L}_{i-1}$, \emph{in most cases}, has $\Omega(\ce L_{i-1})$ leftover credits.
This shaves another $1/\epsilon$ factor from the credit lower bound achieved by the BLW approach.
\end{itemize}

The {\em in most cases} reservation in Step 2 cannot be omitted. Recall that in Invariant (b) mentioned in Subsection~\ref{subsec:BLW}, even the weaker bound of $\Omega(\ce \epsilon L_{i-1})$ is not achieved for {\em all} clusters in $\mathcal{L}_{i-1}$ by the BLW approach; in fact, it is highly nontrivial to achieve that weaker bound even in most cases.
Note that before $C$'s credit is taken out by clusters in $\mathcal{L}_{i}$, it only has $\Omega(\ce L_{i-1})$ credits. Thus,  Step 2 is essentially equivalent to showing that $C$ can keep a constant fraction of its credit to pay for its level-$i$ spanner edges.

Since each of the two steps shaves a $1/\epsilon$ factor, we ultimately obtain the optimal lightness upper bound of $O(\epsilon^{-d})$.
We now sketch the high-level ideas required for implementing each of these steps.

Let $C$ be a level-$(i-1)$ cluster and let $k$ be the number of level-$i$ spanner edges incident to $C$.  Let $C_1,\ldots, C_k$ be $C$'s {\em neighbors};
a level-$(i-1)$ cluster $C_j$ is said to be a neighbor of $C$ if there is a level-$i$ spanner edge connecting a point in $C_j$ with a point in $C$. For notational convenience, let $C_0 = C$.  Since the edges connecting $C_0$ and $C_j$ for all $j \in [k]$ have length $\Theta(L_i)$, using the fact that the greedy has stretch $1+\epsilon$, we can show that the distance between $C_p$ and $C_q$, for any $p\not= q \in \{0,\ldots, k\}$, is $\Omega(\epsilon L_i)$. Thus the standard packing argument implies that $k = O((\frac{L_i}{\epsilon L_i})^d) = O(\epsilon^{-d})$.  To shave  an $O(\epsilon^{-1})$ factor (see Lemma~\ref{lm:deg-K} below), we partition the space into $O(\epsilon^{-d+1})$ cones around an arbitrary point in $C$. Our insight is that, if the clustering hierarchy is constructed carefully, we can apply a basic property of the greedy spanner (see Fact~\ref{fact:edge-path-weight} below) to show that in each cone, at most one  cluster $C_j$  could be incident to $C$ in the greedy spanner $S_{\gr}$.  Thus the number of neighbors of $C$ is at most the number of cones, which is $O(\epsilon^{-d+1})$. This completes Step 1.

Before going into the details of Step 2, we first sketch the idea in the BLW approach used for obtaining the leftover credit bound of $\Omega(\ce \epsilon L_{i-1})$; here we take constant $a$ above to be $1$.  
For each level-$i$ cluster $X$, let $\child(X)$ be the set of level-$(i-1)$ clusters whose union is $X$. 
By a relatively simple argument, we can focus on the case where  $|\child(X)| = \Theta(\frac{1}{\epsilon})$.  The general idea is to show the existence of \emph{at least one} child of $X$ whose credit was not used for maintaining the credit lower bound for $X$. Inductively, such a child should have at least $\Omega(\ce  L_{i-1})$ credits by Invariant (a).
By distributing the credits of this child of $X$ to all other children of $X$, each child would get $\Omega(\ce L_{i-1}) /\Theta(\frac{1}{\epsilon}) = \Omega(\ce \epsilon L_{i-1})$ credits. This credit lower bound is of course insufficient for Step 2. To implement Step 2, we need to show
the existence of $\Omega(\frac{1}{\epsilon})$ children of $X$ (rather than a single child) whose credits were not used for maintaining the credit lower bound of $X$;
distributing the credits of this many children of $X$ to all other children of $X$ will provide the required leftover credit lower bound.

However, there are two technical subtleties of the charging argument that make the task challenging. To understand those subtleties, it is instructive to examine two simple ideas and explain why they fail. The first simple idea is to allow $X$  to take the credit of half of its children, say $\frac{1}{2\epsilon}$ children, assuming $|\child(X)| = \frac{1}{\epsilon}$.  Then $X$ would have at least $\frac{1}{2\epsilon}\Omega(\ce L_{i-1}) = \Omega(\ce L_i)$ credits. The first subtlety of the credit argument lies in the constant behind the $\Omega$ notation. Specifically, each cluster $X \in \mathcal{C}_i$ must have at least $g \ce L_i$ credits  for some universal constant $g$, for all $i$. Thus,  the total credit over $\frac{1}{2\epsilon}$ clusters in $\child(X)$ is only $ \frac{ g\ce L_{i-1}}{2\epsilon} =  \frac{g\ce L_i}{2}$, which is less than the credit lower bound $g \ce L_i$ required for $X$. Consequently, the second simple idea is to guarantee that $X$ has at least $\frac{2}{\epsilon}$ (rather than $\frac{1}{\epsilon}$)  children. Then, the total credit over half of $x$'s children would provide the required bound of $g \ce L_i$.
However, the second subtlety of the credit argument is that each cluster $X \in \mathcal{L}_i$ must have a {\em diameter}, denoted by $\dm(X)$, of at most $O(L_i)$, which in some cases prevents $X$ from having more than $\frac{1}{\eps}$ children.
The more children $X$ has, the bigger diameter it may have and, as a result, it should take more credits from its children to maintain a stronger credit lower bound of $\ce \dm(X)$.
This stronger lower bound is particularly useful, when applied inductively, in the complementary case where $X$ has relatively few children.
The precise credit invariant that we guarantee depends on both $\dm(X)$ and $L_i$ in the following natural way: Each cluster $X \in \mathcal{L}_i$ has at least $\ce \max\{\dm(X), L_i/200\}$ credits.

Due to the term $\ce \dm(X)$ in the credit lower bound, even showing the existence of a single child in $X$ whose credit is not used by $X$ (for maintaining its credit lower bound) is nontrivial. In the worst case, $\dm(X) \geq \sum_{C\in \child(X)}\dm(C)$ and hence,  to maintain the credit lower bound $\ce \dm(X)$, we must take the credits  of all children of $X$ since each child $C$ of $X$ is only guaranteed to have $\ce \dm(C)$ credits. The main insight of BLW  is that in this worst-case scenario, there would be no spanner edge of length $\Theta(L_i)$ connecting two different children of $X$. (There may be edges from $X$'s children to other level-$(i-1)$ clusters not in $X$, but this is not a problem because they can be paid for by the clusters not in $X$.)  If there is at least one spanner edge connecting two children of $X$, then 
BLW was able to show that there is at least one child in $X$ whose credit is not taken by $X$ (to maintain $X$'s credit invariant).

 Recall that our goal is to show the existence of at least $\Omega(\frac{1}{\epsilon})$ such children  of $X$ rather than one. Alas, this is not always possible as we have already pointed out. 
 To overcome this hurdle, we provide two novel and highly nontrivial insights:
\begin{itemize}[noitemsep]
\item{\bf Insight (1).~} We can relax the leftover credit lower bound of each child $C \in X$, from $\Omega(\ce L_{i-1})$ (as stated in Step 2)
to $\Omega(\ce \epsilon L_{i-1}\cdot \deg_{i-1}(C) )$ where $\deg_{i-1}(C)$ is the number of level-$i$ edges incident to $C$. When $|\deg_{i-1}(C)| \ll 1/\epsilon$, this lower bound is much smaller than   $\Omega(\ce L_{i-1})$.

\item{\bf Insight (2).~} We identify a special type of structure --- the precise definition will be given  in Section~\ref{app:details-ub-spanner} --- and  introduce the notion of \emph{debt} to handle the payment of level-$i$ edges incident to $X$  when  $X$ has such special structure.   If $X$ does not have such structure, we can show that $\sum_{C \in \child(X)} \dm(C) - \dm(X) = \Omega(t L_{i-1})$, where $t$ is the number of level-$i$ spanner edges between $X$'s children. This diameter surplus allows us to prove the leftover credit lower bound of Insight (1).
\end{itemize}

To be able to apply Insights (1) and (2), we first need to  identify all possible ``easy cases'' where the constructed clusters have a sufficiently large leftover credit. As a result, our cluster construction is significantly more technical and intricate than the construction of BLW. Before getting into the details of our analysis in Section~\ref{app:details-ub-spanner}, we next state a couple of known facts concerning the greedy spanner.

\subsection{Basic facts concerning the greedy spanner}

The following fact concerning the greedy spanner $S_{\gr}$ will be extensively used. We assume that the stretch in the greedy spanner $S_{\gr}$ is $(1+s\epsilon)$, for a sufficiently large constant $s \gg 1$.  We can recover a stretch of $(1+\epsilon')$ by setting $\epsilon' = s\epsilon$, thereby incurring a constant factor overhead in the   lightness bound.  

\begin{fact} \label{fact:edge-path-weight} For any edge $e \in S_{\gr}$  and any path $P$ in $S_{\gr} \setminus \{e\}$ 
between $e$'s endpoints, $(1 + s\epsilon)w(e) < w(P)$.
\end{fact}
\begin{proof}
Let $e'$ be the last edge examined by the greedy algorithm among the edges of $P$ and $e$; it is possible that $e' = e$. 
By the description of the greedy algorithm, $w(e) \leq w(e')$.
Moreover, by the time the algorithm examines edge $e'$, all edges of the path $P' := P \circ e \setminus e'$, which connects the endpoints of $e'$, have already been added to the greedy spanner. By the description of the greedy algorithm, we have $(1+s\epsilon)w(e') < w(P')$, hence
	\begin{equation*}
	(1 + s\epsilon)w(e) ~\leq~ (1+s\epsilon)w(e') ~<~ w(P') ~=~ w(P) + w(e) - w(e') ~\leq~ w(P),
	\end{equation*}
as required.\QED
\end{proof}

We will use the following sparsity bound of the greedy spanner in our analysis.

\begin{lemma}[Narasimhan and Smid~\cite{NS07}, Lemma 14.2.2]\label{lm:sparsity-spanner-d}  For any $n$ and for any set of $n$ points in $\mathbb{R}^d$, every vertex of $S$ has degree $O(\epsilon^{1-d})$. In particular, $|E(S_{\gr})| = O(\epsilon^{1-d}n)$.
\end{lemma}

\subsection{Light greedy spanners in $\mathbb{R}^d$: A proof of Theorem~\ref{thm:ub-d}}\label{app:details-ub-spanner}

Let $\bar{w} = \frac{w(\mst)}{n-1}$ be  the average weight of $\mst$ edges. The main goal of this section is to show the following lemma.

\begin{lemma}\label{lm:weight-reduction} Let $\delta  = 4\cdot 10^{-4}$,  $J = \lceil \log_{1+\delta}\frac{1}{\epsilon} \rceil$ and $I = \lceil \log_{1/\epsilon} n \rceil$.  Fix an arbitrary index $j \in [0,J]$ 
	and let $S^j = \cup_{i=0}^{I} S_i^j$ where $S_i^j = \{e \in S_{\gr}: \frac{(1+\delta)^j\bar{w}}{\epsilon^{i}} \leq w(e) < \frac{(1+\delta)^{j+1}\bar{w}}{\epsilon^{i}}\}$. It holds that:
	\begin{equation*}
	w(S^j ) = O(\epsilon^{-d})w(\mst)
	\end{equation*}
\end{lemma}

The value of $\delta$ in Lemma~\ref{lm:weight-reduction} is somewhat arbitrary; that is, any sufficiently small constant works.  We next argue that Lemma~\ref{lm:weight-reduction} yields Theorem~\ref{thm:ub-d}.

\begin{claim}
If Lemma~\ref{lm:weight-reduction} is true, then Theorem~\ref{thm:ub-d} holds.
\end{claim}
\begin{proof}
First, by Lemma~\ref{lm:sparsity-spanner-d}, the greedy spanner has $O(\epsilon^{-d+1} n)$ edges, hence the total weight of the spanner edges of weight at most $\bar{w}/\epsilon$ is bounded by  $O(\epsilon^{-d+1} n \frac{\bar{w}}{\epsilon}) ~=~ O(\epsilon^{-d})w(\mst)$. Next, we observe that $\cup_{j=0}^J S^j$ contains every edge of weight at least $\frac{\bar{w}}{\epsilon}$ of $S_{\gr}$ since $w(e) \leq w(\mst)$ for every edge $e$. Thus, if Lemma~\ref{lm:weight-reduction} is true, then $w(S^j) ~=~ O(\epsilon^{-d})w(\mst)$ for every $j \in [0,J]$, hence $$w(S_{\gr}) ~=~ (J+1) O(\epsilon^{-d})w(\mst) ~=~ O(\epsilon^{-d}\log_{1+\delta}\frac{1}{\epsilon})w(\mst) ~=~ \tilde O_\eps(\eps^{-d}) w(\mst),$$ and so Theorem~\ref{thm:ub-d} holds.\QED	
\end{proof}

We now proceed to proving Lemma~\ref{lm:weight-reduction}. 
Fix an arbitrary index $j \in [J]$.
We call edges in $S_i^j$ \emph{level-$i$ edges}. Edges at the same level have equal weights up to a $(1+\delta)$ factor for a very small constant $\delta$, while edges at a higher level have weights larger by at least a factor of $\frac{1}{\epsilon} / (1+\delta) \approx \frac{1}{\epsilon}$.  
Let $L_i = \frac{(1+\delta)^{j+1} \bar{w}}{\eps^i}$; $L_i$ is an upper bound on the weight of level-$i$ edges. Here we abuse  by using $L_i$ as BLW; $L_i$ in our paper is defined differently. 

As described in Section~\ref{subsec:our-high-lv}, to bound $w(S^j)$, we will construct a hierarchy of clusters  $\{\mathcal{L}_0,\mathcal{L}_1,\ldots, \}$, 
assign appropriate credit to each cluster, and use the assigned credit to pay for every edge of $w(S^j)$; thus, the total weight of $w(S^j)$ will be bounded by the total amount of credit. 

Our construction becomes simpler if every $\mst$ edge has weight at most $\bar{w}$. To guarantee this property, we subdivide each $\mst$ edge $e$ of weight more than $\bar{w}$ into $\lceil\frac{w(e)}{\bar{w}}\rceil$ edges of weight at most $\bar{w}$. The subdividing points are called \emph{virtual points}. The clusters in our construction will also include virtual points; however, virtual points are not incident to any spanner edge and hence, they are only used to simplify the argument for paying spanner edges using credits. Next, we allocate each $\mst$ edge (of length at most $\bar{w}$) $\ce \bar{w}$ credits. 
\begin{observation}\label{obs:allocated-credit} The total amount of allocated credit is $O(\ce w(\mst))$.
\end{observation}
\begin{proof}
	The total amount of allocated credit is:
	\begin{equation*}
		\sum_{e \in \mst} \ce \bar{w}\lceil \frac{w(e)}{\bar{w}} \rceil \leq  \sum_{e \in \mst} \ce (w(e) + \bar{w}) = \ce w(\mst) + \ce \sum_{e \in \mst} \bar{w} = 2\ce w(\mst) ,
	\end{equation*} 
where in the above equation, we iterate over edges in $\mst$ \emph{before} the subdivision by virtual points.
	\QED
\end{proof}

We will use this amount of allocated credit, by carefully assigning it to clusters during the course of the analysis. The following two invariants will be  inductively guaranteed in our cluster construction for each $0 \le i \le I$:
\begin{itemize}[noitemsep]
		\item \textbf{(I1)~} $\dm(C) \leq g L_i$ for all $C \in \mathcal{L}_i$, where $g = 34$.
	\item \textbf{(I2)~} Each cluster $C \in \mc_i$ has at least $\ce \max\{\dm(C), \zeta L_i\}$ credits, where $\zeta = \frac{1}{200}$.
\end{itemize}
The constant values of $\zeta$ and $g$ are also somewhat arbitrary; all we need is that $\zeta$ will be sufficiently smaller than $1$ and $g$ will be sufficiently larger than $1$.

Maintaining both invariants while guaranteeing leftover credits to pay for edges in $S^j$ is a delicate task. To facilitate the task, we will associate each cluster $C$ with a subgraph of $S_{\gr}[C]$, which is a subgraph of $S_{\gr}$ induced by points in $C$; 
it is is easier to inductively bound the diameter of a subgraph of $S_{\gr}[C]$ than  bounding the diameter of $C$ directly. 
Clearly $\dm(C)$ is at most the diameter of the associated subgraph, and hence, we it suffices to guarantee both invariants for the subgraph rather than for $C$; that is, replacing $\dm(C)$ in Invariants (I1) and (I2) with the diameter of the subgraph of $S_\gr$ associated with $C$. We shall abuse the notation and use $C$ also for referting to the subgraph of $S_{\gr}$ associated with $C$.

One key idea in our cluster construction is the notion of \emph{debt}, which helps us in handling a certain structural case, as mentioned in the paragraph that discusses Insight (2) in Section~\ref{subsec:our-high-lv}. 
Intuitively, the debt of a cluster is the total weight of spanner edges incident to points in the cluster
that were left ``unpaid'' at lower levels of the construction. However, only clusters that can potentially gain credit in subsequent levels of the construction are allowed to have debt; we identify such clusters, called \emph{debted clusters},  
by examining the way in which they are connected via $\mst$ edges. To this end, we maintain a \emph{cluster tree} $\mathcal{T}_i$ for each level $i$, as described next.

\begin{itemize}[noitemsep]
	\item \textbf{(I3)~}   There is a cluster tree $\mathcal{T}_i$ whose nodes correspond to level-$i$ clusters and edges correspond to $\mst$ edges connecting level-$i$ clusters, such that debted clusters are leaves of $\mathcal{T}_i$ and each has debt at most $ 4g^2\zeta^{-2}\epsilon^{-2}(\sum_{j=1}^{i}L_j)$. Additionally, the credit of edges of $\mathcal{T}_i$ has not been  assigned to any cluster in levels less than $i$.
\end{itemize}

That is, we only allow leaves of $\mathcal{T}_i$ to have debt, and the debt of each cluster is not too big.  

To construct a set of clusters at level 0 satisfying invariants (I1-I3), we prove the following lemma.

\begin{lemma}\label{lm:base-case}  There is a vertex partition, including virtual points,
	such that each vertex set $C$ induces a subtree $T_C$ of $\mst$ such that $L_0 \leq \dm(T_C)\leq 6L_0$, and the total credit of all edges in $T_C$ is at least $\ce \max\{\dm(T_C),\zeta L_0\}$.
\end{lemma}
\begin{proof}
	We greedily break the $\mst$ into subtrees of diameter at least $L_0$ and at most $6L_0$ in two steps. In Step 1, in, we iteratively break a minimal subtree of diameter at least $L_0$ (and at most $2L_0$ since each edge has length at most $L_0$) from a tree of diameter at least $L_0$. After the first step, each remaining subtree, say $X$, has diameter at most $L_0$ and has an edge, say $e$ to a subtree $T_C$ formed in Step 1. In Step 2, we augment $X$ and $e$ to $T_C$. The augmentation in Step 2 (additively) increases the diameter of $T_C$ by at most $4L_0$ since $e$ has length at most $L_0$. Thus, $T_C$ has diameter at most $6L_0$ after Step 2.
	
	For each subtree $T_C$, let $D$ be a path realizing the diameter of $T_C$. Clearly, $D$ is a subpath of $\mst$. Recalling that each $\mst$ edge (of length at most $\bar{w}$) is allocated $\ce \bar{w}$ credits and since $\dm(T_C)\geq L_0$, it follows that the total credit of edges in $D$ is at least $\ce w(D) = \ce \dm(T_C) = \max\{\dm(T_C), \zeta L_0\}$. \QED
\end{proof}

 We take the level-$0$ clusters $\mathcal{L}_0$ to be vertices in the subtrees $T_C$ provided by Lemma~\ref{lm:base-case}.  Invariant (I1) is satisfied since $g > 6$. Invariant (I2) is satisfied by assigning the credit of the $\mst$ edges in $C$ to $C$. 
The cluster tree $\mathcal{T}_0$ has a node corresponding to a subtree of $\mst$ and each edge between two nodes is the edge connecting two corresponding subtrees.  Since level-$0$ clusters are disjoint subtrees of $\mst$ and have no debt,  Invariant (I3) is satisfied.

\subsubsection{Level-$i$ cluster construction}

We refer to level-$(i-1)$ clusters as \emph{$\epsilon$-clusters}. We simply use {\em clusters} to refer to  level-$i$ clusters. Likewise, we refer to level-$i$ spanner edges as \emph{spanner edges}, unless specified otherwise. Let $\mathcal{K}$ be the {\em cluster graph}, where each vertex of $\mathcal{K}$ corresponds to an $\epsilon$-cluster and each edge of $\mathcal{K}$ corresponds to a spanner edge connecting the two respective $\epsilon$-clusters. Let $\dk$ be the maximum degree of vetices  in $\mathcal{K}$. BLW used
the standard packing argument to show (Lemma 3.1 in~\cite{BLW19}) that $\dk = \epsilon^{-O(\ddim)}$, for metrics of doubling dimension $\ddim$. By adapting their proof to point sets in $\mathbb{R}^d$, specifically, using Lemma~\ref{lm:packing-Euclidean}, one can derive an improved upper bound of $O(\epsilon^{-d})$ on $\dk$. 
Our first insight is that the degree of $\mathcal{K}$ is smaller than this bound by a factor of $\eps$. 

\begin{wraptable}{r}{ 9cm}
	\centering
	\begin{tabular}{|l | l|}
		\hline
		$s$  & constant in the stretch $t = 1+s \epsilon$; $s\geq 10g+3$. \\ \hline
		$g$ & constant in Invariant (I1); $g = 34$. \\ \hline
		$\zeta$ & constant in Invariant (I3); $\zeta = 1/200$. \\ \hline
		$\delta$ & constant in Lemma~\ref{lm:weight-reduction}; $\delta = 4\cdot 10^{-4}$. \\ \hline
	\end{tabular}
	\caption{\footnotesize{Important constants used in cluster construction.}}
	\label{table:const}
\end{wraptable}

\begin{restatable}{lemma}{degK} \label{lm:deg-K} 
	$\mathcal{K}$ is a simple graph with $\dk = O(\epsilon^{-d+1})$, where $s \geq 40g + 3$, $\delta = \frac{1}{8(6g+2)}$ and $\epsilon \ll \frac{1}{s}$.
\end{restatable}

Recall 
our assumption that the spanner has stretch $(1+s \epsilon)$, for a sufficiently big constant $s$ (see Table~\ref{table:const}). 
The proof of Lemma~\ref{lm:deg-K} is somewhat technical, and is deferred to Section~\ref{subsec:Proof-Deg-K}.  

Let $\mathcal{T}_{i-1}$ be the cluster tree of $\epsilon$-clusters as guaranteed inductively by Invariant (I3) for level $i-1$. Let $\mathcal{G} = \mathcal{K}\cup \mathcal{T}_{i-1}$. We denote by $\mv(\mathcal{G})$ and $\me(\mathcal{G})$ the vertex and edge sets of $\mg$, respectively. 
Since $\mv(\mathcal{K}) = \mv(\mathcal{T}_{i-1})$, $\mathcal{T}_{i-1}$  is a spanning tree of $\mathcal{G}$. We refer to vertices of $\mathcal{G}$ as nodes. Instead of constructing level-$i$ clusters explicitly, we construct a collection of disjoint subgraphs of $\mathcal{G}$, where each subgraph can be mapped to a level-$i$ cluster in a natural way. Each node $\pzc[x]\in \mathcal{G}$ will be assigned a weight in the following way:
\begin{equation}\label{eq:weight-node}
	w(\pzc[x]) = \max\{\dm(\pzc[x]), \zeta L_{i-1}\};	
\end{equation}
 here $\dm(\pzc[x])$ is the diameter of the $\eps$-cluster $\pzc[x]$. Thus, edges and nodes of $\mathcal{G}$ are both weighted, where the weight of each edge is the Euclidean distance between its endpoints.  Given a path $\mathcal{P}$ in $\mathcal{G}$, we define the \emph{augmented weight} of $\mathcal{P}$, denoted by $\adm(\mathcal{P})$, to be  the total weight of  nodes and edges in $\mathcal{P}$. That is:
\begin{equation}\label{eq:def-adm}
\adm(\mathcal{P}) = \sum_{\mathpzc{x} \in \mv(\mathcal{P})}w(\mathpzc{x}) + \sum_{e\in \me(\mathcal{P})} w(e)
\end{equation}
   
The \emph{augmented distance} betweeen any two nodes of $\mathcal{G}$ is the minimum augmented weight of any path between them.   The {\em augmented diameter} of a subgraph $\mc$ of $\mathcal{G}$, denoted by $\adm(\mathcal{C})$, is the maximum augmented distance (in $\mc$) between any two nodes in $\mc$. 
Instead of bounding the diameter of a level-$i$ cluster, we can bound the augmented diameter of the corresponding subgraph $\mc$  of $\mathcal{G}$ since the diameter of the cluster is at most $\adm(\mc)$. If we can guarantee that $\mc$ has at least $\ce\max\{\adm(\mc),\zeta L_i\}$ 
credits and $\adm(\mc) \leq gL_i$, then the corresponding cluster will satisfy Invariants (I1) and (I2). 
Thus, we can work exclusively with subgraphs of $\mathcal{G}$ without reference to the corresponding level-$i$ clusters.

 For each node  $\pzc[x] \in \mv(\mathcal{G})$, denote by $\cred(\pzc[x])$ the credit of $\pzc[x]$. Similarly, the credit of an edge $e\in \me(\mg)$ is denoted by $\cred(e)$. Since edges of $\mathcal{K}$ are not $\mst$ edges, $\cred(e) = 0$ for every edge $e \in \me(\mathcal{K})$.
 
  Note by Invariants (I2) and (I3) for level $i-1$ that:
 \begin{equation}\label{eq:credit-nodes-edges}
 	\cred(\pzc[x]) \geq \ce w(\pzc[x]) \quad \mbox{and} \quad \cred(e) \geq \ce w(e) \quad \forall e \in \me(\mathcal{T}_{i-1})
 \end{equation}

 Let $\mathcal{C}$ be a subgraph of $\mathcal{G}$; we denote by $\cred(\mathcal{C})$ the total credit of nodes and edges in $\mathcal{C}$. That is,
 \begin{equation}\label{eq:cred-subgraph}
 	\cred(\mathcal{C}) =  \sum_{\mathpzc{x} \in \mv(\mathcal{C})}\cred(\mathpzc{x}) + \sum_{e\in \me(\mathcal{C})} \cred(e) 
 \end{equation}
Similarly, for a susbet of nodes $\mathcal{X}\subseteq \mv(\mathcal{G})$, we define $\cred(\mathcal{X}) = \sum_{\pzc[x] \in \mv(\mathcal{G})}\cred(\pzc[x])$.
 
\begin{lemma}\label{lm:prop-T} $\mathcal{T}_{i-1}$ satisfies the following: 
	\begin{enumerate}[nolistsep,noitemsep]
	\item For any path $\mathcal{P}$ of $\mathcal{T}_{i-1}$,  $\cred(\mathcal{P})  \geq \ce \adm(\mathcal{P})$.
		\item Each leaf node of $\mathcal{T}_{i-1}$ has debt at most $O(\epsilon^{-1})L_i$. Internal nodes have no debt.
		\item $\zeta \epsilon L_i \leq w(\pzc[x]) \leq g\epsilon L_i$ for every node $\pzc[x] \in \mathcal{T}_{i-1}$.
	\end{enumerate} 
\end{lemma}
\begin{proof}
Item (1) follows directly from Equation~\eqref{eq:credit-nodes-edges}. 
For Item (2), invariant (I3) implies that each leaf node of $\mathcal{T}_{i-1}$ has debt at most
\begin{equation*}
\begin{split}
4g^2 \zeta^{-2} \epsilon^{-2} \left(\sum_{a=1}^{i-1}L_a \right) &= ~ O(\epsilon^{-2})\left(\sum_{a=0}^{i-1}\epsilon^{a}\right)L_{i-1}~=  ~O(\epsilon^{-2})L_{i-1}  ~=~ O(\epsilon^{-1})L_i
\end{split}
\end{equation*}
when $\epsilon \leq 1/2$. Non-leaf nodes of $\mathcal{T}_{i-1}$ have  no debt by Invariant (I3).

Item (3) follows from the definition of the weight (Equation~\eqref{eq:weight-node}) and Invariant (1) for level $i-1$.\QED
\end{proof}

By Item (2) of Lemma~\ref{lm:prop-T} and Lemma~\ref{lm:deg-K}, the debt of any $\epsilon$-cluster $\mathpzc{x}$ does not exceed the (worst-case) total weight of all spanner edges incident to $\mathpzc{x}$ by more than a constant factor, when $d=2$. When $d \geq 3$, this debt is negligible compared to the worst-case bound on the total weight of all incident  spanner edges.

\paragraph{Stopping condition.} If $\adm(\mathcal{T}_{i-1}) \leq  gL_i$, we stop the cluster construction after level $i-1$, i.e., the level-$i$ cluster construction doesn't do anything, and this is the end of the cluster construction. Note that there can be no level-$j$ edges for any $j \ge i+1$, since a level-$j$ spanner edge, if  any, has length at least $\frac{L_{i+1}}{(1+\delta)} \geq  \frac{L_i}{2\epsilon} \gg g L_i$ when $\epsilon \ll \frac{1}{g}$; contradicting that each edge is the shortest path between its endpoints. 
We use all the credit of each $\epsilon$-cluster, say $\pzc[x]$, to pay for all of its incident level-$i$ spanner edges and debt. By Invariant (I1), $\pzc[x]$ has at least $\ce \zeta\epsilon L_i$ credits. Thus, by Lemma~\ref{lm:deg-K} and Item (2) of Lemma~\ref{lm:prop-T}, $\pzc[x]$'s credit is sufficient when  $\ce ~=~ \Omega(\epsilon^{-d})$.

Henceforth, we assume that $\adm(\mathcal{T}_{i-1}) > g L_i$. Our construction has five steps. As mentioned above, we will focus on the construction of subgraphs of $\mathcal{G}$, and each subgraph is mapped to a corresponding level-$i$ cluster in a natural way. In what follows, we will abuse notation by referring to subgraphs of $\mathcal{G}$ as clusters

\begin{definition}[Leftover Credit]\label{def:leftover} For each subgraph $\mathcal{C}$ of $\mathcal{G}$ constructed in the following steps that corresponds to a level-$i$ cluster, we assign at least $\ce \max\{\mathcal{C},\zeta L_i\}$ credits from nodes and $\mst$ edges in $\mathcal{C}$ to $\mathcal{C}$ so that it satisfies Invariant (I2). The remaining credit of nodes and $\mst$ edges in $\mathcal{C}$ is called the \emph{leftover credit} of $\mathcal{C}$.
\end{definition}

\paragraph{Step 0: Type-0 clusters.} A node $\pzc[x] \in \mathcal{G}$ is called a  \emph{high degree} node if its degree  in $\mathcal{G}$ is at least $2g\zeta^{-1} \epsilon^{-1}$; otherwise it is called a \emph{low degree} node. 

\begin{lemma}[Type-0 Clusters]\label{lm:Type-0} We can construct a collection $\mathbb{H}_0$ of subgraphs of $\mathcal{G}$, called Type-0 clusters, such that:
	\begin{enumerate}[noitemsep,nolistsep]
		\item Every high degree node and its neighbors in $\mathcal{G}$ are contained in subgraphs of $\mathbb{H}_0$. 
		\item Each subgraph $\mathcal{C}\in \mathbb{H}_0$ has $\adm(\mathcal{C})\leq 16L_i$ and  contains a high degree node $\pzc[x]$ and all of its neighbors.
		\item Let $\mathcal{F}_{i-1}$ be the forest obtained from $\mathcal{T}_{i-1}$ by removing every node in subgraphs of $\mathbb{H}_0$. Then, every node in $\mathcal{F}_{i-1}$ is incident to $2g\zeta^{-1} \epsilon^{-1} = O(\epsilon^{-1})$ spanner edges,  and every tree $\mathcal{T} \subseteq \mathcal{F}_{i-1}$ has $\adm(\mathcal{T}) \geq \zeta L_i$.
	\end{enumerate}
\end{lemma}
\begin{proof} We construct $\mathbb{H}_0$ in several steps. Initially, every node is  \emph{unmarked}.
	
\begin{wrapfigure}{r}{0.5\textwidth}
	\vspace{-30pt}
	\begin{center}
		\includegraphics[width=0.30\textwidth]{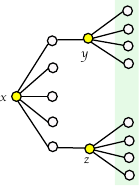}
	\end{center}
	\caption{\footnotesize{$\pzc[x]$ and its neighbors are grouped in step 1. Nodes $\pzc[y]$ and $\pzc[z]$ are augmented to the cluster of $\pzc[x]$ in step 2, and their neighbors in the green-shaded region are augmented in step 3.}}
	\label{fig:Type0}
		\vspace{-15pt}
\end{wrapfigure}

\begin{enumerate}
	\item For each high-degree unmarked  node $\pzc[x]$ whose neighbors are all unmarked, we form a subgraph from  $\pzc[x]$, its neighbors and $\pzc[x]$'s incident edges in $\mathcal{G}$, and then add the subgraph to  $\mathbb{H}_0$;
	we mark $\pzc[x]$ and all its neighbors.
	We repeat until every high-degree node has at least one marked neighbor.
	\item For each remaining high-degree unmarked node $\pzc[y]$,  where at least one neighbor of $\pzc[y]$ in $\mathcal{G}$ was marked in step 1,  
	we mark $\pzc[y]$ and then add an edge connecting $\pzc[y]$ and one arbitrary
	marked neighbor  of $\pzc[y]$, denoted by $\pzc[x]$, to the subgraph containing $\pzc[x]$. 
	We repeat this step until every high-degree node is marked.
	\item For each unmarked node $\pzc[y]$ that is a neighbor of at least one high-degree node, we mark $\pzc[y]$ and then add an edge connecting $\pzc[y]$ and one arbitrary high-degree neighbor of $\pzc[y]$, denoted by $\pzc[x]$, to the subgraph containing $\pzc[x]$. We repeat this step until every node $\pzc[y]$ with a high-degree neighbor is marked.
	\item  Let $\mathcal{T}$ be any connected component of the forest obtained from $\mathcal{T}_{i-1}$ by removing all marked nodes from $\mathcal{T}_{i-1}$. If $\mathcal{T}$  has augmented diameter less than $\zeta L_i$, it must be connected by an $\mst$ edge, say $e$, to a node, say $\pzc[x]$, marked in steps 1-3. We add $\mathcal{T}$ and $e$ to the subgraph containing $\pzc[x]$, and then mark every node of $\mathcal{T}$. We repeat this step until every tree in the forest induced by the unmarked nodes has augmented diameter at least $\zeta L_i$. We denote the forest by $\mathcal{F}_{i-1}$.
\end{enumerate}

Clearly, Items (1) and (3) of the lemma statement follow directly from the construction. As for Item (2), 
for every subgraph $\mathcal{C}\in \mathcal{H}_0$, by the construction in Step 1, $\mathcal{C}$ contains a high degree node and all of its neighbors in $\mathcal{G}$. 
Thus, it remains to bound the augmented diameter of $\mathcal{C}$.

By construction, after step 3, $\mathcal{C}$ has hop-diameter at most $6$ (see Figure~\ref{fig:Type0}).
Since each node has weight at most $g\epsilon L_i$ and each edge has weight at most $\max\{\bar{w},L_i\} = L_i$, $\adm(\mathcal{C}) \leq 6L_i + 7\epsilon g L_i$. 
In step 4, $\mathcal{C}$ is augmented by trees of augmented diameter at most $\zeta L_i$ in a star-like structure via $\mst$ edges, which implies that:
\begin{equation*}
\adm(\mathcal{C}) \leq 6L_i + 7\epsilon g L_i + 2\bar{w} + 2\zeta L_i \leq 16 L_i
\end{equation*}
In the above equation, 
we use the fact that $\bar{w} \leq L_i$, $\epsilon \ll \frac{1}{g}$ and $\zeta < 1/2$.	\QED
\end{proof}

The purpose of constructing Type-$0$ clusters is to guarantee that in subsequent steps of the construction, the number of spanner edges incident to a node is $O(\epsilon^{-1})$ by Item (3) of Lemma~\ref{lm:Type-0}, which is useful for $d > 2$. 
When $d=2$, Lemma~\ref{lm:deg-K} implies that each node has at most $O(\epsilon^{-1})$ incident spanner edges, and hence, there is no need to construct Type-0 clusters.

Clearly, Type-0 clusters satisfy Invariant (I1) since $g = 42$, as the diameter of a subset of points is at most the diameter of its associated subgraph that spans the set of points.  
We remark that our construction of Type-0 clusters has not finished yet;  in the following steps of the construction  
we further augment Type-$0$ clusters by adding more subtrees of $\mathcal{T}_{i-1}$ via $\mst$ edges (see Lemma~\ref{lm:Invariant-I1}). However, this augmentation blows up the diameter of Type-0 clusters additively by at most $18 L_i$ and hence the final diameter bound is still in check.

Given that Type-0 clusters satisfy Invariant (I1), we can show that the credit lower bound invariant (I2) can be maintained while  all spanner edges incident to Type-0 clusters can be paid for by leftover credits (see Definition~\ref{def:leftover}).   Indeed, this follows from a more general lemma stated below.

\begin{lemma}\label{lm:Type0-payfor} 
	Let $\mathcal{C}$ be a cluster containing at least $\frac{2g}{\zeta \epsilon}$ nodes  such that $\adm(\mathcal{C}) \geq gL_i$, i.e, $\mathcal{C}$ satisfies Invariant (I1)   $\mathcal{C}$ satisfies Invariant (I1). Then the leftover credit of $\mathcal{C}$ after maintaining  Invariant (I2) can pay for their incident spanner edges and debt when $\ce  = \Omega(\epsilon^{-d})$.
\end{lemma}
\begin{proof}
Let $\mathcal{Y}\subseteq \mv(\mathcal{C})$ be a set of (arbitrary) $\frac{g}{\zeta\epsilon}$ nodes in $\mathcal{C}$. By Invariant (I2), it holds that: 
\begin{equation*}
	\cred(\mathcal{Y}) \geq  \frac{g}{\zeta\epsilon} \cdot \ce(\zeta L_{i-1}) \frac{g}{\zeta\epsilon} \cdot \zeta\epsilon\ce L_i ~=~ g\ce L_i ~\geq~ \ce \max\{\adm(\mathcal{C}), \zeta L_i\}
\end{equation*}
Thus, $\cred(\mathcal{Y})$ is sufficient to maintain Invariant (I2) of $\mathcal{C}$.

Let $\mathcal{Z}$ be a set of another $\frac{g}{\zeta\epsilon}$ nodes in $\mathcal{C}$. We redistribute $\cred(\mathcal{Z})$ to all nodes in $\mathcal{Y}\cup \mathcal{Z}$, each gets at least $\zeta\ce \epsilon L_i/2$ credits. Note that each node in $\mathcal{C}\setminus \{\mathcal{Y}\cup \mathcal{Z}\}$ has at least $\zeta\ce \epsilon L_i$ credits by Invariant (I2) for level $i-1$. That is, each $\epsilon$-cluster in $\mathcal{C}$ has at least $\zeta\ce \epsilon L_i/2$ leftover credits after maintaining Invariant (I2) for $\mathcal{C}$.  Since each node is incident to at most $O(\epsilon^{-d+1})$ level-$i$ edges, each of which has weight at most $L_i$, and has debt, if any, at most $O( \epsilon^{-1})L_i$, its leftover credit can pay for its incident spanner edges and debt if $\ce = \Omega(\epsilon^{-1}(\epsilon^{-d+1} + \epsilon^{-1})) = \Omega(\epsilon^{-d})$. \QED
\end{proof}

Recall that by Item (3) of  Lemma~\ref{lm:Type-0}, every tree  $\mathcal{T}\in \mathcal{F}_{i-1}$ has an augmented diameter at least $\zeta L_i$.

\paragraph{Step 1: Type-I clusters and contracted nodes.~} 
We say a node $x$ in a tree $T$ is \emph{$T$-branching} if it is incident to at least $3$ edges in $T$. When the tree $T$ is clear from the context, we simply call $x$ a \emph{branching} node. 
The construction in this step uses the following tree clustering lemma, whose proof is deferred to Section~\ref{subsec:Proof-Tree-Cluster}.

\begin{restatable}[Tree Clustering]{lemma}{treeclustering} \label{lm:tree-clustering} 
	Let $T$ be a tree with node and edge weights. Let $L, \beta, \eta, \gamma$ be parameters where $\eta \ll \gamma \ll 1$ and $\beta \geq 1$. Suppose that for any node $v\in T$ and any edge $e\in T$, $w(e) \leq w(v) \leq \eta L$ and $w(v)\geq (\eta L)/\beta$. There is a polynomial-time algorithm that finds a collection of node-disjoint  subtrees $\mathcal{F} = \{T_1,\ldots,T_k\}$ of $T$ such that:
	\begin{enumerate}[noitemsep]
		\item[(1)] $\adm(T_i) \leq 190\gamma L$ for each $1\leq i \leq k$.
		\item[(2)] Each branching node is contained in some tree in $\mathcal{F}$. 
		\item[(3)] Each tree $T_i$ contains a $T_i$-branching node $b_i$ and three paths $P_1,P_2,P_3$ that intersect only at $b_i$ and are otherwise node-disjoint,
		such that $\adm(P_1\cup P_2) = \adm(T_i)$ and $\adm(P_3 \setminus \{b_i\})= \Omega(\adm(T_i)/\beta)$. We call $b_i$ the \emph{center} of $T_i$.
		\item[(4)] Let $\widehat{T}$ be obtained by contracting each subtree of $\mathcal{F}$ into a single node. Then each $\widehat{T}$-branching node corresponds to a subtree  of augmented diameter at least $\gamma L$.
	\end{enumerate}
\end{restatable}

\begin{figure}[!htb]
	\center{\includegraphics[width=1.0\textwidth]{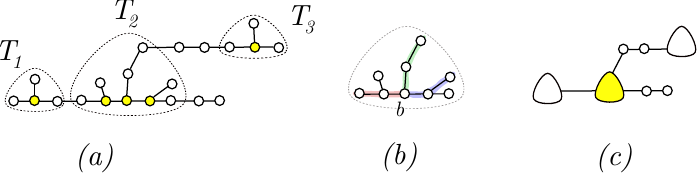}}
	\caption{(a) A tree $T$ and a set of trees $\mathcal{F}  = \{T_1,T_2,T_3\}$ as described in Lemma~\ref{lm:tree-clustering}; yellow nodes are branching vertices of $T$. (b) The tree $T_2\in \mathcal{F}$ and its center $b$; three paths $P_1,P_2,P_3$ are highlighted by three different colors. (c) The tree $\widehat{T}$ obtained from $T$ by contracting each tree in $\mathcal{F}$ into a single node; the contracted nodes have triangular shape. The branching (yellow) node of $\widehat{T}$ has augmented diameter at least $\gamma L$.}
	\label{fig:TC}
\end{figure} 

See Figure~\ref{fig:TC} for an illustration of Lemma~\ref{lm:tree-clustering}. Let $\cluster(T,L,\eta,\gamma,\beta)$ be the set of subtrees obtained by applying the construction of Lemma~\ref{lm:tree-clustering} to a tree $T$ with parameters $L,\eta,\gamma,\beta$. In the following lemma, we will apply the construction of Lemma~\ref{lm:tree-clustering} to each tree $\mathcal{T}\in \mathcal{F}_{i-1}$; note that  by Equation~\eqref{eq:weight-node} and Invariant (I1) for level $i-1$, each node has weight at least $\zeta \epsilon L_{i}$ and at most $g\epsilon L_i$. 

\begin{lemma}\label{lm:tree-credits} Let $\mathbb{U} = \cup_{\mathcal{T} \in \mathcal{F}_{i-1}}\cluster(\mathcal{T}, L_i, g\epsilon, \zeta, \frac{g}{\zeta})$.   Let $\mathcal{T}$ be any tree in $\mathbb{U}$. Then $\cred(\mathcal{T}) =  \ce \adm(\mathcal{T})+  \Omega(\ce \epsilon L_i |\mathcal{V}(\mathcal{T})|).$

\end{lemma}
\begin{proof} 
Let $\pzc[b]$ be the center of $\mathcal{T}$ and  $\mathcal{P}_1,\mathcal{P}_2, \mathcal{P}_3$ be three paths of $\mathcal{T}$ that intersect at $\pzc[b]$, as guaranteed by Item (3) of  Lemma~\ref{lm:tree-clustering}. Note that $\adm(\mathcal{T}) = \adm(\mathcal{P}_1 \cup \mathcal{P}_2)$.   Let $\mathcal{Q}_3 = \mathcal{P}_2\setminus \{\pzc[b]\}$.

Observe that $\sum_{\pzc[x] \in \mathcal{Q}_3}w(\pzc[x])  \geq \adm(\mathcal{Q}_3)/2$ since the edge weight is at most the vertex weight. Recall that $\zeta \epsilon L_i \leq w(\pzc[x]) \leq g\epsilon L_i$ for every node $\pzc[x] \in \mathcal{T}$; see Item (3) in Lemma~\ref{lm:prop-T}.  We have:
\begin{align*}
	 (|\mathcal{V}(\mathcal{Q}_3)|) &\geq \frac{\sum_{\pzc[x] \in \mathcal{Q}_3}w(\pzc[x])}{g\epsilon L_i}	  \geq \frac{\adm(\mathcal{Q}_3)}{2g \epsilon L_i} = \Omega(\frac{\adm(\mathcal{P}_1 \cup \mathcal{P}_2)}{2g \epsilon L_i}) \\
	&\geq \Omega(\frac{|\mv(\mathcal{P}_1 \cup \mathcal{P}_2)| \epsilon \zeta L_i}{2g\epsilon L_i}) = \Omega(|\mv(\mathcal{P}_1 \cup \mathcal{P}_2)|)
\end{align*}
This implies $|\mv(\mathcal{Q}_3)| = \Omega(|\mv(\mathcal{P}_1\cup \mathcal{P}_2\cup \mathcal{Q}_3)|)$.  Thus, it holds that:
\begin{align*}
	\sum_{\pzc[x]\in \mv(\mathcal{T})} w(\pzc[x]) + \sum_{e\in \me(\mathcal{T})}w(e) &\geq  \adm(\mathcal{P}_1 \cup\mathcal{P}_1) + \adm(\mathcal{Q}_3) + \sum_{\pzc[x] \in \mathcal{T}\setminus (\mathcal{P}_1\cup \mathcal{P}_2 \cup \mathcal{Q}_3)} w(\pzc[x])\\ 
	&  \geq \adm(\mathcal{T}) + \zeta \epsilon L_i (|\mathcal{V}(\mathcal{Q}_3)|) + \zeta \epsilon L_i(|\mathcal{V}(\mathcal{T}) - |\mv(\mathcal{P}_1\cup \mathcal{P}_2\cup \mathcal{Q}_3)|)\\
	& \geq \adm(\mathcal{T}) + \zeta \epsilon L_i (\Omega(|\mv(\mathcal{P}_1\cup \mathcal{P}_2\cup \mathcal{Q}_3)|)) +  \zeta \epsilon L_i(|\mathcal{V}(\mathcal{T}) - |\mv(\mathcal{P}_1\cup \mathcal{P}_2\cup \mathcal{Q}_3)|)\\
	&= \adm(\mathcal{T}) + \Omega(\epsilon L_i|\mv(\mathcal{T})|)
\end{align*}

This implies that $\cred(\mathcal{T})\geq \ce \adm(\mathcal{T}) + \Omega(\ce \epsilon L_i |\mv(\mathcal{T})|)$ as desired.\QED
\end{proof}

Let $\mathbb{U}$ be the collection of trees provided by Lemma~\ref{lm:tree-credits}.
Note that by Lemma~\ref{lm:tree-clustering}, the augmented diameter of every tree in $\mathbb{U}$ is at most $190\zeta L_i$;
a tree of $\mathbb{U}$ is said to have a {\em high} augmented diameter if its  augmented diameter is at least $\zeta L_i$.
In the following lemma, we form Type-I clusters from the high diameter trees of $\mathbb{U}$ and show that the remaining trees of $\mathbb{U}$ induce a special structure. 
A forest $F$ is called a \emph{linear forest} if every tree in $F$ is a path. 

\begin{lemma}[Type-I Clusters]\label{lm:Type-I}  Let $\mathbb{H}_1$ be the set of trees in  $\mathbb{U}$, called Type-I clusters, whose augmented diameter is at least $\zeta L_i$ and at most $190\zeta L_i$. Then, for any cluster $\mathcal{C}\in \mathbb{H}_1$, $$\cred(\mathcal{C}) \geq \ce(\adm(\mathcal{C}) + \Omega(L_i)).$$

Furthermore, let $\mathbb{U}^{-} = \mathbb{U}\setminus \mathbb{H}_1$ and  $\wmc[F]_{i-1}$ be the forest obtained from $\mathcal{F}_{i-1}$ by removing every node in $\mathbb{H}_1$ and contracting every tree in $\mathbb{U}^{-}$ into a single node. Then $\wmc[F]_{i-1}$ is a linear forest.
\end{lemma}
\begin{proof}
	 By Item (3) of Lemma~\ref{lm:tree-clustering}, $\cred(\mathcal{C}) \geq \ce (\adm(\mathcal{C}) + \Omega(\adm(\mathcal{C})\zeta/g)) = \ce(\adm(\mathcal{C}) + \Omega(L_i))$ as claimed.
	
	For any tree $\mathcal{T}\in \mathcal{F}_{i-1}$, let $\widehat{\mathcal{T}}$ be the tree obtained by contracting each subtree of $\mathcal{T}$ in $\mathbb{U}$ into a single node. By Item (4) of Lemma~\ref{lm:tree-clustering}, if we remove from $\widehat{\mathcal{T}}$ the contracted nodes corresponding to  subtrees of augmented diameter at least $\zeta L_i$, we obtain a linear forest; this implies the lemma.
	\QED
\end{proof}

We remark that some nodes of $\wmc[\mathcal{F}]_{i-1}$ might be nodes of $\mathcal{F}_{i-1}$, i.e, they are \emph{uncontracted nodes}. Additionally, in the construction of the following steps, Type-I clusters (defined in Lemma~\ref{lm:Type-I}) could be augmented further.  We call nodes of $\wmc[F]_{i-1}$ that are contracted from trees in $\mathbb{U}^{-}$ \emph{contracted nodes}. We say that a spanner edge is incident to a contracted node if it is incident to a node in the corresponding tree in $\mathbb{U}^-$.

\paragraph{Assigning weights and credits to contracted nodes of $\wmc[F]_{i-1}$.~} For each contracted node $\widehat{\pzc[x]}\in \wmc[F]_{i-1}$ and the corresponding tree $\mathcal{T}\in \mathcal{U}^{-}$, we assign:
\begin{equation}\label{eq:credit-x-hat}
	w(\widehat{\pzc[x]}) = \adm(\mathcal{T}) \quad \mbox{and} \quad \cred(\widehat{\pzc[x]}) = \ce \adm(\mathcal{T}).
\end{equation}

By Lemma~\ref{lm:tree-credits}, after the credit of a tree $\mathcal{T}$ in $\mathbb{U}^{-}$ is assigned to the corresponding contracted node, each node $\pzc[x] \in \mathcal{T}$ has $\Omega(\ce \zeta L_i) = \Omega(\epsilon\ce L_i)$ 
 leftover credits. By Lemma~\ref{lm:deg-K}, the leftover credit of $\pzc[x]$ can pay for all of its incident spanner edges when $\ce = \Omega(\epsilon^{-d})$. Thus, we can regard spanner edges incident to contracted nodes as \emph{paid edges}.

Let $E_{a}$ be spanner edges incident to a Type-$a$ cluster for $a = 0,1$. 
Let $E_{\cn}$ be the set of edges incident to contracted nodes (of $\wmc[F]_{i-1}$ in Lemma~\ref{lm:Type-I}). The set of remaining spanner edges that are incident to uncontracted nodes is denoted by $E_{\un}$. In the following step of the construction, we form new clusters in such a way that $\epsilon$-clusters have sufficient leftover credits to pay for edges in $E_{\un}$. The reason we focus on edges in $E_{\un}$ is because we argue later  in Lemmas~\ref{lm:Type-0},~\ref{lm:Type-I} and~\ref{lm:tree-credits} that edges in $E_0,E_1$ and $E_{\cn}$ can be paid for by using the leftover credit of one of its endpoints.

\paragraph{Step 2: Type-II clusters}

Let $\widehat{\pzc[x]}$ be a node of $\wmc[F]_{i-1}$. Let $\wmc[I](\widehat{\pzc[x]},r; \wmc[F]_{i-1})\subseteq \wmc[F]_{i-1}$ be a minimal subpath of a path, say $\wmc[P]$, of $\wmc[F]_{i-1}$ containing $\widehat{\pzc[x]}$ and all nodes of augmented distance at most $r$ from $\widehat{\pzc[x]}$; that is, for every node $\widehat{\pzc[y]} \in \wmc[I][\widehat{\pzc[x]},r; \wmc[F]_{i-1}]$, $\adm(\wmc[P][\widehat{\pzc[x]} ,\widehat{\pzc[y]}])\leq r$. 
  We say that $\widehat{\pzc[x]}$ is \emph{$r$-deep} if $\wmc[I](\widehat{\pzc[x]},r; \wmc[F]_{i-1})$ does not containcany of the two endpoints of $\wmc[P]$. We say that two nodes $\widehat{\pzc[x]},\widehat{\pzc[y]}$ are \emph{$r$-far} from each other if $\wmc[I](\widehat{\pzc[x]},r; \wmc[F]_{i-1})\cap  \wmc[I](\widehat{\pzc[y]},r; \wmc[F]_{i-1}) = \emptyset$. In particular, if $\widehat{\pzc[x]}$ and $\widehat{\pzc[y]}$ belong to different paths of $\wmc[F]_{i-1}$, then they are $r$-far from each other for any $r > 0$. We say that an edge $e \in E_{\un}$ is \emph{$r$-clusterable} w.r.t.\ $\wmc[F]_{i-1}$ if its two endpoints are $r$-deep and if they are $r$-far from each other.

 \begin{lemma}[Type-II Clusters]\label{lm:Type-II} 
 	We can construct a collection of  subgraphs $\mathbb{H}_2$ of $\mathcal{G}$, called Type-II clusters, such that
	for each  subgraph $\mathcal{C} \in \mathbb{H}_2$, we have:
 	\begin{enumerate}[nolistsep,noitemsep]
 		\item  $\mathcal{C}$  corresponds to two subpaths of $\wmc[F]_{i-1}$ connected by a $(2L_i)$-clusterable edge. 
 		\item $L_i/2 \leq \adm(\mathcal{C}) \leq 9L_i$.
 		\item $\cred(\mathcal{C}) \geq \ce \left(\adm(\mathcal{C})  + \Omega(L_i)\right)$.
 	\end{enumerate}
 Furthermore, 	let $\wmc[F]_{i-1}^{1}$ be the forest obtained from $\wmc[F]_{i-1}$ by removing every node in $\mathbb{H}_2$.  Then there is no $(2L_i)$-clusterable edge in $E_{\un}$ w.r.t $\wmc[F]_{i-1}^{1}$.
 \end{lemma}
\begin{proof}
	The construction is greedy.  
	
	\begin{wrapfigure}{r}{0.35\textwidth}
		\vspace{-25pt}
		\begin{center}
			\includegraphics[width=0.35\textwidth]{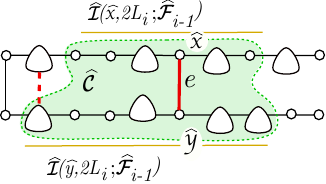}
		\end{center}
		\vspace{-20pt}
		\caption{\footnotesize{A cluster $\wmc[C]$ formed from 
				a $(2L_i)$-clusterable edge $e$ and the two subpaths of $\wmc[F]_{i-1}$ enclosed in the green-shaded region.}}
		\label{fig:Step3-opt}
	\end{wrapfigure}

	If there is an edge $e \in E_{\un}$ that is $(2L_i)$-clusterable w.r.t $\wmc[F]_{i-1}$, we form a new cluster $\widehat{\mathcal{C}} = \wmc[I](\widehat{\pzc[x]},2L_i; \wmc[F]_{i-1}) \cup  \wmc[I](\widehat{\pzc[y]},2L_i; \wmc[F]_{i-1}) \cup \{e\}$, where $\widehat{\mathpzc{x}}$ and $\widehat{\mathpzc{y}}$ are endpoints of $e$. (See Figure~\ref{fig:Step3-opt}.) Let $\mathcal{C}$ be obtained from $\widehat{\mathcal{C}}$ by uncontracting the contracted nodes; we then add $\mathcal{C}$ to $\mathbb{H}_2$.  Next, we remove all nodes in $\widehat{\mathcal{C}}$ from $\wmc[F]_{i-1}$ and repeat this step until it no longer applies. By the greedy nature of  the construction, when this step finishes, 
	 there is no $(2L_i)$-clusterable edge in $E_{\un}$ w.r.t the resulting $\wmc[F]_{i-1}^{1}$.

	Note  that we update $\widehat{\mathcal{F}}_{i-1}$ every time a new cluster is formed, so the set of $(2L_i)$-clusterable edges change accordingly. That is, a $(2L_i)$-clusterable edge may not be  $(2L_i)$-clusterable after removing some nodes of $\wmc[F]_{i-1}$.
	
	Item (1) follows directly from the construction. For Item (2), observe that, since $\mathcal{C}$ contains a level-$i$ edge $e$, $\adm(\mathcal{C})\geq w(e)\geq \frac{L_i}{1+\delta} ~\geq~ L_i/2$ since $\delta \leq 1$. 
	For the upper bound, we note that for any subpath $\wmc[P] \subseteq \wmc[F]_{i-1}$, the corresponding subtree, denoted by $\mathcal{P}$, obtained from $\wmc[P]$ by uncontracting the  contracted nodes will have $\adm(\mathcal{P}) \leq \adm(\wmc[P])$. (The weight function on the nodes of $\wmc[P]$ is defined in Equation~\eqref{eq:credit-x-hat}). Thus, it holds that
	\begin{equation*}
		\adm(\mathcal{C}) \leq  \adm(\wmc[I](\widehat{\pzc[x]},2L_i; \wmc[F]_{i-1})) + \adm(\wmc[I](\widehat{\pzc[y]},2L_i; \wmc[F]_{i-1})) + w(e) \leq 4L_i + 4L_i + L_i = 9 L_i~.
	\end{equation*}  
 	Item (2)   follows, and it thus remains to prove Item (3).

	Let $\wmc[D]$ be a path realizing the augmented diameter of $\wmc[C]$. Observe that $\adm(\wmc[D]) = \adm(\wmc[C]) \geq \adm(\mathcal{C})$. The analysis splits into two cases:
		\begin{itemize}
		\item Case 1: $\wmc[D]$ does not contain $e$ where $e$ is the $(2L_i)$-clusterable edge in $\mathcal{C}$. 
		Then, $\cred(\wmc[D]) \geq \ce \adm(\wmc[D])$ by the way we assign credits to nodes of $\wmc[F]_{i-1}$ (see Equation~\eqref{eq:credit-x-hat}). Since $\wmc[D]$ does not contain $e$, at least one endpoint of $e$, without loss of generality $\widehat{\pzc[x]}$, does not belong to $\wmc[D]$.  Thus, $\wmc[I](\widehat{\pzc[x]},2L_i; \wmc[F]_{i-1}) \cap \wmc[D] = \emptyset$. Recall that $\adm(\wmc[I](\widehat{\pzc[x]},2L_i; \wmc[F]_{i-1}))$ is a  subpath of a path of $\wmc[F]_{i-1}$ containing $\widehat{\pzc[x]}$ and all nodes of augmented distance at most $2L_{i}$ from $\widehat{\pzc[x]}$. Thus, by the minimality,  it holds that $$\adm(\wmc[I](\widehat{\pzc[x]},2L_i; \wmc[F]_{i-1})) ~\geq~ 4L_i - 2\bar{w} - 2(\zeta L_i) ~\geq~ 4L_i - 2L_i - L_i ~=~ L_i~; $$ where in the first inequality, we subtract the weight of two $\mst$ edges and two endpoint nodes of $\adm(\wmc[I](\widehat{\pzc[x]},2L_i; \wmc[F]_{i-1})$.	
	We thus have that $\cred(\wmc[I](\widehat{\pzc[x]},2L_i; \wmc[F]_{i-1})) \geq \ce \adm(\wmc[I](\widehat{\pzc[x]},2L_i; \wmc[F]_{i-1})) \ge \ce L_i$. It follows that $$\cred(\wmc[C]) ~\geq~ \cred(\wmc[D]) + \cred(\wmc[I](\widehat{\pzc[x]},2L_i; \wmc[F]_{i-1}))  ~\ge~ \ce (\adm(\wmc[D])  + L_i) ~\geq~ \ce  (\adm(\mathcal{C})  + L_i) .$$
		
		\item Case 2: $\wmc[D]$ contains $e$. Let  $\wmc[P]_1, \wmc[P]_2$ be the two subpaths of $\wmc[I](\widehat{\pzc[x]},2L_i; \wmc[F]_{i-1})$ sharing the same endpoint $\widehat{\pzc[x]}$ such that $\wmc[P]_1 \cup  \wmc[P]_2 = \wmc[I](\widehat{\pzc[x]},2L_i; \wmc[F]_{i-1})$. We define the two subpaths  $\wmc[Q]_1, \wmc[Q]_2$ of $\wmc[I](\widehat{\pzc[y]},2L_i; \wmc[F]_{i-1})$ sharing the same endpoint $\widehat{\pzc[y]}$ such that $\wmc[Q]_1 \cup  \wmc[Q]_2 = \wmc[I](\widehat{\pzc[y]},2L_i; \wmc[F]_{i-1})$. 
		Observe that at least two paths among four paths $\wmc[P]_1, \wmc[P]_2, \wmc[Q]_1, \wmc[Q]_2$,~w.l.o.g $\wmc[P]_2$ and $\wmc[Q]_2$, that only intersect $\wmc[D]$ at their endpoints. 
		Note that $e$ has length at most $L_i$ and has no credit; recall that only $\mst$ edges have credits. We assign the credit of $\wmc[P]_2\setminus \{\pzc[x]\}$ to  $e$.  Since $e$ is $(2L_i)$-clusterable, $\adm(\wmc[P]_2) \geq 2L_i - \bar{w} - (\zeta L_i)$ 
		and the endpoints of $e$ are uncontracted nodes. Thus, $$\adm(\wmc[P]_2\setminus \{\pzc[x]\}) ~\geq~ 2L_i - 2\bar{w} - g\epsilon L_i - (\zeta L_i) ~\geq~ 2L_i - (g+2)\epsilon L_i - \zeta L_i ~\geq~ L_i$$ when $\epsilon \ll \frac{1}{g}$ and $\zeta < 1/2$; here we use the fact that $\bar{w} \leq L_{i-1} = \eps L_i$. Thus, $e$ is assigned at least $\ce L_i \geq \ce w(e)$ credits, implying that:
		\begin{equation*}
			\cred(\wmc[D]) = \sum_{\pzc[x] \in \mv(\mathcal{D})} \cred(\pzc[x]) + \sum_{e'  \in \me(\mathcal{D})} \cred(e')\geq \ce \adm(\wmc[D]). 
		\end{equation*}
		Observe that $\adm(\wmc[Q]_2\setminus \{x\}) \geq L_i$ (by the same reason that  $\adm(\wmc[P]_2\setminus \{x\}) \geq L_i$ shown above). 
		Thus, $\cred(\wmc[Q]_2\setminus \{x\}) \geq \ce L_i$; this implies $$\cred(\wmc[C]) \geq \cred(\wmc[D]) + \ce L_i \geq \ce(\wmc[D] + L_i) \geq  \ce(\adm(\mathcal{C})  + L_i) .$$
	\end{itemize}
	
	In both cases, we have shown that $\cred(\mathcal{C}) \geq \ce (\adm(\mathcal{C}) + L_i)$. Observe that by the way we assign weights to contracted nodes of $\mathcal{C}$, we have:
	\begin{equation*}
		\sum_{\pzc[x] \in\mv( \mathcal{C})}w(\pzc[x]) + \sum_{e'\in \me( \mathcal{C})}w(e') \geq 	\sum_{\widehat{\pzc[x]} \in\mv( \wmc[C])}w(\widehat{(\pzc[x]}) + \sum_{e'\in \me( \wmc[C])}w(e'),
	\end{equation*}
	 and hence, $\cred(\mathcal{C}) \geq \cred(\wmc[C]) \geq (\adm(\mathcal{C}) + L_i)$
	as claimed.\QED
\end{proof}

\paragraph{Step 3: Augmenting existing clusters and constructing Type-III clusters.~} Let $\wmc[P]$ be a path in $\wmc[F]^{1}_{i-1}$ (in Lemma~\ref{lm:Type-II}). We consider two cases:

\begin{itemize}
	\item \textbf{Case 1: $\adm(\wmc[P]) < 8L_i$.~}   Let $\mathcal{P}$ be obtained from $\wmc[P]$ by uncontracting the contracted nodes. Let $\mathcal{C}$ be a cluster in $\mathbb{H}_0\cup \mathbb{H}_1\cup \mathbb{H}_{2}$ that $\mathcal{P}$ is connected to via an $\mst$ edge $e$. Such an edge $e$ must exist since $\mathcal{T}_{i-1}$ is a spanning tree of $\mathcal{G}$.  We augment $\mathcal{C}$ by adding $\mathcal{P}$ and $e$ to $C$.  
	\item  \textbf{Case 2: $\adm(\wmc[P]) \geq 8L_i$.~}   Let $\widehat{\mathcal{P}}_1, \widehat{\mathcal{P}}_2$ be the two minimal prefix/suffix subpaths of $\widehat{\mathcal{P}}$, each of augmented diameter at least $2L_i$. Let $\mathcal{P}_j$ be obtained from $\widehat{\mathcal{P}}_j$ by uncontracting the contracted nodes, $j \in \{1,2\}$; $P_j$ is not necessarily a path anymore (due to uncontraction). 
	 If $\mathcal{P}_j$, for each $j \in \{1,2\}$, is connected to a cluster $\mathcal{C} \in \mathbb{H}_0\cup \mathbb{H}_1\cup \mathbb{H}_{2}$ via an $\mst$ edge $e$, we augment $\mathcal{C}$ by adding   $\mathcal{P}_j$ and $e$ to $\mathcal{C}$. ($\mathcal{C}$ is arbitrarily chosen among clusters that $\mathcal{P}_j$ is connected to.)
	 Otherwise, we form a \emph{Type-III cluster} from $\mathcal{P}_j$ (see Figure~\ref{fig:Type3}). Let $\mathbb{H}_3$ be the set of all Type-III clusters.
\end{itemize}

\begin{figure}[!htb]
	\center{\includegraphics[width=0.9\textwidth]{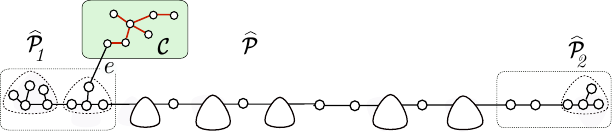}}
	\caption{Two minimal prefix/suffix subpaths $\wmc[P]_1, \wmc[P]_2$ of a path $\wmc[P]$, and the corresponding trees $\mathcal{P}_1,\mathcal{P}_2$ obtained by uncontracting their contracted nodes (which have triangular shape  in the figure). The subpath $\mathcal{P}_1$ is augmented to $\mathcal{C}$ since it has an MST edge to $\mathcal{C}$; there could be multiple such clusters, and we choose one of them arbitrarily. $\mathcal{P}_2$ becomes a Type-III cluster.}
	\label{fig:Type3}
\end{figure} 

We apply the construction in both cases above to each path of $\wmc[F]^{1}_{i-1}$. Let $\wmc[F]^2_{i-1}$ be obtained from $\wmc[F]^{1}_{i-1}$ by removing nodes that are augmented to an existing cluster or grouped to Type-III clusters.

Unlike Type-$a$ clusters for $a \in \{0,\mbox{I}, \mbox{II}\}$, we cannot guarantee that nodes in Type-III clusters have non-zero leftover credits. However, we can show later that any Type-III cluster, say $\mathcal{C}$, would be a leaf of the level-$i$ cluster tree $\mathcal{T}_i$ in Invariant (I3) (the formal proof is provided in Lemma~\ref{lm:leaves-Ti}), and hence the weight of level-$i$ edges (in $E_{\un}$) incident to nodes in $\mathcal{C}$ can be deposited to the debt account of $\mathcal{C}$. Recall that only leaves of $\mathcal{T}_i$ have such a debt account by Invariant (I3)
 The challenge then is to argue that the total debt of $\mathcal{C}$ is in check as imposed by Invariant (I3).

\begin{lemma}\label{lm:structure-F2} (1) Each path $\wmc[P]\in \wmc[F]^2_{i-1}$ has  $ 2L_i \leq \adm(\wmc[P]) \leq 3L_i$, and (2) there is no edge in $E_{\un}$ connecting two different paths of $\wmc[F]^2_{i-1}$.
\end{lemma}
\begin{proof} The lower bound $2L_i \leq \adm(\wmc[P])$  follows directly from the construction. Note that each prefix/suffix of the path in Case 2 has augmented diameter at most $2L_i$  plus the weight of an MST edge and a node by the minimality. Thus, 
	$\adm(\wmc[P]) ~\leq~ 2L_i + \bar{w} + \zeta L_i \leq 2L_i + (\zeta+\epsilon) L_i < 3L_i$ when $\epsilon \leq \frac{1}{2}$ and $\zeta < 1/2$; this implies Item (1). 
	
For Item (2), if there were such an edge in $E_{\un}$, then, due to the construction in Case 2 of Step 3, it would be a $(2L_i)$-clusterable edge in $\wmc[F]^1_{i-1}$, contradicting Lemma~\ref{lm:Type-II}.	\QED
\end{proof}

\paragraph{Step 4: Type-IV clusters.~} By Lemma~\ref{lm:structure-F2}, for any path $\wmc[P] \in \wmc[F]^2_{i-1}$ and any edge $e\in E_{\un}$  incident to a node in $\wmc[P]$, either both endpoints of $e$ belong to $\wmc[P]$ or the other endpoint of $e$ (the one not in $\wmc[P]$) belongs to a Type-$a$ cluster for some $a\in \{0,\mbox{I},\mbox{II},\mbox{III}\}$. 
In the latter case, we can pay for $e$ using the leftover credit of the other endpoint of $e$ not in $\wmc[P]$. In the former case, where
both endpoints of $e$ belong to $\wmc[P]$,
we form a new Type-IV cluster in such a way that the endpoints of $e$ have leftover credits to pay for $e$.

\begin{lemma}[Type-IV Clusters]\label{lm:Type-IV} Let $\wmc[P]$ be a path in $\wmc[F]^2_{i-1}$ and $E_{\un}(\wmc[P])$ be the edges of $E_{\un}$
with both endpoints in $\wmc[P]$. 
Let $\mathcal{P}$ be obtained from $\wmc[P]$ by uncontracting its contracted nodes. Let $\mathcal{W}$ be the set of nodes in $\mathcal{P}$ that correspond to the contracted nodes of $\wmc[P]$.  We can construct a set of clusters, denoted by $\mathbb{C}_4(\wmc[P])$, such that: 	
	\begin{enumerate}
		\item  The clusters in $\mathbb{C}_4(\wmc[P])$ contain every node in $\mathcal{P}$.
		\item For every cluster $\mathcal{C}\in \mathbb{C}_4(\wmc[P])$, $\zeta L_i \leq \adm(\mathcal{C})\leq 20L_i$. Furthermore, $\mathcal{C}$ is a subtree of $\mathcal{P}$ and some edges in  $E_{\un}(\wmc[P])$  whose both endpoints are in $\mathcal{C}$.
		\item  There is an orientation of edges in $E_{\un}(\wmc[P])$  
		such that if the total number of out-going edges incident to nodes in a cluster $\mathcal{C}\in  \mathbb{C}_4(\wmc[P])$ is $t$ for any $t\geq 0$, then:
\begin{equation}\label{eq:TypeIV-credit}
	 \cred(\mathcal{C}) \geq \ce\left(\adm(\mathcal{C}) + \Omega(t\epsilon + |\mathcal{C}\cap \mathcal{W}|) \epsilon L_i\right) 
\end{equation}
		\end{enumerate}
	Clusters in $\mathbb{H}_4 = \cup_{\wmc[P]\in \wmc[F]^2_{i-1}}\mathbb{C}_4(\wmc[P])$ are called Type-IV clusters.
\end{lemma}

\begin{figure}[!htb]
	\center{\includegraphics[width=0.9\textwidth]{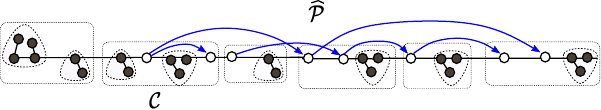}}
	\caption{The path $\wmc[P]$, a cluster $\mathcal{C}$, and a set of (blue) edges in $E_{\un}(\wmc[P])$. White nodes are uncontracted nodes and black nodes are those in contracted nodes. 
		Edges in $E_{\un}(\wmc[P])$ are only incident to uncontracted (white) 
		nodes and are oriented during the construction of Type-IV clusters. (The construction of the orientation is given in Subsection~\ref{subsec:Type-IV}.) 
		 The amount of leftover credit $\mathcal{C}$ has is proportional to the number of (black) nodes in $\mathcal{C}\cap \mathcal{W}$ and the number of out-going edges from nodes in $\mathcal{C}$; there could be edges with both endpoints in $\mathcal{C}$.}
	\label{fig:Type4}
\end{figure} 

In Subsection~\ref{subsec:Type-IV}, we give the details of the construction of Type-IV clusters and the orientation of edges in $E_{\un}(\wmc[P])$ for the path $\wmc[P]$ in Lemma~\ref{lm:Type-IV}. 

Intuitively, Type-IV clusters are constructed from each path $\wmc[P]$ of $\wmc[F]^2_{i-1}$ separately.  See Figure~\ref{fig:Type4} for an illustration. Edges in $E_{\un}(\wmc[P])$ will be oriented along with our construction of Type-IV clusters in such a way that Equation~\eqref{eq:TypeIV-credit} holds.  While it is not hard to see that the diameter bounds --- Item (2) --- follow directly from the construction ( in Subsection~\ref{subsec:Type-IV}), 
it is highly nontrivial to lower bound the amount of leftover credit that $\mathcal{C}$ has as stated in Item (3).  This is the crux of the argument. An interesting special case is when $\mathcal{C}\cap \mathcal{W} = \emptyset$, that is, $\mathcal{C}$ contains no node in the contracted nodes of $\wmc[P]$. Item (3) means that, in this special case, the more edges outgoing from $\mathcal{C}$, the more leftover credit $\mathcal{C}$ will have. That is, Item (3) provides the formal condition behind Insight (2) that was discussed in Section~\ref{subsec:our-high-lv}.

Lemma~\ref{lm:Type-IV} completes the construction of level-$i$ clusters. Table~\ref{table:cluster-summary} summarizes important properties of each type of clusters in our construction.

\begin{table}[h!]
	\centering
	\begin{tabular}{| l | p{13.5cm} |}
		\hline \\[-1em]
		\textbf{Cluster Type} & \textbf{Important Property} \\ \hline \\[-1em]
		Type-0 &  Every high degree node and its neighbors are grouped into Type-0 clusters.\\
		\hline \\[-1em]
		Type-I &   Each Type-I cluster is a subtree of $\mathcal{T}_{i-1}$ of augmented diameter at least $\zeta L_i$ (and at most $g L_i$) and has three internally node-disjoint paths of augmented diameter $\Omega(L_i)$ each. \\
		\hline \\[-1em]
		Type-II &   Each Type-II cluster is of the form $ \wmc[I](\widehat{\pzc[x]},2L_i; \wmc[F]_{i-1}) \cup  \wmc[I](\widehat{\pzc[y]},2L_i; \wmc[F]_{i-1}) \cup \{e\}$ where $e$ is a $(2L_i)$-clusterable edge.  \\
		\hline \\[-1em]
		Type-III &  Each Type-III cluster corresponds to a prefix/suffix subpath of a path $\wmc[P]$ in the forest $\wmc[F]^1_{i-1}$. \\
		\hline \\[-1em]
		Type-IV &  Each Type-IV cluster $\mathcal{C}$ has an amount of leftover credit  proportional to the number of out-going edges in the orientation of $E_{\un}(\wmc[P])$. \\
		\hline
	\end{tabular}
	\caption{Four types of clusters and their important properties.}
	\label{table:cluster-summary}
\end{table}

\subsubsection{Cluster Invariants and a proof of Lemma~\ref{lm:weight-reduction}}
We now argue that level-$i$ clusters satisfy all three invariants while level-$i$ spanner edges can be paid for by leftover credits.

\begin{lemma}\label{lm:Invariant-I1}  $\adm(\mathcal{C}) \leq 34 L_i$ for any cluster  $\mathcal{C} \in \cup_{i=0}^{4}\mathbb{H}_i$.
\end{lemma}
\begin{proof}
Recall that the maximum weight of an $\mst$ edge (after subdivision) is at most $\bar{w} \leq  L_0 \leq \epsi L_i$ for $i\geq 1$.  We observe that, since the augmentation of Type-$a$ clusters for $a \in \{0, \mbox{I}, \mbox{II}\}$ in Step 3  is via $\mst$ edges in a star-like way, the diameter of the augmented clusters increases by at most $ 2\bar{w} + 2\cdot 8L_i \leq 18L_i$. Thus, by Lemmas~\ref{lm:Type-0},~\ref{lm:Type-I}, and~\ref{lm:Type-II}, the diameters  of clusters formed in Step 0, Step 1 and Step 2 clusters are at most $16L_i, 190\gamma L_i$ and $9L_i$, respectively; note that $190\gamma < 1$. Thus, it holds that:
\begin{equation*}
\adm(\mathcal{C}) \leq 16L_i + 18L_i = 34 L_i
\end{equation*} 

By the construction in Step 3, each Type-III cluster is a minimal suffix/prefix of a path $\widehat{\mathcal{P}}$ of augmented diameter at least $2L_i$. Since each node of  $\widehat{\mathcal{P}}$ has weight at most $\zeta L_i$ by the construction in Step 1 (Lemma~\ref{lm:Type-I}), and each edge has weight at most $\bar{w}$, a Type-III cluster has augmented diameter at most $2L_i +   \bar{w} + \zeta L_i ~\leq~ 4L_i$. Type-IV clusters have augmented diameter at most $20L_i$ by Lemma~\ref{lm:Type-IV}. Thus, in every case, Lemma~\ref{lm:Invariant-I1} holds.\QED
\end{proof}

By Lemma~\ref{lm:Invariant-I1}, Invariant (I1) is satisfied. We now focus on Invariants (I2) and (I3). By construction, except for Type-III clusters, nodes in clusters of any other type have positive leftover credits. We will show in Lemma~\ref{lm:leaves-Ti} that Type-III clusters correspond to leaves of the cluster tree $\mathcal{T}_i$ at level $i$. Hence, by Invariant (I3), Type-III clusters are allowed to have debt of $4g^2\zeta^{-2}\epsilon^{-2}\sum_{j-1}^i L_j$. We regard this debt as the \emph{debt credit} of each Type-III cluster. We then argue that we can pay for level-$i$ spanner edges incident to any Type-III cluster using the debt credit of that cluster.

\paragraph{Constructing the cluster tree $\mathcal{T}_i$.~} Recall that every level-$i$ is associated with a subgraph of $\mathcal{G}$,  and that  $\mathcal{T}_{i-1}$ is a spanning tree of $\mathcal{G}$.  We contract each level-$i$ cluster into a single node, then  $\mst$ edges in the resulting graph induce a connected spanning subgraph. Let $\mathcal{T}_i$ be an arbitrary spanning tree of the resulting graph that only contains $\mst$ edges.

\begin{lemma}\label{lm:leaves-Ti}
Type-III clusters are leaves of $\mathcal{T}_i$. 
\end{lemma}
\begin{proof}
	By Case 2 of Step 3 in the construction, each Type-III cluster $\mathcal{C}$ corresponds to a suffix/prefix of a path $\wmc[P]$, and that cluster $\mathcal{C}$ is not connected by any $\mst$ edge to other clusters before the construction of Step 4; indeed, otherwise $\mathcal{C}$ would be augmented to another cluster via $\mst$ edges. Since the remaining subpath of $\wmc[P]$ after the construction in Step 3 will be grouped via Type-IV clusters by Lemma~\ref{lm:Type-IV}, each Type-III cluster is connected to a Type-IV cluster by an $\mst$ edge, and is thus a leaf of $\mathcal{T}_i$. \QED
\end{proof}

We are now ready to show that Invariants (I2) and (I3) are satisfied 
and prove that, in addition, level-$i$ spanner edges can be paid for.  
\begin{lemma}\label{lm:I2-main}  Let $\ce = \Omega(\epsilon^{-d})$ for $d\geq 2$. If every cluster $\mathcal{C} \in \cup_{i=0}^{4}\mathbb{H}_i$ takes exactly $\ce\max\{\adm(\mathcal{C}),\zeta L_i\}$ credits from nodes and $\mst$ edges in $\mathcal{C}$, then:
	\begin{enumerate}[nolistsep,noitemsep]
		\item Every level-$i$ spanner edge can be paid for by either leftover credit or debt credit (of Type-III clusters).
		\item Every leaf node of $\mathcal{T}_{i-1}$ in a level-$i$ cluster $\mathcal{C}$ can pay for its debt by using either its leftover credit or its debt credit. In the latter case, $\mathcal{C}$ is a Type-III cluster.  
	\end{enumerate}
Furthermore, the total debt of each leaf of $\mathcal{T}_i$ is at most $4g^2\zeta^{-2}\epsilon^{-2}\sum_{j=1}^i L_j$.
\end{lemma}
\begin{proof}
By Lemma~\ref{lm:Type0-payfor}, nodes in Type-0 clusters, and more generally, clusters with at least $2g(\zeta\epsilon)^{-1}$ nodes, can maintain Invariant  (I2) and pay for incident spanner edges and debt.  Thus, it remains to consider   clusters with at most $2g(\zeta\epsilon)^{-1}$ nodes. Note that every node not in a Type-$0$ cluster is incident to at most $2g\zeta^{-1}\epsilon^{-1}$ edges by Item (3) in Lemma~\ref{lm:Type-0}  and has debt at most $O(\epsilon^{-1} L_i)$ by Item (3) in  Lemma~\ref{lm:prop-T}.  Thus, the total number of incident level-$i$ spanner edges and the total debt of nodes  in a cluster considered henceforth are $4g^2(\zeta\epsilon)^{-2}$ and $O(\epsilon^{-2} L_i)$, respectively. We consider three cases.

\begin{figure}[!htb]
	\center{\includegraphics[width=0.9\textwidth]{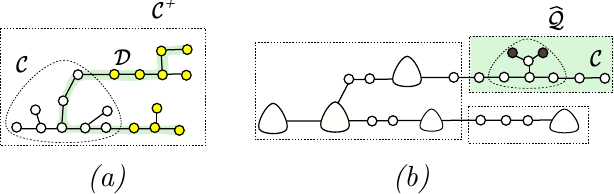}}
	\caption{(a) A Type-II cluster $\mathcal{C}$. $\mathcal{C}^{+}$ is obtained from $\mathcal{C}$ by the augmentation in Step 3 via $\mst$ edges that adds subtrees of yellow vertices to $\mathcal{C}$; the path $\mathcal{D}$  realizing the diameter of $\mathcal{C}^+$ is highlighted green. (b) A cluster $\mathcal{C}$ is obtained from $\wmc[Q]$ by uncontracting the contracted nodes. $\mathcal{C}$ has two leaves, while only one leaf node is white;  the other two leaf nodes are black since they belong to a contracted node. }
	\label{fig:leaf2}
\end{figure} 

\hspace{0.2cm}

\noindent \textbf{Case 1: Type-I and Type-II clusters.~} Let $\mathcal{C}$ be a Type-I or Type-II cluster formed in Lemma~\ref{lm:Type-I} or Lemma~\ref{lm:Type-II}, respectively. Let $\mathcal{C}^{+}$ be obtained by the augmentation of $\mathcal{C}$ in Step 3 (see Figure~\ref{fig:leaf2}(a)).  Let $\mathcal{D}$ be the path realizing the diameter of $\mathcal{C}^+$. Observe that the augmentation is in a star-like way, and by adding subtrees of $\mathcal{F}_{i-1}$ to $\mathcal{C}$ via $\mst$ edges. Thus, $\mathcal{D}' = \mathcal{D}\cap \mathcal{C}$ is a path. We take exactly $\ce \adm(\mathcal{D}')$ credits from $\mathcal{C}$ and all the credit of nodes and ($\mst$) edges of $\mathcal{D}\setminus \mathcal{D}'$; thus the total credit taken is:
\begin{equation*}
\ce \adm(\mathcal{D}') + \ce \left(\sum_{e\in E(\mathcal{D})\setminus E(\mathcal{D}')}w(e) +  \sum_{\pzc[x]\in \mathcal{V}(\mathcal{D})\setminus \mathcal{V}(\mathcal{D}')}w(\pzc[x])\right) \geq \ce \adm(\mathcal{D})  = \ce \max\{\adm(\mathcal{D}), \zeta L_i\}
\end{equation*}
Since we take  $\ce \adm(\mathcal{D}') \leq \ce \adm(\mathcal{C})$ credits from $\mathcal{C}$, by Lemma~\ref{lm:Type-I} and Lemma~\ref{lm:Type-II},  the total leftover credit of nodes and $\mst$ edges in $\mathcal{C}$ is $\Omega(\ce L_i)$. Since nodes in $\mathcal{C}^+$ are incident to at most $4g^2(\zeta\epsilon)^{-2}$ edges (of length at most $L_i$ each) and have at most $O(\epsilon^{-2}L_i)$ total debt,  the leftover credit is sufficient to pay for these edges and debt when $\ce = \Omega(\epsilon^{-2})$.

\hspace{0.2cm}

\noindent  \textbf{Case 2: Type-III clusters.~} Let $\mathcal{C}$ be a Type-III cluster and $\widehat{\mathcal{Q}}$ be a subpath of $\wmc[F]^1_{i-1}$ corresponding to $\mathcal{C}$ in the construction of Step 3; $\mathcal{C}$ is obtained from  $\widehat{\mathcal{Q}}$ by uncontracting the contracted nodes. Observe that $\cred(\wmc[Q]) \geq \ce \adm(\wmc[Q]) \geq \ce \adm(\mathcal{C})$;	thus, $\mathcal{C}$ can take exactly $\ce \adm(\mathcal{C})$ credits from nodes and $\mst$ edges in $\mathcal{C}$ to maintain Invariant (I2).

Since $\mathcal{C}$ is incident to at most $4g^2(\zeta\epsilon)^{-2}$ level-$i$  spanner edges of weight at most $L_i$ each,  and it is a leaf of $\mathcal{T}_i$ by Lemma~\ref{lm:leaves-Ti}, we can pay for these edges using the debt credit of $\mathcal{C}$; by Invariant (I3), $\mathcal{C}$ has $\sum_{j=1}^{i} 4g^2(\zeta\epsilon)^{-2} L_j$ debt credits. Thus, we have:
 
\begin{observation}\label{obs:debt-C}
	 The leftover debt credit of $\mathcal{C}$ after paying for incident spanner edges is $\sum_{j=1}^{i-1} 4g^2(\zeta\epsilon)^{-2} L_j$, which is equal to the debt of \emph{exactly one} node in $\mathcal{C}$.
\end{observation}

Next, we pay for the debt of nodes in $\mathcal{C}$ (if any), and  there could be many of them who have non-zero debt. We have noted that the remaining debt credit of $\mathcal{C}$ can only pay for the debt of exactly one node  in $\mathcal{C}$; our goal is to show that the debt of other nodes can be paid for by other means. Note that $\mathcal{C}$ is a subtree of $\mathcal{F}_{i-1}$ and that only leaves of $\mathcal{F}_{i-1}$ have debt by Invariant (I3).  We say that a leaf node $\mathcal{F}_{i-1}$ in $\mathcal{C}$ is \emph{black} if it is contained in the subtree corresponding to a contracted node in $\widehat{\mathcal{Q}}$; otherwise it is {\em white} (see Figure~\ref{fig:leaf2}(b)). We claim that:
\begin{claim}\label{clm:structure-TypeIII}
	Every leaf of $\mathcal{F}_{i-1}$ in $\mathcal{C}$ is black, except for one node, which is the (uncontracted) endpoint node of $\widehat{\mathcal{Q}}$.
\end{claim}
\begin{proof}
Recall that  $\widehat{\mathcal{Q}}$ is a suffix/prefix subpath of a path $\wmc[P]$ in Step 3. Thus, the only (uncontracted) node of $\widehat{\mathcal{Q}}$ that is a leaf of $\mathcal{F}_{i-1}$ is the endpoint node. Thus, other leaves in $\mathcal{C}$ must be in contracted nodes of $\widehat{\mathcal{Q}}$; this implies the claim. 
\QED
\end{proof}

Let $\pzc[x]$ be a black node in $\mathcal{C}$, and $\mathcal{T}$ be the subtree of $\mathcal{F}_{i-1}$ corresponding to the contracted node  that contains $\pzc[x]$. Let $\mathcal{D}$ be a  path realizing the diameter of $\mathcal{C}$.  Note that $\mathcal{C}$ is a subtree of $\mathcal{F}_{i-1}$; this implies that $\mathcal{D}_{\mathcal{T}} = \mathcal{D}\cap \mathcal{T}$ is a  simple path. Since we take exactly $\ce \adm(\mathcal{C}) = \ce \adm(\mathcal{D})$ credits from nodes and edges of $\mathcal{C}$, for each subtree  $\mathcal{T}$  corresponding to a contracted node in $\wmc[Q]$, we take exactly $\ce \adm(\mathcal{D}_{\mathcal{T}})$ credits from $\mathcal{T}$.  By Lemma~\ref{lm:tree-credits}, each node in $\mathcal{T}$ has at least $\Omega(\ce \epsilon L_i)$ credits left. In particular, $\pzc[x]$ has $\Omega(\ce \epsilon L_i)$ credits as leftover. Since  $\pzc[x]$ has at most $O(\epsilon^{-1} L_i)$ debt by Item (2) in Lemma~\ref{lm:prop-T}, the leftover credit of $\pzc[x]$ is sufficient to pay for the debt when $\ce = \Omega(\epsilon^{-2})$. Thus, we only need to pay for the debt of the (only one) non-contracted node in Claim~\ref{clm:structure-TypeIII}, which can be done by using the leftover debt credit of $C$ by Observation~\ref{obs:debt-C}.
\hspace{0.2cm}

\noindent  \textbf{Case 3: Type-IV clusters.~} Let $\mathcal{C}$ be a Type-IV cluster. By  Item (3) of Lemma~\ref{lm:Type-IV}, each node $\pzc[x]\in \mathcal{C}$ in a contracted node has $\Omega(\ce\epsilon L_{i})$ credits. Thus, it can pay for (at most $O(\eps^{-1})$) incident spanner edges and its debt (of at most $O(\epsilon^{-1}L_i)$ by Lemma~\ref{lm:prop-T}) when $\ce = \Omega(\epsilon^{-d})$. This implies that unpaid edges incident to $\wmc[F]^{2}_{i-1}$ are incident to uncontracted nodes only; these edges are in $E_{\un}$. By Lemma~\ref{lm:structure-F2}, unpaid edges have both endpoints in the same path of $\wmc[F]^2_{i-1}$. Let $\widehat{\mathcal{P}}$ be such a path, and $E_{\un}(\widehat{\mathcal{P}})$ be the set of unpaid edges incident to nodes in  $\widehat{\mathcal{P}}$; we use the notation used in Lemma~\ref{lm:Type-IV} here. Note that the construction of Type-IV clusters is applied to each path of $\wmc[F]^2_{i-1}$ separately, and that uncontracted nodes of  $\widehat{\mathcal{P}}$ have no debt since they are not leaves of $\mathcal{T}_{i-1}$.

Following the notation used in Lemma~\ref{lm:Type-IV}, each Type-IV cluster  in $\mathbb{C}_4(\wmc[P])$ pays for its out-going edges in $E_{\un}(\wmc[P])$. (The orientation of edges in  $E_{\un}(\wmc[P])$ is from Item (3) of Lemma~\ref{lm:Type-IV}.)
 Since there are $t$ incident out-going edges (of total weight $tL_i$) while the leftover credit is $\Omega(t \ce \epsilon^2L_{i})$, the leftover credit is sufficient to pay for the out-going edges when $\ce = \Omega(\epsilon^{-2})$.

In summary, after guaranteeing Invariants (I2) and (I3), every level-$i$ spanner edge can be paid for by leftover and debt credits when $\ce = \max\{\epsilon^{-d},\epsilon^{-2}\} = \Omega(\epsilon^{-d})$. 	\QED
\end{proof}

\paragraph{Proof of Lemma~\ref{lm:weight-reduction}.~}  Lemmas~\ref{lm:Invariant-I1},~\ref{lm:leaves-Ti} and~\ref{lm:I2-main} imply that, by choosing $\ce = \Theta(\epsilon^{-d})$, we can maintain Invariants (I1), (I2) and (I3) for level-$i$ clusters and, at the same time, pay for all  level-$i$ spanner edges incident to $\epsilon$-cluster using the leftover credits.  Thus,  $w(S^j) \leq \ce w(\mst) = O(\epsilon^{-d} w(\mst))$.\QED

\subsubsection{The construction of Type-IV clusters}\label{subsec:Type-IV}

In this section, we construct a set of Type-IV clusters as claimed in Lemma~\ref{lm:Type-IV}. The construction is applied to each path $\widehat{\mathcal{P}}\in \wmc[F]^2_{i-1}$ separately. By Lemma~\ref{lm:structure-F2}, $\adm(\widehat{\mathcal{P}}) \geq 2L_i$. The construction has two steps.

\paragraph{Step 1: tiny clusters.~} We greedily break $\widehat{\mathcal{P}}$ into subpaths, each of augmented diameter at least $\zeta L_i$ and at most $3\zeta L_i$; note that each edge of  $\widehat{\mathcal{P}}$ has weight at most $\bar{w} \leq \epsilon L_i \leq \zeta L_i$. We then regard each broken subpath as a \emph{tiny cluster}. Let $\mathbb{C}_{tiny}$ be the set of all tiny clusters.

Since $\zeta = 1/200$, the augmented diameter of each tiny cluster is much smaller than $L_i/2$. Thus, there is no edge in $E_{\un}(\widehat{\mathcal{P}})$ with both endpoints in the same tiny cluster. Furthermore, since no edge in $E_{\un}(\widehat{\mathcal{P}})$ is $(2L_i)$-clusterable by Lemma~\ref{lm:Type-II},  we have:

\begin{observation}\label{obs:tiny-structure} For any two endpoints $\widehat{\pzc[x]},\widehat{\pzc[y]}$ of an edge $e \in E_{\un}(\widehat{\mathcal{P}})$, $\widehat{\mathcal{I}}(\widehat{\pzc[x]}, 2L_i; \wmc[P])\cap \widehat{\mathcal{I}}(\widehat{\pzc[y]}, 2L_i; \wmc[P]) \not=\emptyset$.
\end{observation}

Let  $\dwh{\mathcal{P}}$ be the path obtained from $\widehat{\mathcal{P}}$ by contracting each tiny cluster in to a single node.  Given a node $\dwh{\pzc[x]} \in \dwh{\mathcal{P}}$,  we say that an edge $e \in E_{\un}(\wmc[P])$ \emph{shadows} $\dwh{\pzc[x]}$ if  $\dwh{\pzc[x]}$ lies on the subpath of $\dwh{\mathcal{P}}$ between $e$'s endpoints (see Figure~\ref{fig:clustering-tiny}). By definition, edges incident to $\dwh{\pzc[x]}$ shadow $\dwh{\pzc[x]}$.

\paragraph{Step 2: construct Type-IV clusters and orient edges} We iteratively construct Type-IV clusters, orient edges of $E_{\un}(\wmc[P])$ and mark nodes of $\dwh{\mathcal{P}}$ along the way.  Let  $\dwh{\pzc[x]} \in \dwh{\mathcal{P}}$ be an unmarked node  incident to a maximum number of \emph{unoriented edges} in $E_{\un}(\wmc[P])$. Let $E_{\dwh{\pzc[x]}}\subseteq E_{\un}(\wmc[P])$ be the set of unoriented edges shadowing $\dwh{\pzc[x]}$; $E_{\dwh{\pzc[x]}}$ could be empty. Let  $\dwh{\mathcal{Q}}$ be the minimal subpath of  $\dwh{\mathcal{P}}$ that contains $\dwh{\pzc[x]}$ and the endpoints of every edge in $E_{\dwh{\pzc[x]}}$.
 (If $E_{\dwh{\pzc[x]}} = \emp$ then $\dwh{\mathcal{Q}}$ contains a single node $\dwh{\pzc[x]}$.)
We regard $\dwh{\mathcal{Q}} \cup E_{\dwh{\pzc[x]}}$ as a Type-IV cluster. (We slightly abuse notation here; to be precise, the Type-IV cluster is obtained by uncontracting every node of $\dwh{\mathcal{Q}}$ and adding edges in $ E_{\dwh{\pzc[x]}}$.) Every unoriented edge incident to a node $\dwh{\pzc[y]}$ in $\dwh{\mathcal{Q}}$ will be oriented as out-going from $\dwh{\pzc[y]}$; edges with both endpoints in $\dwh{\mathcal{Q}}$ are oriented arbitrarily. (See Figure~\ref{fig:clustering-tiny} for an illustration.) We then remove nodes of $\dwh{Q}$ from $\dwh{P}$ and repeat this step to remaining subpaths of $\dwh{P}$ until every node of $\dwh{\mathcal{P}}$ is grouped into a Type-IV cluster.

Note that removing nodes of $\dwh{Q}$ from  $\dwh{P}$ could break $\dwh{P}$ into two subpaths, say $\dwh{P}_1$ and $\dwh{P}_2$. Since the endpoints of every edge shadowing $\dwh{\pzc[x]}$ are in $\dwh{Q}$, there is no edge connecting a node in $\dwh{P}_1$ and a node in $\dwh{P}_2$. Thus, the construction in Step 2 can be recursively applied to each subpath $\dwh{P}_1$ and $\dwh{P}_2$ until every node of $\dwh{\mathcal{P}}$ is grouped into a Type-IV cluster.

\begin{figure}[htb]
	\center{\includegraphics[width=0.9\textwidth]{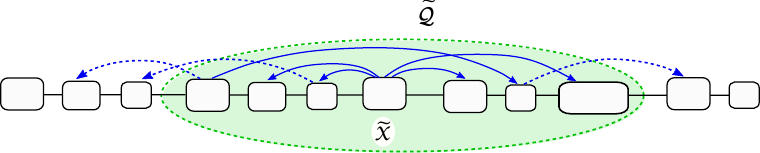}}
	\caption{Rectangular nodes are tiny clusters. Every unoriented edge between a node  in $\dwh{\mathcal{Q}}$ and a node outside $\dwh{\mathcal{Q}}$   will be oriented as out-going from $\dwh{\mathcal{Q}}$ (the dashed edges). Edges with both endpoints in $\dwh{\mathcal{Q}}$ are oriented arbitrarily. Solid edges in this figure are edges in $E_{\dwh{\pzc[x]}}$.}
	\label{fig:clustering-tiny}
\end{figure}

This completes the construction of Type-IV clusters. We now prove claims stated in Lemma~\ref{lm:Type-IV}.

\paragraph{Proof of Lemma~\ref{lm:Type-IV}.~}  Observe that Item (1) of Lemma~\ref{lm:Type-IV} holds by construction. 	

\begin{claim} \label{clm:edm-IV}
	$\zeta L_i \leq \adm(\dwh{\mathcal{Q}}) \leq 20L_i$
\end{claim}
\begin{proof}
	The lower bound of $ \adm(\dwh{\mathcal{Q}}) $ follows directly from the construction.  Let $\dwh{\pzc[y]}$ be an endpoint of $\dwh{\mathcal{Q}}$. By Observation~\ref{obs:tiny-structure}, $\adm(\dwh{\pzc[x]}, \dwh{\pzc[y]}) \leq 4L_i + \adm(\dwh{\pzc[x]}) + \adm(\dwh{\pzc[y]}) \leq (4+6\zeta)L_i \leq 10L_i$. Thus, $\adm(\dwh{\mathcal{Q}}) \leq 20L_i$.\QED
\end{proof}

Observe that Item (2) in Lemma~\ref{lm:Type-IV} follows directly from Claim~\ref{clm:edm-IV}. We now focus solely on lower bounding the total credit of edges and nodes in a Type-IV cluster. Indeed, this is the most difficult part in proving Lemma~\ref{lm:Type-IV}.

Let $\dwh{\mathcal{C}}$ be a Type-IV cluster where $\dwh{\mathcal{C}}  = \dwh{\mathcal{Q}} \cup E_{\dwh{\pzc[x]}}$ as described in Step 2. Let $\wmc[Q]$ be the  obtained from $\dwh{\mathcal{Q}}$ by uncontracting tiny clusters. Let $\mathcal{Q}$ be obtained from $\widehat{\mathcal{Q}}$ by uncontracting contracted nodes. We color  a node of $\mathcal{Q}$ \emph{black} if it belongs to a contracted node of $\wmc[Q]$; otherwise, we color the node \emph{white}.  In the same way, we denote by $\widehat{\mathcal{C}}$ (resp. $\mathcal{C}$) be obtained from $\dwh{\mathcal{C}}$ by uncontracting (resp. $\widehat{\mathcal{C}}$) tiny clusters (contracted nodes).   Note that by construction, $\widehat{\mathcal{Q}}$ is a path of $\widehat{\mathcal{F}}^{2}_{i-1}$, but $\mathcal{Q}$ may not be a path (see Figure~\ref{fig:Type4-uncontract}).

\begin{figure}[!htb]
	\center{\includegraphics[width=0.9\textwidth]{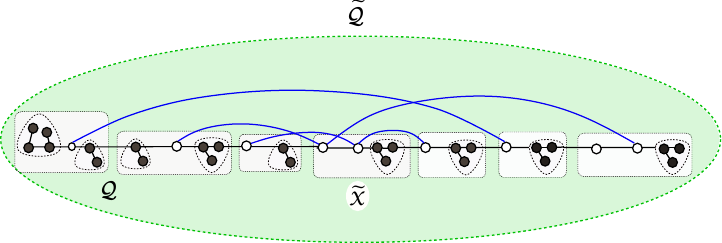}}
	\caption{Blue edges are edges in $E_{\dwh{\pzc[x]}}$; both endpoints of an edge in  $E_{\dwh{\pzc[x]}}$ are white nodes.}
	\label{fig:Type4-uncontract}
\end{figure}

Note by construction that, $\dwh{\pzc[x]}$ is incident to the most number of unoriented edges among all tiny clusters in $\dwh{\mathcal{C}}$. Since  each tiny cluster has augmented diameter at least $\zeta L_i$, by Claim~\ref{clm:edm-IV}, we have:

\begin{claim} \label{clm:deg-Type-IV}
	Let $E_{\dwh{\mathcal{C}}}$ be the subset of unoriented edges incident to tiny clusters in  $\dwh{\mathcal{C}}$. Then $|E_{\dwh{\mathcal{C}}}| = O(1) |E_{\dwh{\pzc[x]}}|$.
\end{claim}
\begin{proof}
	Since each tiny cluster has augmented diameter at least $\zeta L_i$, by Claim~\ref{clm:edm-IV}, there are at most $O(\zeta^{-1}) = O(1)$ tiny clusters in $\dwh{\mathcal{C}}$. Thus, the claim follows from the fact that $\dwh{\pzc[x]}$ is incident to most number of unoriented edges among all tiny clusters in $\dwh{\mathcal{C}}$. \QED
\end{proof}

Let $\mathcal{D}$ be the path realizing the diameter of $\mathcal{C}$; $\mathcal{D}$ may contain (level-$i$) spanner edges. Recall that no spanner edge in $\mathcal{C}$ is incident to a black node.  Thus, for each subtree $\mathcal{T}\in \mathbb{U}^{-}$ (in Lemma~\ref{lm:Type-I}) of $C$, $\mathcal{T}\cap \mathcal{D}$ is a path, i.e, $\mathcal{T}\cap \mathcal{D}$ has exactly one connected component. Let $\wmc[D]$ be obtained from $\mathcal{D}$ by contracting every subpath of  $\mathcal{D}$ in a contracted node of  $\wmc[Q]$. Observe that each subpath being contracted contains only black nodes, and that $\wmc[D]$ is a simple path.  Similarly, let $\dwh{\mathcal{D}}$ be obtained from $\wmc[D]$ by contracting tiny clusters into nodes; $\dwh{\mathcal{D}}$ also is a simple path.

Suppose that $\mathcal{D}$ contains a level-$i$ spanner edge $e$; we replace $e$ by the subpath of $\mathcal{F}_{i-1}$ connecting two endpoints of $e$. We repeat this replacement until we obtain a \emph{walk} of $\mathcal{F}_{i-1}$ between $\mathcal{D}$'s endpoints.  Let $\mathcal{D}_0$ be the (simple) path obtained by simplifying the walk; that is, $\mathcal{D}_0$ is obtained by removing closed subwalks on the walk. Since $\mathcal{D}_0$ is a subpath of $\mathcal{C}$ between two endpoints of $\mathcal{D}$, it holds that:
\begin{equation}\label{eq:D-vs-D0}
\adm(\mathcal{D}_0) \geq \adm(\mathcal{D})
\end{equation}

 By construction, $\mathcal{D}_0$ is a subpath of $\mathcal{F}_{i-1}$; this implies  $\cred(\mathcal{D})_0 \geq \ce \adm(\mathcal{D}_0)$.
 
\begin{lemma}\label{lm:credit-D0} if $\mathcal{C}$ has at most $2g\zeta^{-1}\epsilon^{-1}$ nodes, then $$\ce(\adm(\mathcal{D}_0)) + \sum_{\pzc[x] \in \mathcal{C}\setminus \mathcal{D}_0}\cred(\pzc[x])= \ce\left(\adm(\mathcal{C}) + \Omega(t\epsilon^2 L_i)\right)~.$$
\end{lemma}

 We will show that Lemma~\ref{lm:credit-D0} implies imply Item (3) of Lemma~\ref{lm:Type-IV}.
 
 \begin{lemma} If Lemma~\ref{lm:credit-D0} holds, then  $$\cred(\mathcal{C}) =  \ce\left(\adm(\mathcal{C}) + \Omega(t\epsilon + |\mathcal{C}\cap \mathcal{W}|) \epsilon L_i\right) .$$
 \end{lemma}
\begin{proof} Note that $\mathcal{C}\cap \mathcal{W}$ is the set of nodes contained in contracted nodes of $\wmc[C]$. If $\mathcal{C}$ has at least  $2g\zeta^{-1}\epsilon^{-1}$ nodes, since $\adm(\mathcal{C}) \leq gL_i$ and each node has at least $\zeta \epsilon L_i$ credits by Invariant (I2) for level-$(i-1)$, $\cred(\mathcal{Y}) \geq \ce \adm(\mathcal{C})$ for any set $\mathcal{Y}$ of $g\zeta^{-1}\epsilon^{-1}$ nodes of $\mathcal{C}$. Since $|\mv(\mathcal{C})| \geq 2 |\mathcal{Y}|$, it holds that:
	\begin{align*}
		\cred(\mathcal{C})  \geq \ce (\adm(\mathcal{C})) + |\mv(\mathcal{C} \setminus \mathcal{Y})|\zeta \epsilon L_i =  \ce(\adm(\mathcal{C})) + \Omega(\epsilon L_i|\mv(\mathcal{C})|)
	\end{align*}
	Clearly, $|\mv(\mathcal{C})| \geq |\mathcal{C} \cap \mathcal{W}|$ and observe that $t =  O(|\mv(\mathcal{C})|/\epsilon)$ since each node in $\mathcal{C}$ has degree $O(\epsilon^{-1})$. Thus, $ \Omega(\epsilon L_i|\mv(\mathcal{C})|) = \Omega(t\epsilon + |\mathcal{C}\cap \mathcal{W}|) \epsilon L_i$.
	
	We now consider the case that  $\mathcal{C}$ has at most $2g\zeta^{-1}\epsilon^{-1}$ nodes. We will show that:
	\begin{equation}\label{eq:blacknodes-C}
	 \cred(\mathcal{C}) =	\ce\left(\adm(\mathcal{C}) + \Omega(|\mathcal{C} \cap \mathcal{W}|) \epsilon L_i\right)
	\end{equation}
 The lemma then follows from Equation~\eqref{eq:blacknodes-C} and Lemma~\ref{lm:credit-D0}.
 
 Our argument to establish Equation~\eqref{eq:blacknodes-C} is similar to the proof of Case 2 in Lemma~\ref{lm:I2-main}.  Let $\pzc[x]$ be a node in $\mathcal{C}\cap \mathcal{W}$ --- $\pzc[x]$ is a black node ---  and $\mathcal{T}$ be the subtree of $\mathcal{F}_{i-1}$ corresponding to the contracted node  that contains $\pzc[x]$. Observe that $\mathcal{D}_{\mathcal{T}} = \mathcal{D}_0\cap \mathcal{T}$ is a  simple path. We take exactly $\ce \adm(\mathcal{D}_0) \geq \ce \adm(\mathcal{C})$ credits from nodes and edges of $\mathcal{D}_0$ by taking exactly $\ce \adm(\mathcal{D}_{\mathcal{T}})$ credits from $\mathcal{T}$ for each subtree  $\mathcal{T}$  corresponding to a contracted node in $\wmc[C]$.  By Lemma~\ref{lm:tree-credits}, each node in $\mathcal{T}$ has at least $\Omega(\ce \epsilon L_i)$ credits left. In particular, $\pzc[x]$ has $\Omega(\ce \epsilon L_i)$ credits as leftover. Thus, the total amount of leftover credits of nodes in $\mathcal{C}\cap \mathcal{W}$ is $\Omega(\ce|\mathcal{C} \cap \mathcal{W}|\epsilon L_i)$; this implies Equation~\eqref{eq:blacknodes-C}.
	\QED
	
\end{proof}

Henceforth, we focus on proving Lemma~\ref{lm:credit-D0}. Note that $t$ is the number of out-going edges incident to nodes in $\mathcal{C}$. Lemma~\ref{lm:credit-D0} trivially holds when $t = 0$ (in this case $\mathcal{C}$ is a tiny cluster) as $\adm(\mathcal{D}_0) \geq \adm(\mathcal{C})$ by Equation~\eqref{eq:D-vs-D0}. Thus, we assume that $t\geq 1$.  Our proof uses the fact that the spanner has stretch $1+s\epsilon$. First, we claim that:

\begin{claim}\label{clm:detour-saving} Let $\pzc[x]$ and $\pzc[y]$ be  endpoints of a level-$i$ spanner edge $e$. Let $\mathcal{P}$ be any path between $\pzc[x]$ and $\pzc[y]$ in $\mathcal{G} \setminus \{e\}$.  Then $\adm(\mathcal{P}) \geq w(e) + 3g\epsilon L_i$ when $s \geq 6g$.
\end{claim}
\begin{proof}
	Here we note that $\mathcal{P}$ could contain spanner edges. Since $\mathcal{P}$ induces a path of length at most $\adm(\mathcal{P})$ in $S_\gr$ between $e$'s endpoints, by Fact~\ref{fact:edge-path-weight}, $
	\adm(\mathcal{P})  \geq (1+s\epsilon)w(e) ~\geq~ w(e) + s/2\epsilon L_i ~\geq~ w(e) + 3g\epsilon L_i$. \QED
\end{proof}

We observe the following immediate corollary of Claim~\ref{clm:detour-saving}.

\begin{corollary}\label{cor:diam-path-IV}
	If $\mathcal{D}$ contains both endpoints of an edge $e \in E_{\dwh{\pzc[x]}}$, then $e \in \mathcal{D}$.
\end{corollary}
\begin{proof} Suppose otherwise. Let $\mathpzc{x}$ and $\mathpzc{y}$ be the endpoints of $e$.  Since $e\not\in\mathcal{D}$, we replace the subpath of $\mathcal{D}$ between $\pzc[x]$ and $\pzc[y]$ by $e$ to obtain new path $\mathcal{D}'$ of $C$ of augmented diameter:
	\begin{equation}\label{eq:D-vs-Dprime}
	\begin{split}
	\adm(\mathcal{D}') &= \adm(\mathcal{D}) + w(e) - \adm(\mathcal{D}[\mathpzc{x},\mathpzc{y})] + w(\pzc[x]) + w(\pzc[y]) \\
	&\leq \adm(\mathcal{D}) + w(e) - (w(e) + 3g\epsilon L_i) + 2g\epsilon L_i < \adm(\mathcal{D}),
	\end{split}
	\end{equation}
	by Claim~\ref{clm:detour-saving}. This contradicts that $\mathcal{D}$ is the diameter path of $\mathcal{C}$.\QED
\end{proof}

We now claim a weaker version of Lemma~\ref{lm:credit-D0} which conveys the main intuition of the proof.

\begin{claim}\label{clm:at-least-one-extra} $\ce(\adm(\mathcal{D}_0)) + \sum_{\pzc[x] \in \mathcal{C}\setminus \mathcal{D}}\cred(\pzc[x]) \geq\ce\left( \adm(\mathcal{C}) + \Omega(\epsilon L_i)\right)$.
\end{claim}
\begin{proof}
	If $\mathcal{D}$ does not contain any edge in $E_{\dwh{\pzc[x]}}$, then by Corollary~\ref{cor:diam-path-IV}, at least one node, say $\pzc[y]$, incident to some edge in  $E_{\dwh{\pzc[x]}}$ is not in $\mathcal{D}_0$.  This  node is a white node since edges in $E_{\dwh{\pzc[x]}}$ are incident to white nodes only.  By Invariant (I2) for level $i-1$, $\cred(\pzc[y]) \geq \ce\zeta L_{i-1} =\Omega(\ce \epsilon L_i)$; this implies the claim.

	Suppose that $\mathcal{D}$ contains an edge $e$ with two endpoints $\pzc[u], \pzc[v]$. Let $\mathcal{D}_1$ be obtained from $\mathcal{D}_0$ by replacing the subpath $\mathcal{D}_0(\pzc[u], \pzc[v])$ by $e$. By Claim~\ref{clm:detour-saving}, it holds that:
	\begin{equation*}
	\begin{split}
	\adm(\mathcal{D}_0)  - \adm(\mathcal{D}_1) \geq 3 g\epsilon L_i  - w(\pzc[u]) - w(\pzc[v]) \geq 3 g\epsilon L_i -  2g\epsilon L_i = \Omega(\epsilon L_i).
	\end{split}
	\end{equation*}
	Since $\mathcal{D}$ is a diameter path, $\adm(\mathcal{D}_1) \geq \adm(\mathcal{D}) = \adm(\mathcal{C})$. Thus, $$\adm(\mathcal{D}_0)  - \adm(\mathcal{C}) ~\geq~ \adm(\mathcal{D}_0)  - \adm(\mathcal{D}_1) ~=~ \Omega(\epsilon L_i);$$ the claim holds.\QED
\end{proof}

Observe that Claim~\ref{clm:at-least-one-extra} implies Lemma~\ref{lm:credit-D0} when $t = O(\frac{1}{\epsilon})$.  However, the number of edges out-going from a Type-IV cluster could be up to $\Omega(\epsilon^{-2})$. The following lemma help us handle the case where $t \gg 1/\epsilon$.

\begin{lemma}\label{lm:t-node-IV}
	Let $\mathpzc{u}$ be a node in $\mathcal{C}$ that is incident to  $k$ edges in $E_{\dwh{\pzc[x]}}$. If $\mathpzc{u} \in \mathcal{D}_0$, then
	\begin{equation*}
	\ce(\adm(\mathcal{D}_0)) + \sum_{\pzc[x] \in \mathcal{C}\setminus \mathcal{D}_0}\cred(\pzc[x]) \geq \ce\left(\adm(\mathcal{C}) + \Omega(k \epsilon L_i)\right)
	\end{equation*}
\end{lemma}
\begin{proof}
	Let $E_{\pzc[u]}$ be set of edges in $E_{\dwh{C}}$ incident to $\pzc[u]$.  Let $\mathcal{Z}$ be the set of other endpoints of edges in $E_{\pzc[u]}$.
	Let $\mathcal{D}_0^{left}$ and $\mathcal{D}_0^{right}$ be two subpaths of $\mathcal{D}_0\setminus \{\pzc[u]\}$. Let $\mathcal{Z}^{left} = \mathcal{D}_0^{left}\cap \mathcal{Z}$ and $\mathcal{Z}^{right} = \mathcal{D}_0^{right}\cap \mathcal{Z}$ and  $\mathcal{Z}^{free} = \mathcal{Z} \setminus \{\mathcal{Z}^{left} \cup \mathcal{Z}^{right}\}$ (see Figure~\ref{fig:left-right-free}). 
	\begin{figure}[!htb]
		\center{\includegraphics[width=0.7\textwidth]
			{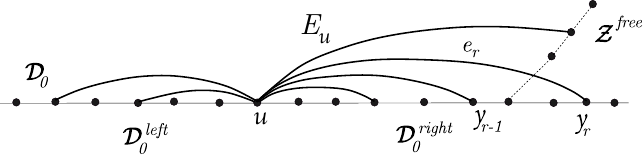}}
		\caption{Dashed edges are in $\mathcal{D}_0$ while solid edges are not in $\mathcal{D}_0$. Any endpoint  not in $\mathcal{D}_0$  of an edge in $E_{\pzc[u]}$ belongs to $\mathcal{Z}^{free}$.}
		\label{fig:left-right-free}
	\end{figure}
	
	We now focus on $\mathcal{D}_0^{right}$. Let $r = |\mathcal{Z}^{right}|$, and assume that $r > 0$. Then, $\mathcal{D}$ contains an edge $e_r \in E_{\pzc[u]}$ with an endpoint in $\mathcal{Z}^{right}$; there is only one such edge since $\mathcal{D}$ is a simple path. Let $\pzc[y]_r$ be another endpoint of $e_r$. Then $\pzc[u]$ is to the left of $\pzc[y]_r$ on $\mathcal{D}_0^{right}$   and by  Corollary~\ref{cor:diam-path-IV}, no node in $\mathcal{Z}^{right}$ is to the right of $\pzc[y]_r$ on $\mathcal{D}_0^{right}$. That is, $\mathcal{Z}^{right}\subseteq \mathcal{D}_0[\pzc[u], \pzc[y]_r]$ (see Figure~\ref{fig:left-right-free}).
	\begin{claim} \label{clm:e-r-length}
		$\adm(\mathcal{D}_0[\pzc[u], \pzc[y]_r])  \geq w(e_r) + (|\mathcal{Z}^{right}| +2)g\epsilon L_i$.
	\end{claim}
	\begin{proof}
		Let $\{e_1,e_2,\ldots, e_{r-1}\}$ be other edges of $E_{\pzc[u]}$ with endpoints $\pzc[y]_1, \ldots, \pzc[y]_{r-1}$ in $\mathcal{Z}^{right}$ where $\pzc[y]_j$ is to the left $\pzc[y]_{j+1}$ on $\mathcal{D}_0^{right}$, $1\leq j\leq r-2$. We prove the claim by induction on $r$. The claim is true when $r = 1$ by Claim~\ref{clm:detour-saving}. By the induction hypothesis, $ \adm(\mathcal{D}_0[\pzc[u], \pzc[y]_{r-1}]) \geq w(e_{r-1}) + (r+1)g\epsilon L_i$. Let $\mathcal{P} = \pzc[u]\circ e_{r-1}\circ \mathcal{D}_0[y_{r-1},y_r]$ be a path from $\pzc[u]$ going through $e_{r-1}$ and following $\mathcal{D}_0$ to $y_r$. By Claim~\ref{clm:detour-saving}, we have $\adm(\mathcal{P}) \geq w(e_r) + 3g\epsilon L_i$. Thus,
		\begin{equation}
		\begin{split}
		\adm(\mathcal{D}_0[\pzc[u], \pzc[y]_r]) &= \adm(\mathcal{D}_0[\pzc[u], \pzc[y]_{-1}]) + \adm(\mathcal{D}_0[\pzc[y]_{r-1}, \pzc[y]_r]) - w(\pzc[y]_{r-1}) \\
		&\geq w(e_{r-1}) + (r+1)g\epsilon L_i + \adm(\mathcal{D}_0[\pzc[y]_{r-1}, \pzc[y]_r]) - w(\pzc[y]_{r-1})\\
		&\geq w(\pzc[x]) + w(e_{r-1})  + \adm(\mathcal{D}_0[\pzc[y]_{r-1}, \pzc[y]_r]) + (r+1)g\epsilon L_i  -w(\pzc[y]_{r-1}) - w(\pzc[x])\\
		&\geq \adm(\mathcal{P}) + (r-1) g\epsilon L_i  \\
		&\geq w(e_r) +  (r-1) g\epsilon L_i + 3g\epsilon L_i  \qquad \mbox{(by Claim~\ref{clm:detour-saving})}\\
		&=  w(e_r) +  (r+2) g\epsilon L_i \inQED
		\end{split}
		\end{equation}
	\end{proof}

	Let $\mathcal{D}^{right}_1$ be the path obtained from $\mathcal{D}_0^{right}$ by replacing the path $\mathcal{D}^{right}_0[\pzc[u], \pzc[y]_{r-1}]$ by $e_r$. Then Claim~\ref{clm:e-r-length} implies that:
	\begin{equation*}
		 \adm(\mathcal{D}_0^{right})  - \adm(\mathcal{D}^{right}_1) \geq \ce(r+2)g\epsilon L_i - \ce(w(\pzc[u]) + w(\pzc[y_r])) \geq r\ce g\epsilon L_i
	\end{equation*}
	
	We now consider $\mathcal{D}_0^{left}$. Let $e_{\ell}$ be the edge of $E_{\pzc[u]}$  in $\mathcal{D}$. We construct $\mathcal{D}^{left}_1$ exactly the same way we construct $\mathcal{D}^{right}_1$: replace the subpath of $\mathcal{D}_0^{left}$ between $e_{\ell}$'s endpoints by $e_\ell$.   By the same argument, we have $ \adm(\mathcal{D}_0^{left})  - \adm(\mathcal{D}^{left}_1) ~\geq~ \ell\ce g\epsilon L_i$.  Let $\mathcal{D}_1 = \mathcal{D}_1^{left}\circ \mathcal{D}_1^{right}$. Since $\mathcal{D}$ is a diameter path, $\adm(\mathcal{D}_1) \geq \adm(\mathcal{D})$. Thus, we have:
	 
	 \begin{equation*}
	 \begin{split}
	 \ce(\adm(\mathcal{D}_0)) + \sum_{\pzc[y] \in \mathcal{C}\setminus \mathcal{D}_0} \cred(\pzc[y]) &\geq \ce\left(\adm(\mathcal{D}_1)  + (|\mathcal{Z}^{left}| + |\mathcal{Z}^{right}|)g\epsilon L_i\right) + |\mathcal{Z}^{free}|\zeta \epsilon L_i\\
	 &\geq \ce\adm(\mathcal{D}) + \Omega(\underbrace{|\mathcal{Z}^{free}| + |\mathcal{Z}^{left}| + |\mathcal{Z}^{right}|}_{= |\mathcal{Z}| = k})\epsilon L_i\\
	 &= \ce\adm(\mathcal{C}) + \Omega(k\epsilon L_i), 
	 \end{split}
	 \end{equation*}
	as claimed. \QED
\end{proof}

\noindent We are now ready to prove Lemma~\ref{lm:credit-D0}.
\vspace{0.2cm}
\begin{proof}[Proof of Lemma~\ref{lm:credit-D0}]
	Observe that the lemma holds when $t = 0$; thus, we can assume that $t > 0$. Let $c_0$ be a constant such that $E_{\dwh{\pzc[x]}} \geq \frac{t}{c_0}$; $c_0$ exists by Claim~\ref{clm:deg-Type-IV}. 

	Recall that every node in $\mathcal{C}$ is incident to at most $\upsilon \stackrel{\footnotesize{\mathrm{def.}}}{=} 2g\zeta^{-1}\epsilon^{-1}$ edges. If $t\leq 2\upsilon$, then $t\epsilon \leq 4g\zeta^{-1} = O(1)$, implying that $\Omega(\epsilon L_i) = \Omega(t\epsilon^2 L_i)$. Thus, the lemma follows directly from Claim~\ref{clm:at-least-one-extra}.
	
	Henceforth, we assume that $t \geq 2\upsilon$. Let $\wmc[X]$ be the set of nodes in the subpath of $\wmc[P]$ in $\dwh{\pzc[x]}$.  Let  $\mathcal{X}$ be the set of nodes obtained by uncontracting contracted nodes in $\wmc[X]$. Let $\mathcal{Y} \subseteq \mathcal{X}$ be the set of nodes where  each node in $\mathcal{Y}$ is incident to at least   $\frac{t}{2\upsilon c_0}$ edges in $E_{\dwh{\pzc[x]}}$.
	
	\begin{claim}  $|\mathcal{Y}| \geq \frac{t}{2\upsilon c_0}$.
	\end{claim}
	\begin{proof}
		Suppose otherwise, then nodes in $\mathcal{Y}$ are incident to less than $\frac{t}{2\upsilon c_0} \upsilon = \frac{t}{2 c_0}$ edges in $E_{\dwh{\pzc[x]}}$. Since $|\mathcal{X}|\leq \frac{2g}{\zeta\epsilon} = \upsilon$ by the assumption  in Lemma~\ref{lm:credit-D0} that $\mathcal{C}$ has at most $\upsilon$ nodes,   nodes in $\mathcal{X}\setminus \mathcal{Y}$ are incident to less than $\frac{t}{2\upsilon c_0} \upsilon = \frac{
			t}{2 c_0}$  edges in $E_{\dwh{\pzc[x]}}$. This contradicts that $|E_{\dwh{\pzc[x]}}| \geq \frac{t}{c_0}$.\QED
	\end{proof}
	
	If $\mathcal{D}_0 \cap \mathcal{Y} = \emptyset$, then $|\mathcal{C}\setminus \mathcal{D}_0| \geq \frac{t}{2\upsilon c_0} = \Omega(t\epsilon)$. Thus, by Invariant (I2) for level $i-1$, $$\sum_{\pzc[x] \in \mathcal{C}\setminus \mathcal{D}_0}\cred(\pzc[x])\geq | \mathcal{C}\setminus \mathcal{D}_0| \epsilon L_i = \Omega(t\epsilon^2 L_i).$$
	
	Otherwise, Lemma ~\ref{lm:credit-D0} follows by applying Lemma~\ref{lm:t-node-IV} with $k = \frac{t}{2\upsilon c_0} = \Omega(t\epsilon)$.  \QED
\end{proof}

\subsection{Deferred proofs}

\subsubsection{Proof of Lemma~\ref{lm:deg-K}}\label{subsec:Proof-Deg-K}

We restate the lemma here for convenience.

\degK*

Recall that $(1+s\epsilon)$ is the stretch of $S_{\gr}$ and $\delta$ is the parameter in Lemma~\ref{lm:weight-reduction}.  The main idea is that for each $\epsilon$-cluster $\pzc[x]$, we partition the space into $O(\epsilon^{-d+1})$ cones around an arbitrary point $p \in \pzc[x]$ and show that, for each cone, there is at most one spanner edge from $p$ to another point (belonging to some other $\epsilon$-cluster) in the cone.

\begin{claim}\label{clm:K-simple}$\mathcal{K}$ is a simple graph when $s\geq 16g$.
\end{claim}
\begin{proof}
	The proof is similar to the proof by BLW (see Lemma 3.1 in~\cite{BLW19}); we include the details here for completeness.  Suppose that there are two parallel edges $e_1 = (x_1,y_1), e_2 = (x_2,y_2)$ where $x_1,x_2$ are in the same $\epsilon$-cluster $\pzc[x]$, and $y_1,y_2$ are in the same $\epsilon$-cluster $\pzc[y]$. Note that $L_i/2 \leq L_i/(1+\delta) \leq |x_1y_1|, |x_2y_2| \leq L_i$.
	
	By Invariant (I1), $\pzc[x]$ and $\pzc[y]$ have diameters at most $g\epsilon L_i$. Thus, there is a path $P_x$ ($P_y$) in $S_{\gr}$ between $x_1$ and $x_2$ ($y_1$ and $y_2$) of weight at most $(1+\epsilon)g\epsilon L_i \leq 2g\epsilon L_i$ since $\eps \leq 1$. Let $Q$ be the path between $x_2$ and $y_2$ composed of $P_x$, edge $x_1y_1$, and $P_y$. By the triangle inequality, we have $|x_2y_2|\geq |x_1y_1| - w(P_x) - w (P_y) ~\geq ~ |x_1y_1| - 4g\epsilon L_i$. Thus, it holds that:
	\begin{align*}
		(1+s\epsilon)|x_2y_2| &\geq (1+s\epsilon)( |x_1y_1| - 4g\epsilon L_i)\\
		& \geq |x_1y_1| + s\epsilon L_i/2 - 4g\epsilon L_i \\
		&\geq  |x_1y_1| + 4g\epsilon L_i \geq  |x_1y_1| + w(P_x) + w(P_y) = w(Q), 
	\end{align*}
contradicting Fact~\ref{fact:edge-path-weight}. 	\QED
\end{proof}

We now focus on bounding $\dk$.  Let $\pzc[x]_0$ be an $\epsilon$-cluster in $\mathcal{K}$ of maximum degree and $\pzc[x]_1, \pzc[x]_2,\ldots,\pzc[x]_{\dk}$ be $\pzc[x]_0$'s neighbors. Let $x_i$ be a point of $P$ in $\pzc[x]_i$.  Note that edge $e_i$ in  $\mathcal{K}$ connecting $\pzc[x]_0$ and $\pzc[x]_i$ has length at most $L_i$ and at least $\frac{1}{(1+\delta)}L_i \geq (1-\delta)L_i$. We denote by $\dm(\pzc[x])$ the diameter of an $\epsilon$-cluster $\pzc[x]$.  Since  $\dm(\pzc[x]_i)\leq gL_{i-1} = g\epsilon L_i$ by invariant (I2), we have:
\begin{equation}\label{eq:x0xi-up-dist}
	L_i(1 - \delta - 2g\epsilon) \leq  |x_0x_i| \leq (2g\epsilon + 1)L_i
\end{equation}
for every $ i \in [\dk]$.  Thus, by Equation~\ref{eq:x0xi-up-dist}, every point in $X = \{x_1,x_2,\ldots, x_p\}$ lies in the annulus $A = B_d(x_0,L_i(1+2g\epsilon))\setminus B_d(x_0,L_i(1-\delta - 2g\epsilon))$ (see Figure~\ref{fig:annulus}). 

\begin{figure}[!htb]
	\center{\includegraphics[width=0.7\textwidth]
		{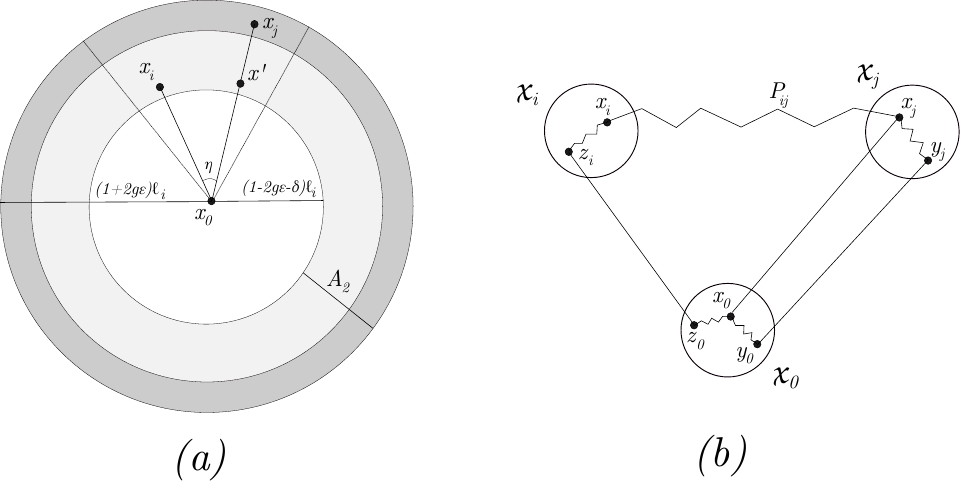}}
	\caption{An illustration for the proof of Lemma~\ref{lm:deg-K}. Zigzag paths are the paths in $S_{\gr}$ between the corresponding endpoints}
	\label{fig:annulus}
\end{figure}

We first prove Lemma~\ref{lm:deg-K} when $d = 2$ to convey the intuition of our argument.

\begin{claim}\label{clm:deg-K-R2}
	$|X| \leq  \frac{1}{\epsilon}$ when $d = 2$.
\end{claim}
\begin{proof}
	First, we divide the circle $B_2(x_0,L_i(1+2g\epsilon))$ into $\Theta(\frac{1}{\epsilon})$ equal sectors where the angle between two radii of the sector is $\Theta(\epsilon)$. To prove the claim, we will show that each sector contains at most one point of $X$.
	
	Suppose that there are two points $x_i,x_j $ of $X$ in the same sector. Then the angle $\eta = \angle x_ix_0x_j$ is at most $\epsilon$.~W.l.o.g, we assume that $x_0x_i \leq x_0x_j$. Our strategy is to show that there is a path connecting $x_0$ and $x_j$ that does not contain $x_0x_j$ and has stretch at most $(1+s\epsilon)$; this contradicts Fact~\ref{fact:edge-path-weight}. The intuition behind the proof is that when $\delta$ is sufficiently small, the distance between $x_i$ and $x_j$ is at most $O(\delta L_i)$. Since there is a good stretch path between $x_i$ and $x_j$ in the spanner, this path with the edge $x_0x_i$ would form an $x_0$-to-$x_j$ path of length at most $(1+s\epsilon)|x_0x_j|$.

	Let $x'$ be the point on the segment $x_0x_j$ such that $|x_0x'| = |x_0x_i|$ (see Figure~\ref{fig:annulus}(a)). We have:
	\begin{equation}\label{eq:angle-dist}
		|x'x_i| ~= ~2|x_0x_i| \sin(\eta/2)  ~\leq~ |x_0x_i| \eta ~\leq ~\epsilon |x_0x_i|
	\end{equation}

	Since both $x'$ and $x_j$ are in the annulus $A_2 = B_2(x_0,L_i(1+2g\epsilon))\setminus B_2(x_0,L_i(1-\delta - 2g\epsilon))$, we have $|x_jx'| \leq  L_i(\delta + 4g\epsilon)$. Thus, by the triangle inequality,
	\begin{equation}\label{eq:xixj-length}
		\begin{split}
			|x_ix_j|~ &\leq ~|x_ix'| + |x'x_j|~ \leq ~  \epsilon|x_0x_i| +   L_i(\delta + 4g\epsilon) \\
			&\leq L_i(1+2g\epsilon)\epsilon+   L_i(\delta + 4g\epsilon) \\
			&= L_i(\delta + (4g+1)\epsilon + 2g\epsilon^2) \\
			&\leq (6g+2)\delta L_i \qquad \mbox{(since }\epsilon \ll \delta \ll 1)\\
			&= L_i/8 \leq |x_0x_j|/4  \quad \mbox{(since } |x_0x_j| \geq \frac{L_i}{1+\delta} \geq L_i/2).
		\end{split}
	\end{equation}
	Let $P_{ij}$ be a path of weight at most $(1+s\epsilon)|x_ix_j|$ between $x_i$ and $x_j$ in $S_{\gr}$ (see Figure~\ref{fig:annulus}(b)). Note that $\epsilon \ll \frac{1}{s}$, so $s\epsilon < 1$. That implies $|x_ix_j| ~\geq~ \frac{w(P_{ij})}{1+s\epsilon} ~\geq ~(1-s\epsilon)w(P_{i,j})$ and
	\begin{equation}\label{eq:Pij-vs-x0xj}
		w(P_{ij}) ~\leq~ 2|x_ix_j|~ \leq ~|x_0x_j|/2
	\end{equation}
	by Equation~\eqref{eq:xixj-length}.  We have:
	\begin{equation}\label{eq:x0xj-vs-x0xi}
		\begin{split}
			(1+s\epsilon)|x_0x_j| &= |x_0x'| + |x'x_j| + s\epsilon|x_0x_j| \\ &\geq |x_0x_i| + |x_ix_j| - |x_ix'| + s\epsilon|x_0x_j| \quad\mbox{ (by the triangle inequality)}\\
			&\geq |x_0x_i| + (1-s\epsilon)w(P_{ij})- |x_ix'| + s\epsilon|x_0x_j|\\
			&\geq  (|x_0x_i| + w(P_{ij})) + s\epsilon |x_0x_j| - \epsilon|x_0x_i| -  s\epsilon|x_0x_j|/2 \qquad \mbox{ (by Equations~\eqref{eq:Pij-vs-x0xj} and~\ref{eq:angle-dist})}\\
			&\geq (|x_0x_i| + w(P_{ij})) + (s/2 -1)\epsilon|x_0x_j|\\
		\end{split}
	\end{equation}
	Let $z_0z_i$ and $y_0y_j$ be $\pzc[x]_0$-to-$\pzc[x]_i$ and $\pzc[x]_0$-to-$\pzc[x]_j$ edges in $\mathcal{K}$, respectively.  Let
	\begin{equation*}
		P_{0j} = (y_0\stackrel{S_{\gr}}{\leadsto}x_0)\circ (x_0\stackrel{S_{\gr}}{\leadsto}z_0) \circ z_0z_i\circ (z_i\stackrel{S_{\gr}}{\leadsto}x_i)\circ (x_i\stackrel{S_{\gr}}{\leadsto}x_j)\circ (x_j\stackrel{S_{\gr}}{\leadsto}y_j)
	\end{equation*}
	be a $y_0$-to-$y_j$-path in $S$ between $y_0$ and $y_j$. We have:
	\begin{equation}\label{eq:P0j-vs-x0xi}
		\begin{split}
			w(P_{0j}) &= 2g\epsilon L_i + |z_0z_i| + (w(P_{ij}) + 2g\epsilon L_i)\\
			&\leq 4g\epsilon L_i + (|x_0x_i| + 2g\epsilon L_i) + w(P_{ij}) \\
			&\leq  (|x_0x_i| + w(P_{ij})) + 6g\epsilon L_i
		\end{split}
	\end{equation}
	
	Thus, it holds that:
	\begin{equation*}
		\begin{split}
			(1+s\epsilon)|y_0y_j| &\geq (1+s\epsilon)(|x_0x_j| - 2g\epsilon L_i)\\
			&\geq (1+s\epsilon)|x_0x_j|  - 4g\epsilon L_i \quad \mbox{ (since }s\epsilon < 1)\\
			&\geq (|x_0x_i| + w(P_{ij})) + (s/2 -1)\epsilon|x_0x_j| - 4g\epsilon L_i \qquad \mbox{ (by Equation~\eqref{eq:x0xj-vs-x0xi})}\\
			&\geq w(P_{0j})  + (s/2 -1)\epsilon|x_0x_j| - 10g\epsilon L_i  \qquad \mbox{(by Equation~\eqref{eq:P0j-vs-x0xi})}\\
			&\geq w(P_{0j}) + (s/4 - 1/2 - 10g)\epsilon L_i \quad \mbox{ (since }|x_0x_j| \geq L_i/2)\\
			&> w(P_{0j}) \qquad \mbox{ since }s\geq 40g+3
		\end{split}
	\end{equation*}
	This contradicts that $y_0y_j$ is an edge of $S_{\gr}$ by Fact~\ref{fact:edge-path-weight}.\QED
\end{proof}
To show the generalized version of Claim~\ref{clm:deg-K-R2} for general $d$, we consider a set of $(d,\epsilon/2)$-spherical code $A_{d,\epsilon/2}$ (see Definition~\ref{def:spherical-code}). By using standard volume argument (see Lemma~\ref{lm:spherical-code-up}), we have:
\begin{equation}\label{eq:A-size}
	|A_{d,\epsilon/2}| = \Theta(\epsilon^{-d+1}))
\end{equation}
For each point $c \in A_{d,\epsilon/2}$, we define a spherical sector $S_c$ with angle $\epsilon$ and apex $x_0$ that has $c$ as the middle point of the cap of $S_c$. We claim that:
\begin{claim}\label{clm:X-cone-C}
	There is at most one point of $X$ in $S_c$.
\end{claim}
\begin{proof}
	Suppose  for contradiction that there are two points $x_i$ and $x_j$ of $X$ in $S_c$. Then, $\angle x_ix_0x_j \leq \epsilon$. Assume that $|x_0x_i| \leq |x_0x_j|$. Let $y_0y_j$ be the $\pzc[x]_0$-to-$\pzc[x]_j$ edge in $\mathcal{K}$. By exactly the same argument as in the proof of Claim~\ref{clm:deg-K-R2}, we conclude that there is a path of weight at most $(1+s\epsilon)|y_0y_j|$ in $S_\gr$, which contradicts that $y_0y_j$ is an edge of $S_\gr$; the claim follows. \QED
\end{proof}

Lemma~\ref{lm:deg-K} then follows directly from Claim~\ref{clm:X-cone-C} and Equation~\eqref{eq:A-size}.

\subsubsection{Tree clustering: Proof of Lemma~\ref{lm:tree-clustering}}\label{subsec:Proof-Tree-Cluster}

In this section, we prove Lemma~\ref{lm:tree-clustering}. We say a node $x$ \emph{$T$-branching} if $x$ has at least three neighbors in $T$. When the tree is clear from the context, we simply say that $x$ is branching. We denote by $|xy|$ the \emph{augmented distance} between two nodes $x,y \in T$. (The augmented distance between two nodes $x$ and $y$ in $T$ is the augmented length of the path between $x$ and $y$ in $T$.) 

\begin{definition}[Branching Radius]\label{def:rx}
	For each branching node $x$, the branching radius of $x$ is the largest positive $r$ such that  there exist \emph{three} internally-node disjoint paths   $P_1,P_2,P_3$ of $T$ such that:
	\begin{enumerate}[noitemsep]
		\item $P_1,P_2,P_3$ share the same endpoint $x$.
		\item  $r - 2\eta L \leq   \adm(P_1), \adm(P_2), \adm(P_3) \leq r$.
	\end{enumerate}
We denote the branching node of $x$ by $r(x)$.
\end{definition}

\begin{wrapfigure}{r}{0.35\textwidth}
	\vspace{-20pt}
	\begin{center}
		\includegraphics[width=0.35\textwidth]{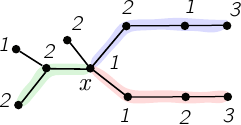}
	\end{center}
	\vspace{-10pt}
	\caption{\footnotesize{Each edge has weight $1$ and each node has weight as annotated in the figure, with $\eta L = 3$. The branching radius $r(x) = 13$; three paths realizing the branching radius of $x$ are highlighted by different colors. One path (colored green) has augmented length $7$ while two other paths have augmented length $10$.}}
	\label{fig:branching-rad}
\end{wrapfigure}

The slack $-2\eta L$ is due to that edges and nodes of $T$ has weight at most $\eta L$. (See Figure~\ref{fig:branching-rad} for an illustration.) Clearly, by definition, $r(x)\geq 2\eta L$ for all $x$. We have:

\begin{observation}\label{obs:branching-rad} Given a branching node $x$,  let $P_1,P_2,P_3$ be any three internally-node disjoint paths sharing the same node $x$ such that $ \adm(P_1) \leq \adm(P_2) \leq \adm(P_3)$. Then (1) $\adm(P_1) \leq r(x) - 2\eta L$ and (2) if $\adm(P_1) \geq \adm(P_3) - 2\eta L$, then $\adm(P_3) \leq r(x)$. 
\end{observation}
\begin{proof} If $\adm(P_1) > r(x) - 2\eta L$, then $$\adm(P_1), \adm(P_2), \adm(P_3) > r(x) - 2\eta L,$$ contradicting the definition of branching radius; this implies (1). If  $\adm(P_1) \geq \adm(P_3) - 2\eta L$, then $\adm(P_3) \geq \adm(P_1), \adm(P_2), \adm(P_3) >  \adm(P_3) - 2\eta L$. Thus, by the definition of branching radius, $r(x) \geq \adm(P_3)$; this implies (2).\QED
\end{proof}

For each branching node $x$, we define $B_x = \{y \in V(T) : |xy| \leq \max\{r(x), \zeta L\}\}$ be a ball center at $x$ in the metric induced by the augmented distance in $T$.  We abuse the notation here by using $B_x$ to denote the subtree of $T$ induced by nodes in $B_x$. Let

\begin{equation}\label{eq:F}
	F = \bigcup_{x \mbox{ is $T$-branching}} B_x.
\end{equation}

Clearly, $F$ is forest whose trees are subtrees of $T$. We will construct a set of subtrees $\mathcal{F}$ that partition the vertex set of $F$ such that $\mathcal{F}$ has all properties claimed by Lemma~\ref{lm:tree-clustering}.

We say a $T$-branching node $x$ \emph{subsumed} by a $T$-branching node $y \not= x$ if
\begin{align}\label{eq:subsume}
	B_x & \subsetneq B_y\\
	\mbox{OR} \quad  B_x &=  B_y  \wedge (r(x) < r(y)) \label{eq:subsum-eq}
\end{align}

The subsumed relationship defines a partial order $\preceq$ on the set of branching nodes of $T$: $x\preceq y$ if $x$ is subsumed by $y$. 

We will apply our construction to each tree $X \subseteq F$ to get a set of subtrees $\mathcal{F}_X$ and our final set $\mathcal{F}$ will be:
\begin{equation}\label{eq:F-final}
	\mathcal{F} = \cup_{X\subseteq F} \mathcal{F}_X
\end{equation} 

\paragraph{Construction of $\mathcal{F}_X$.~}
Let $W$ be the set of all maximal elements in the partial order $\preceq$ defined by the subsuming/subsumed relationship restricted to branching nodes of $X$.  Initially, we mark every node in $X$  \emph{ungrouped}.  There are four steps in our construction. In Step 1 and Step 2, we guarantee that every branching node $x$ with $r(x) \geq \zeta L_i$ is grouped into a subtree in $\mathcal{F}_X$. In Step 3, we deal with branching node with $r(x)  < \zeta L_i$. In this step, we only consider $x$ such that there exists three internally node-disjoint paths of augmented length roughly $r(x)$. The remaining nodes are handled in Step 4, where we merge each  tree of ungrouped nodes (obtained from $X$ by removing grouped nodes)  to an \emph{adjacent tree} in $\mathcal{F}_X$; two (node-disjoint) trees are adjacent if  there is an edge between them. To guarantee that the diameter is bounded by $95\gamma L$,\footnote{We do not try to optimize the constant $95$ here.} for each  tree of ungrouped nodes, we must show that there is an adjacent tree in $\mathcal{F}_X$ of roughly the same diameter; this is the most technical part of our construction.

\paragraph{Step 1:}  Let $x \in W$ be a node  such that $r(x) \geq \gamma L$, if $B_x$ contains no grouped node, we form a new tree $B_x$, add it to $\mathcal{F}_X$ and mark every node of  $B_x$ as \emph{grouped}. We repeat this step until it no longer applies. 

\paragraph{Step 2:}  Let $x \in \mathcal{W}$ such that $r(x) \geq \gamma L$. Then, there must be at least one node in $B_x$ that is marked grouped in Step 1. Let $y$ be the grouped node closest (in augmented distance) to $x$. Ties are broken by the lexicographic order. We include every (ungrouped) node of $X[x, y)$ to the tree in $\mathcal{F}_X$ containing $y$. We then mark every node of  $X[x,y)$ as grouped. (Some nodes of $B_x$ may remain ungrouped.)

The tie-breaking rule in Step 2 guarantees that if $X[x, y) \cap X[x', y') \not=\emptyset$ and $y$ and $y'$ are nodes closest to $x$ and $x'$, respectively, then $y = y'$ and hence, every node of $X[x, y) \cup X[x', y')$ are grouped into the same tree in Step 1.

We remove all grouped nodes in $W$, sort remaining nodes in $\mathcal{W}$ by the non-increasing order of $r(x)$, and proceed to Step 3. By the construction in Step 1 and Step 2, every node $x \in W$ has $r(x) < \gamma L$.

\paragraph{Step 3:}  For each node $x$  in the sorted order in $W$, if there are three internally disjoint path $P_1,P_2,P_3$ starting from $x$ such that (a) $r(x) - 2\eta L \leq \adm(P_1),\adm(P_2), \adm(P_3) \leq r(x)$ and (b) there is no grouped node in $P_1\cup P_2 \cup P_3$, we then choose $P_1,P_2,P_3$ of maximal augmented length subject to (a) and (b);  add the tree $T_x = P_1\cup P_2\cup P_3$ to $\mathcal{F}_X$; and mark every node of $T_x$ as grouped. We call $x$ the \emph{core} node of $T_x$. 

For clarity, let $\mathcal{F}^i_X$ be $\mathcal{F}_X$ after Step $i$, $1\leq i\leq 3$.  We show the following structure of remaining ungrouped nodes of $X$.

\begin{lemma} \label{lm:tree-adj-tree-edm} Let $Y$ be the forest obtained by removing all grouped nodes in $X$. Let $T'$ be any tree in $Y$ and $\adj(T')$ be the tree of maximum augmented diameter in $\mathcal{F}^3_X$ adjacent to $T'$. Then $$\adm(T') \leq 8\adm(\adj(T')).$$
\end{lemma}
\begin{proof}
	Let $D$ be the diameter path of $T'$; $D$ is the path realizing the augmented diameter of $T'$. We say a ball $B_x$ \emph{cut} $D$ if $D \cap B_x \not= \emptyset$. We say that $B_x$ \emph{internally cuts} $D$ if it cuts $D$ and none of $D$'s endpoints belongs to the ball.  We say that $B_x$ is \emph{$D$-maximal} if there is no other ball $B_y$ such that $B_x\cap D \subsetneq B_y\cap D$.  (Note that $B_x\cap D$ is a single subpath of $D$ since $T$ is a tree.) 
	
	\begin{wrapfigure}{r}{0.30\textwidth}
		\vspace{-20pt}
		\begin{center}
			\includegraphics[width=0.30\textwidth]{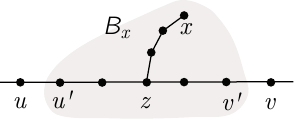}
		\end{center}
		\vspace{-10pt}
	\end{wrapfigure}

	\begin{claim} \label{clm:D-maximal-cut} Any  $D$-maximal ball $B_x$  must contain at least one endpoint of $D$.
	\end{claim}
	\begin{proof}
		Suppose for contradiction that there exists a $D$-maximal ball $B_x$ that internally cuts $D$; by definition $B_x$ does not contain any endpoint of $D$. We can assume that $x \in W$ since otherwise, we can just choose a node $y$ subsuming $x$ in $W$; by the definition of subsumed, $B_x \subseteq B_y$.

		Let $u,v$ be such that $D(u,v) \subseteq B_x$ and $D[u,v]$ is maximal. Let $u'$ and $v'$ be neighbor of $u$ and $v$, respectively, on $D[u,v]$.  Since $u,v \not \in B_x$, there must be a branching node $z$ in $D(u,v)$ where $X[z, x]$ is internally disjoint from $D[u,v]$.  (It is possible that $x = z$.)
		
		\begin{observation}\label{obs:zu-zv}$|zu'| \leq |zv'| + 2\eta L$ and $|zv'| \leq |zu'| + 2\eta L$.
		\end{observation}
		\begin{proof}
			Let $r  = \min(r(x),\zeta L_i)$ and $k = |xz| - w(z)$. Note that $|xu'| \leq r$ and $|zu'| ~=~ |xu'| - (|xz| - w(z)) ~\leq~  r - k$ and that:
			\begin{equation*}
				|xv'| + 2\eta L \geq |xv| \geq r,
			\end{equation*}
		which implies $|zv'| = |xv'| - (|xz| - w(z)) \geq r - k - 2\eta L \geq |zu'| - 2\eta L$. Symmetrically, it holds that $|zu'|\geq |zv'| - 2\eta L$. \QED
		\end{proof}
	We now continue the proof of Claim~\ref{clm:D-maximal-cut}. We consider two cases:
	
		\vspace{0.25cm}
		\noindent\textbf{Case 1: $x = z$.~} Note that $r(x) < \zeta L$ since otherwise, it was grouped in Steps 1 or 2, and hence $x$ is not present in $T'$. Thus, $|xu|, |xv| > r(x)$ since $u$ and $v$ are not in $B_x$. 		Since $x$ is branching, there must be a node $w$ such that $X[x,w]$ is internally disjoint from $X[x,u]$ and $X[x,v]$ and $|xw|\geq r(x)- 2\eta L$. If there is no node in $X[x,w]$ that is marked grouped, then there are three internally node-disjoint paths $X[x,u'], X[x,v'], X[x,w]$ of augmented length  in $[r(x)-2\eta L, r(x)]$ such that no path contains a grouped node; contradicting the construction in Step 3.
		
			\begin{wrapfigure}{r}{0.30\textwidth}
			\vspace{-20pt}
			\begin{center}
				\includegraphics[width=0.30\textwidth]{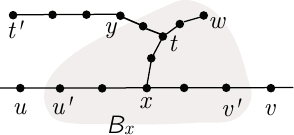}
			\end{center}
			\vspace{-10pt}
		\end{wrapfigure}
		
		Thus, a node, say $t\in X[x,w]$ is grouped to $T_y$, a tree with core  node $y$ in $\mathcal{F}$; we can assume $r(y) \geq r(x)$ since we process nodes in $W$ by the decreasing order of branching radius, and that $t$ is the first node on $X[x,w]$ when walking from $y$ to $x$. Since $y$ is branching, there must be node $t' \in T_y$ such that $X[y,t']$ is internally disjoint from $X[y,t]$, and $|yt'| \geq r(y) - 2\eta L$. This implies that:
		\begin{equation*}
			|xt'| > |tt'| \geq |yt'| \geq r(y) - 2\eta L \geq r(x) - 2\eta L,
		\end{equation*}
		where the first inequality is due to $x\not= t$ (while it could be that $y=t$). Thus, three paths $X[x,u], X[x,v], X[x,t']$ are internally node-disjoint and have minimum length strictly larger than $r(x) - 2\eta L$, contradicting Item (1) in Observation~\ref{obs:branching-rad}.

		\vspace{0.25cm}
		\noindent\textbf{Case 2: $x \not= z$.~} We first observe that:
		\begin{observation}\label{obs:rz-vs-rx} $r(z) > r(x)$.
		\end{observation}
		\begin{proof}
			Assume otherwise, that $r(z) \leq r(x)$. Let $w$ be a node such that $X[x,z]$ and $X[x,w]$ are internally disjoint and $|xw| \geq r(x) - 2\eta L$; $w$ exists since $x$ is a branching node. Observe that:
			
			\begin{equation*}
				|zw| > |xw| \geq r(x) - 2\eta L \geq r(z) - 2\eta L
			\end{equation*}

			Observe that $\max\{|zu'|, |zv'|\} \leq r(z)$ since otherwise, say $|zu'| > r(z)$, and hence, by Observation~\ref{obs:zu-zv}, $|zv'| \geq |zu'| -2\eta L > r(z) - 2\eta L $. That is, three paths $X[z,w], X[z,u'], X[z,v']$ have minimum augmented length strictly larger than $r(z)- 2\eta L$, contradicting Item (1) of Observation~\ref{obs:branching-rad}.
			
			\begin{wrapfigure}{r}{0.30\textwidth}
				\vspace{-25pt}
				\begin{center}
					\includegraphics[width=0.30\textwidth]{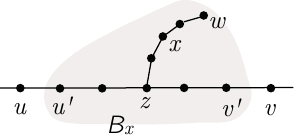}
				\end{center}
				\vspace{-20pt}
			\end{wrapfigure}

			Additionally, $\min(|zu|, |zv|) \leq r(z)$ since otherwise, three paths $X[z,w], X[z,u], X[z,v]$ have minimum augmented length strictly larger than $r(z) $, contradicting Item (1) of Observation~\ref{obs:branching-rad}.~W.l.o.g, we assume that $|zu| \leq r(z)$ and hence, $D[u,v'] \in B_z$, contradicting the $D$-maximality of $B_x$.\QED
		\end{proof}
		Next, we observe that:
		\begin{observation}\label{obs:Bx-Bz} $B_x\subseteq B_z$.
		\end{observation}
		\begin{proof}
			Suppose otherwise; let $t$ be any node in $B_x$ such that $|zt| > r(z)$. By Observation~\ref{obs:rz-vs-rx}, $z\not\in X[x,t]$. Thus, $X[z,t]$ is internally node-disjoint from $X[z,u], X[z,v]$. By Item (1) of Observation~\ref{obs:branching-rad}, $\min(|zu|, |zv|) \leq r(z) - 2\eta L < r(z)$. W.l.o.g, assume that $|zu| < r(z)$ and since $|zv'|\leq r(x) < r(z)$ (again by Observation~\ref{obs:rz-vs-rx}), $D[u,v'] \in B_z$, contradicting the $D$-maximality of $B_x$. \QED
		\end{proof}
	 We now complete the proof of Claim~\ref{clm:D-maximal-cut}.	By Observation~\ref{obs:rz-vs-rx} and Observation~\ref{obs:Bx-Bz}, we conclude that $z$ subsumes $x$ and hence, $x$ is not in $W$. This is a contradiction as we assumed earlier that $x \in W$. \QED	
	\end{proof}

	We are now continuing the proof of Lemma~\ref{lm:tree-adj-tree-edm}. Let $x\in W$ be such that $B_x$ maximally cuts $D$; there must be such a ball since by the definition of $F$, every node of $F$ is contained in some ball centered at a node in $W$. Let $D[u,v]$ be the maximal subpath of $D$ that belongs to $B_x$. Clearly, $|uv|\leq 2\max\{r(x),\zeta L_i\}\leq 2\gamma L$ since both $u,v$ are in $B_x$. The following claim is the key in showing that $\adm(T') \leq 8\adm(\adj(T'))$.

	\begin{claim} \label{clm:Duv-length}
		$|uv| \leq  4 \adm(\adj(T'))$.
	\end{claim}
	\begin{proof}
		Let $z$ be the first node on $D[u,v]$ when we walk from $x$ to $u$.~W.l.o.g, we assume that $|zu| \geq |zv|$. We consider two cases:
		
			\begin{wrapfigure}{r}{0.30\textwidth}
			\vspace{-20pt}
			\begin{center}
				\includegraphics[width=0.30\textwidth]{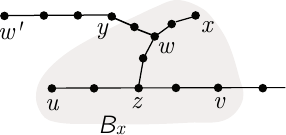}
			\end{center}
			\vspace{-10pt}
		\end{wrapfigure}
		\vspace{0.25cm}
		\noindent\textbf{Case 1: there is a node $w$ in $X[x,z]$ that is marked as grouped.~} Let $y \in W$ be such that $w\in B_y$. If $r(y)\geq\gamma L$, the claim holds since $|uv| \leq 2\gamma L$. Thus, we can assume that $r(y) < \gamma L$. If $|zu| > 2r(y)$, then for any $w' \in B_y$,
		\begin{equation*}
			\begin{split}
			|xw'|&\leq |xw|  + |wy|  + |yw'| \leq |xw| + 2r(y) \\ &< |xw| + |zu| \leq |xu| \leq r(x),
			\end{split}
		\end{equation*}
		which implies $B_y \subseteq B_x$ and $r(y) < r(x)$. This means $y$ is subsumed by $x$, contradicting that $y\in W$. We conclude that $|zu| \leq 2r(y)$. Since $|zv|\leq |zu|$, it holds that $|uv| \leq 4 r(y) \leq 4 \adm(\adj(T'))$; the claim follows.
		
		 \vspace{0.25cm}
		 \noindent\textbf{Case 2: every node in $X[x,z]$ that is ungrouped.~} This implies $x \in T'$, and hence $r(x) < \gamma L$ by the construction of Steps 1-2.  Since in Step 3, we added trees to $\mathcal{F}_X$ in the decreasing order of branching radius, there must be a node $y$ such that $r(y)\geq r(x)$ and $w \in B    _x\cap T_y$ is marked grouped. Assume that $w$ is closest to $x$ among nodes in $B_x\cap T_y$.
		 
		 If no node in $X[x,w)$ is marked grouped, then $r(x) \leq r(y) \leq \adm(\adj(T'))$. Since $|uv|\leq 2r(x)$, the claim holds.

		 Otherwise, there exists a node $a \in X[x,w]$ that is marked grouped, and that $a\in T_b$ for some node $b\in W$. W.l.o.g, we assume that $a$ is the first node on the path from $b$ to $x$ on $X[x,w]$. We only need to consider the case $|zu| > 2r(b)$, since otherwise, the claim holds. This implies that $r(x)\geq |zu| > 2r(b)$.
		 
		 If $a = b$, let $t_1,t_2$ be such that $X[t_1,t_2] = B_b \cap X[x,w]$ (see Figure~\ref{fig:case2}(a)). Observe that $x \not \in X[t_1,t_2]$ since $x$ is ungrouped and $w \not\in X[t_1,t_2]$ since $w \in B_y$. Thus,  $X[t_1,t_2]\subseteq X(x,w)$, and hence for any $t_3 \in B_b$, $|bt_3| \leq |bw| < r(x)$. This implies that $|xt_3| \leq |xw| \leq r(x)$. Hence, $B_b\subseteq B_x$ while $r(x) > r(b)$. This means $b$ is subsumed by $x$, contradicting that $b \in W$. 
		 
		 \begin{figure}[!htb]
		 	\center{\includegraphics[width=0.7\textwidth]
		 		{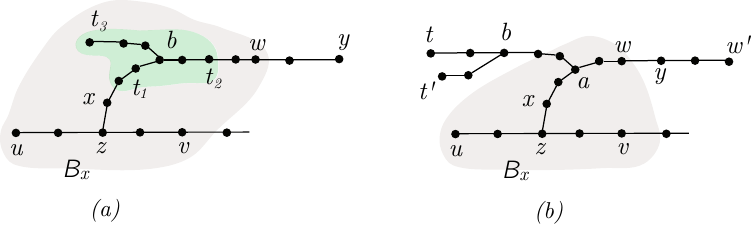}}
		 	\caption{(a) The case $a = b$, and (b) the case $a\not= b$.}
		 	\label{fig:case2}
		 \end{figure}

		 % 	\begin{wrapfigure}{r}{0.30\textwidth}
		  %	\vspace{-20pt}
		 % 	\begin{center}
		  %		\includegraphics[width=0.30\textwidth]{max-ball-cut4}
		 % 	\end{center}
		  %	\vspace{-20pt}
		 % \end{wrapfigure}
	  
		 We now assume that $a\not = b$ (see Figure~\ref{fig:case2}(b)). Let $t$ be such that $X[b,t]$ is internally disjoint from $X[b,a]$ and $|bt|\geq r(b) - 2\eta L$. Then, $|at| > r(b) - 2\eta L$. Let $w'$ be such that $X[y,w']$ is internally disjoint from $X[y,w]$ and $|yw'| \geq r(y) - 2\eta L$; $w'$ exists by the definition of $r(y)$ and the fact that $y$ is the core of $T_y$. Then $|ww'| \geq r(y) - 2\eta L > 2r(b)-2\eta L$.  Since $|au| > |zu| \geq 2r(b)$, $X[a,t], X[a,u], X[a,w']$ are three internally node-disjoint path of augmented length strictly larger than $r(b) - 2\eta L$. Thus, $r(a) > r(b)$ by Item (1) of Observation~\ref{obs:branching-rad}.  The following observation, with the fact that $r(a) > r(b)$, implies that $b$ is subsumed by $a$, contradicting that $b\in W$.
		 
		 \begin{observation}\label{obs:Bb-Ba} $B_b\subseteq B_a$.
		 \end{observation}
		 \begin{proof}
		 	Let $t'$ be any node in $B_b\setminus B_a$. Since $r(a) > r(b)$, $a \not\in X[b,t']$. Thus, $X[a,t']$ are internally disjoint from $X[a,u]$ and $X[a,w']$. Since $|at'| \leq 2r(b)$ and $|au|, |aw'| > 2r(b) - 2\eta L$, $r(a) > |at'|$ by the definition of branching radius. But that implies $t' \in B_a$, a contradiction.\QED
		 \end{proof}
		 As we noted above, Observation~\ref{obs:Bb-Ba} completes the proof of Claim~\ref{clm:Duv-length}. \QED
	\end{proof}
	
	We are now finishing the proof of Lemma~\ref{lm:tree-adj-tree-edm}. By Claim~\ref{clm:D-maximal-cut},  there are at most two $D$-maximal balls $B_x, B_y$, each contains one endpoint of $D$  and that $D\subseteq B_x\cup B_y$ because every node in $\mathcal{D}$ must belong to some ball by the definition of $F$. Thus, Lemma~\ref{lm:tree-adj-tree-edm} follows directly from Claim~\ref{clm:Duv-length}.\QED
\end{proof}

Equipped with Lemma~\ref{lm:tree-adj-tree-edm}, in the last step of the construction, we simply augment each remaining tree (of ungrouped nodes) to an adjacent tree of grouped nodes that have maximum augmented diameter. 

\paragraph{Step 4} Let $Y$ be the forest obtained from $X$ by removing all grouped nodes after Step 3. For each tree $T'$ in $Y$, we augment $\adj(T')$ by adding $T'$ and the edge connecting $\adj(T')$ and $T$ to $\adj(T')$.

Step 4 completes the construction of tree clustering. We now show that all claims in Lemma~\ref{lm:tree-clustering} hold; we restate the lemma below.

\treeclustering*
\begin{proof}
 Let $T_x$ be a tree in $F_X$ with core node $x$; the core node of a tree may not be the center as described in Item(3) of Lemma~\ref{lm:tree-clustering}. Observe that Item (2) in Lemma~\ref{lm:tree-clustering} follows directly from the construction, specifically, from the definition of $F$ (Equation~\ref{eq:F}).

\vspace{0.25cm}
\noindent\textbf{Proof of Item (1).~} Let $T^{i}_x$ be $T_x$ after step $i$ for each $i \in \{1,2,3,4\}$; $T^4_x = T_x$. If $T_x$ is formed in Step 3, then $T^1_x = T^2_x = \emptyset$; otherwise, $T^2_x = T^3_x$. If $T^1_x \not= \emptyset$, then $\adm(T^1_x)\geq \gamma L - 2\eta L \geq \gamma L/2$ since $\eta \ll \gamma$. The augmentation in Step 2 increases the augmented diameter of $T^1_x$ by at most $2\gamma L$. This implies $\adm(T^2_x) \leq 5\adm(T^1_x)$. 

Let $k_e$ be the upper bound on the weight of edges in $T_x$. Clearly, $k_e \leq \adm(T^3_x)$ since the weight of every edge is at most the weight of (any) node. By Lemma~\ref{lm:tree-adj-tree-edm}, the augmentation to $T^3_x$ is by a star-like way via edges of $X$, it holds that $\adm(T^4_x)\leq \adm(T^3_x) + 2k_e + 16\adm(T^3_x) \leq 19 \adm(T^3_x)$. Thus, $\adm(T^4_x) \leq 38\gamma L$ if $T_x$ is formed in Step 3; otherwise, $\adm(T^4_x) \leq 19 \adm(T^2_x) \leq 95 \adm(T^1_x) \leq 190 \gamma L$.

\vspace{0.25cm}
\noindent\textbf{Proof of Item (3).~} Let $D$ be the diameter path of $T_x$. Since $x$ is branching, one of the three paths from $x$ to a node, say $t$, of augmented length $|xt|$ is internally node-disjoint from $D$, and that the path from $t$ to a node in $D$ must go through $x$. Let $b$ the node closest to $x$ on $D$; it is possible that $x = b$. Two paths $P_1,P_2$ in Item (3) of Lemma~\ref{lm:tree-clustering} are subpaths of $D$ sharing the same endpoint $b$ such that $P_1\cup P_2 = D$. Let $P_3$ be $T_x[b,t]$. 

	\begin{wrapfigure}{r}{0.30\textwidth}
	\vspace{-20pt}
	\begin{center}
		\includegraphics[width=0.30\textwidth]{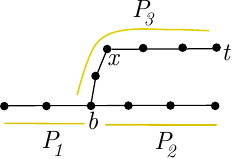}
	\end{center}
	\vspace{-20pt}
\end{wrapfigure}

If $r(x) < 8\eta L$, then $\adm(P_3\setminus \{b\}) \geq w(t) \geq \frac{\eta L}{\beta}$. From the proof of Item (1), we deduce that $\adm(T_x) = \adm(T^4_x) \leq 19\adm(T^3_x) \leq 38 r(x) = O(\eta L)$; note that $T^1_x = T^2_x = \emptyset$ in this case. Thus, $\adm(P_3\setminus \{b\}) = \Omega(\adm(T_x)/\beta)$. 

Otherwise; $r(x)\geq 8\eta L$. Observe that $|bt|\geq |xt| \geq (r(x) - 2\eta L)$. Thus, we have:

\begin{equation*}
	\adm(P_3\setminus \{b\}) \geq (r(x) - 2\eta L) - 2\eta L \geq r(x)/2,
\end{equation*}
since $r(x)\geq 8\eta L$.  From the proof of Item (1), we deduce that $\adm(T_x) = \adm(T^4_x) \leq 19\adm(T^3_x) \leq 38 r(x)$ if $T_x$ is formed in Step 3, and that $\adm(T_x) \leq 95\adm(T^1_x) \leq 190 r(x)$ if $T_x$ is formed in Step 1. Thus, $\adm(T_x) = O(r(x))$ and hence, $\adm(P_3\setminus \{b\}) = \Omega(\adm(T_x)) = \Omega(\adm(T_x)/\beta)$ since $\beta \geq 1$.  

\vspace{0.25cm}
\noindent\textbf{Proof of Item (4).~} If $T_x$ is formed in Step 1, then $\adm(T_x) \geq 2(\gamma L - 2\eta L) -\eta L\geq \gamma L$ since $\eta \ll \gamma$. Item (4) then follows from the following claim.

\begin{claim} \label{clm:deg2-zoomout}
	If $T_x$ is formed in Step 3, then its corresponding node has degree $2$ in $\widehat{T}$.
\end{claim}
\begin{proof}
Suppose otherwise; there exist three nodes $x_1,x_2,x_3$ incident to $T_x$ in $T$. We consider two cases:

\vspace{0.25cm}
\noindent\textbf{Case 1: $T[x,x_1], T[x,x_2], T[x,x_3]$ are pairwise internally disjoint.~}  Let $T_i$ be the subtree of $T\setminus \{x\}$ containing $x_i$, $i \in [1,3]$. Let $y_i\in T_i$ be the furthest node (in augmented distance) from $x$. W.l.o.g, we assume that $|xy_1| \leq \min(|xy_2|, |xy_3|)$. Then $r(x) > |xy_1|$ by Observation~\ref{obs:branching-rad}, and hence $T_1\subseteq B_x$. Since $T_1\subseteq F$ (in Equation~\eqref{eq:F}) and $x_1\not\in T_x$, there must be $y\in W$ such that $x_1\in T_y$. Since $T_1\subseteq B_x$, we conclude that $B_y\subseteq B_x$ and that $r(x) > r(y)$. That implies $y$ is subsumed by $x$, contradicting that $y\in W$.

\vspace{0.25cm}
\noindent\textbf{Case 2: $T[x,x_1], T[x,x_2], T[x,x_3]$ are not pairwise internally disjoint.~} W.l.o.g, we assume that $T[x,x_1]$ and $T[x,x_2]$ shares a node $z\not=x$. We choose $z$ to be the furthest node from $x$ (in augmented distance). For notational convenience, let $x_0 = x$. Clearly $T[z,x_0], T[z,x_1], T[z,x_2]$ are pairwise internally disjoint. Let $T_i$ be the subtree of $T\setminus \{z\}$ containing $x_i$, $i \in [0,2]$. Let $y_i\in T_i$ be the furthest node (in augmented distance) from $z$. W.l.o.g, we assume that $|zy_1| \leq |zy_2|$.

\begin{wrapfigure}{r}{0.30\textwidth}
	\vspace{-25pt}
	\begin{center}
		\includegraphics[width=0.30\textwidth]{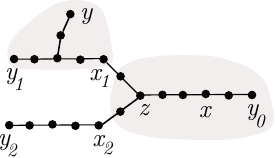}
	\end{center}
	\vspace{-25pt}
\end{wrapfigure}

If $|zy_0| \geq |zy_1|$, then $r(z) > |zy_1|$ and hence $T_1\subseteq B_z$. This implies $T_1\subseteq F$, and since $x_1\not\in T_x$, there must exist $y\in W$ such   that  $x_1 \in B_y$. By the same argument in Case 1, we deduce that $B_y\subseteq B_z$ and $r(z) > r(y)$. This means $y$ is subsumed by $z$, contradicting that $y\in W$.

Thus, we can assume that $|zy_0|\leq |zy_1|$ and hence $r(z) > |zy_0|$. This implies $T_0\subseteq B_z$ and hence, $B_x\subseteq  B_z$ and $r(z) > r(x)$. That is, $x$ is subsumed by $z$, contradicting that $x\in W$. \QED
\end{proof}
Claim~\ref{clm:deg2-zoomout} completes the proof of Lemma~\ref{lm:tree-clustering}. \QED
\end{proof}

\paragraph{Acknowledgements.}
The second-named author is grateful to Michael Elkin, Ofer Neiman and Michiel Smid for fruitful discussions.
Both authors are indebted to the anonymous referees for their thorough and helpful comments, 
which greatly helped us to improve the presentation of the paper.

\bibliographystyle{plain}
\bibliography{spanner}

\pagebreak
\appendix
\section{A Simple Upper Bound on Spherical Code}\label{app:Code}

\begin{lemma}\label{lm:spherical-code-up} Let $C$ be a $(d,\theta)$-spherical code. Then $|C| = O_d((\theta)^{d-1})$ when $\theta \leq 1/2$.
\end{lemma}
\begin{proof} We use the volume argument to derive an upper bound on $C$. By construction, the surface of the hyperspherical cap, say $C_{d,\theta}$, centered at $x$ of with cap angle $\theta$ does not contain any point of $C$. The surface area (see Li~\cite{Li11}) of $C_{d,\theta}$, denoted by $A_d(C_{d,\theta})$, is:
	\begin{equation*}
		A_d(C_{d,\theta}) = \frac{1}{2}A_d(1)I_{\sin^2\theta}((d-1)/2,1/2)
	\end{equation*}
	where $A_d(1)$ is the surface area of $\mathbb{S}_{d}$ and $I_{x}(a,b)$ is the regularized incomplete beta function. By definition, $I_{x}(a,b) = \frac{B_x((d-1)/2,1/2)}{B((d-1)/2,1/2)} $ where $B(a,b)$ ($B_x(a,b))$) is the (incomplete) beta function. By taking Taylor expansion at $0$, $B_x(a,1/2) = x^a(1/a + \frac{1/2x}{a} + O(x^2)/a) \geq x^{a}/a$. Thus, we have
	\begin{equation*}
		A_d(C_{d,\theta}) = \frac{A_d(1)}{B((d-1)/2,1/2)}\frac{(\sin\theta)^{d-1}}{d}
		\geq \frac{A_d(1)}{B((d-1)/2,1/2)}\frac{\theta^{d-1}}{d2^d}~.
	\end{equation*}
Here we use $\sin(\theta) \geq \theta/2$ when $\theta \leq 1/2$.  Since the total surface area of the unit sphere is $A_d(1)$, we have:
	\begin{equation*}
		|C| \leq \frac{A_d(1)}{A_d(C_{d,\theta})} = d 2^d B((d-1)/2,1/2)(\theta)^{d-1} = O_d((\theta)^{d-1}), 
	\end{equation*}
as desired. \QED
\end{proof}

\end{document}